\pdfoutput=1

\documentclass[11pt]{article}
\usepackage{amsmath,amssymb,amsfonts,latexsym,graphicx,amsthm}
\usepackage{fullpage,color}
\usepackage{url,hyperref}
\usepackage{comment}
\usepackage[linesnumbered,boxed,ruled,noend]{algorithm2e}
\usepackage{framed}
\usepackage[shortlabels, inline]{enumitem}
\usepackage{cleveref}
\usepackage{setspace}
\usepackage{longtable}
\usepackage{thmtools}
\usepackage{thm-restate}
\usepackage{xcolor}
\usepackage[final, multiuser]{fixme}
\fxusetheme{color}
\fxuselayouts{inline, nomargin}
\FXRegisterAuthor{ab}{aab}{Anton}
\FXRegisterAuthor{ss}{ass}{Shay}
\FXRegisterAuthor{tz}{atz}{Tianyi}

\usepackage{tikz}
\usetikzlibrary{backgrounds}
\usetikzlibrary{decorations.pathreplacing,calligraphy}
\usetikzlibrary{decorations.pathmorphing}
\usetikzlibrary{fit}
\usetikzlibrary{patterns}
\usepackage{subfigure}
\usepackage{pgfplots}
\usetikzlibrary{shapes.misc}
\tikzset{cross/.style={cross out, draw=black, fill=none, minimum size=2*(#1-\pgflinewidth), inner sep=0pt, outer sep=0pt}, cross/.default={10pt}}
\usepackage[margin=1in]{geometry}

\usepackage[margin=1in]{geometry}

\newtheorem{theorem}{Theorem}[section]
\newtheorem{lemma}{Lemma}[section]
\newtheorem{definition}{Definition}[section]
\newtheorem{corollary}{Corollary}[section]
\crefname{corollary}{corollary}{corollaries}
\Crefname{corollary}{Corollary}{Corollaries}
\newtheorem{claim}{Claim}[section]
\crefname{claim}{claim}{claims}
\Crefname{claim}{Claim}{Claims}
\newtheorem{question}{Question}[section]
\newtheorem{invariant}{Invariant}[section]
\crefname{invariant}{invariant}{invariants}
\Crefname{invariant}{Invariant}{Invariants}
\newtheorem{observation}{Observation}[section]
\crefname{observation}{observation}{observations}
\Crefname{bservation}{Observation}{Observations}

\newtheorem{property}{Property}[section]
\crefname{property}{property}{properties}
\Crefname{property}{Property}{Properties}

        {\medskip}

\newcommand{\eps}{\epsilon}

\newcommand{\opt}{\mathsf{OPT}}

\newcommand{\ceil}[1]{\left\lceil #1 \right\rceil}
\newcommand{\floor}[1]{\left\lfloor #1 \right\rfloor}

\newcommand{\set}{\mathcal{S}}

\newcommand{\univ}{\mathcal{U}}
\newcommand{\wts}{\omega}
\newcommand{\lev}{\mathsf{lev}}
\newcommand{\ilev}{\mathsf{ilev}}
\newcommand{\zlev}{\mathsf{zlev}}
\newcommand{\bs}{\mathsf{base}}
\newcommand{\del}{\mathsf{Delete}}
\newcommand{\ins}{\mathsf{Insert}}
\newcommand{\fix}{\mathsf{FixLevel}}
\newcommand{\reset}{\mathsf{Rebuild}}
\newcommand{\water}{\mathsf{WaterFilling}}
\newcommand{\trunc}{\mathsf{Truncate}}
\newcommand{\procdet}{\mathsf{HandleDet}}
\newcommand{\procrand}{\mathsf{HandleRand}}

\newcommand{\decrease}{\mathsf{DecILev}}
\newcommand{\up}{\text{up}}
\newcommand{\down}{\text{down}}
\newcommand{\lift}{\text{lift}}
\newcommand{\clean}{\text{clean}}
\newcommand{\old}{\text{old}}
\newcommand{\nw}{\text{new}}
\newcommand{\randtime}{\lambda_\text{rand}}
\newcommand{\dettime}{\lambda_\text{det}}
\newcommand{\aux}{\mathsf{aux}}
\newcommand{\tm}{\mathsf{tm}}
\newcommand{\clog}{(5\log)}

\newcommand{\brac}[1]{\left(#1\right)}

\begin{document}
\begin{titlepage}
	\title{Nearly Optimal Dynamic Set Cover: \\Breaking the Quadratic-in-$f$ Time Barrier}
	
		\author{Anton Bukov \thanks{Tel Aviv University, \href{}{bukov.anton@gmail.com}}\and Shay Solomon \thanks{Tel Aviv University, \href{}{shayso@tauex.tau.ac.il}}\and Tianyi Zhang \thanks{Tel Aviv University, \href{}{tianyiz21@tauex.tau.ac.il}}}
	
	\date{}
	
	\maketitle
	\thispagestyle{empty}
	
	\begin{abstract}
	The dynamic set cover problem has been subject to extensive research since the pioneering works of [BHI, ICALP'15] and [GKKP17, STOC'17]. The input is a set system $(\univ, \set)$ on a fixed collection $\set$ of sets and a dynamic universe of elements, where each element appears in a most $f$ sets and the cost of each set lies in the range $[1/C, 1]$;
 the ultimate goal is to maintain a set cover under insertions and deletions of elements, with optimal bounds on both the approximation factor and the update time. 

	Most previous works considers the low-frequency regime, namely $f = O(\log n)$, and this line of work has culminated with a deterministic $(1+\epsilon)f$-approximation algorithm with amortized update time $O(\frac{f^2}{\epsilon^3} + \frac{f}{\epsilon^2}\log C)$ [BHNW, SODA'21] and a randomized $f$-approximation algorithm against an oblivious adversary with expected amortized update time $O(f^2)$ for the unweighted case [AS, ESA'21]. In the high-frequency regime of $f = \Omega(\log n)$, an $O(\log n)$-approximation algorithm with amortized update time $O(f\log n)$ was given by [GKKP17, STOC'17], and recently [SU, STOC'23] showed that the same update time of $O(f \log n)$ suffices for achieving approximation $(1+\epsilon)\ln n$.

	Interestingly, at the intersection of the two regimes, i.e., $f = \Theta(\log n)$, the state-of-the-art results coincide (ignoring the dependencies on $\eps$ and $C$): approximation  $\Theta(f) = \Theta(\log n)$ with amortized update time $O(f^2) = O(f \log n) = O(\log^2 n)$. Up to this date, no previous work achieved update time of $o(f^2)$, even allowing randomization against an oblivious adversary and even for a worse approximation guarantee.

	In this paper we break the $\Omega(f^2)$ update time barrier via the following results:
		\begin{itemize}[leftmargin=*]
			\item $(1+\epsilon)f$-approximation can be maintained in $O\left(\frac{f}{\epsilon^2}\log^*f  + \frac{f}{\epsilon^3} + \frac{f}{\epsilon^2}\log C\right)
			= O_{\epsilon,C}(f \log^* f)$ expected amortized update time \footnote{$\log^*$ is the iterated logarithm;
			we use the notation $O_{\epsilon,C}(\cdot)$ to suppress factors that depend on $\eps$ and $C$.}; our algorithm works against an adaptive adversary.
			\item $(1+\epsilon)f$-approximation can be maintained deterministically in $O\left(\frac{1}{\epsilon}f\log f + \frac{f}{\epsilon^3} + \frac{f}{\epsilon^2}\log C\right) = O_{\epsilon,C}(f \log f)$  amortized update time.
	\end{itemize}
	Assuming element updates are specified {\em explicitly}, our randomized algorithm is near-optimal: $(1+\epsilon)f$ approximation is optimal up to the $\epsilon$-dependence and the update time $O_{\epsilon,C}(f \log^*f)$ exceeds the time needed to specify an update by a $\log^*f$ factor.
	We view this slack of $\log^* f$ factor as interesting in its own right --- we are not aware of any problem for which the state-of-the-art dynamic algorithm admits a slack of  $\log^*f = O(\log^* n)$ from optimality.
	\end{abstract}
\end{titlepage}

\begin{spacing}{1.3}
	\tableofcontents
\end{spacing}

\thispagestyle{empty}
\clearpage
\pagenumbering{arabic}
\setcounter{page}{1}

\newpage

\section{Introduction}
In the static set cover problem, we are given a set system $(\univ, \set)$, where
$\univ$ is a universe of $n$ elements and $\set$ is a collection of $m$ sets 
$s \in \set$ of elements in $\univ$, each of which associated with a cost $c_s\in [\frac{1}{C}, 1]$.
The {\em frequency} of the set system $(\univ, \set)$, denoted by $f = f(\univ, \set)$, is the maximum number of sets in $\set$ any element in $\univ$ belongs to. 
A collection of sets $\set^\prime\subseteq \set$ is called a {\em set cover} of $\univ$ if any element in $\univ$ belongs to at least one set in $\set^\prime$. The basic goal is to compute a {\em minimum set cover}, i.e., a set cover 
$\set^*\subseteq \set$ whose cost $c(\set^*) = \sum_{s\in \set^*}c_s$ is minimum.

The set cover problem is a central NP-hard problem, which admits two classic algorithms: a {\em greedy} $\ln n$-approximation and a primal-dual $f$-approximation.
Importantly, one cannot achieve approximation $(1-\eps) \ln n$ unless P = NP \cite{williamson2011design,dinur2014analytical}
as well as approximation $f-\eps$ for any fixed $f$ under the unique games conjecture \cite{khot2008vertex}.
The greedy and primal-dual approximation algorithms for set cover have been extremely well-studied in the static setting and are well-understood by now, and an extensive body of work from recent years aims at efficiently ``dynamizing'' these algorithms. 
In the {\em dynamic} setting of set cover, the goal is to maintain a set cover of low cost, while the universe $\univ$ evolves over time. 
More specifically, the goal is to maintain a set cover $T\subseteq \set$ of low cost while supporting two types of element updates:
\begin{itemize}[leftmargin=*]
	\item \textbf{Insertion.} A new element $e$ enters $\univ$, and the input specifies the sets in $\set$ that it belongs to.
	
	\item \textbf{Deletion.} An existing element in $\univ$ is deleted from all sets in $\set$ that it belonged to.
\end{itemize}

The two main quality measures of a dynamic algorithm are its approximation ratio $\frac{c(T)}{c(S^*)}$ and update time,
where the holy grail is to achieve approximation approaching that of the best {\em static} algorithm with as small as possible update time. 
In the context of set cover: (1) for approximation, given the aforementioned lower bounds, the goal would be either an 
$O(\log n)$ or $O(f)$ approximation, and (2) for update time, since it takes $\Theta(f)$ time to explicitly represent an element update 
(by specifying all the sets to which it belongs), the natural goal would be update time $O(f)$.

The dynamic set cover problem was first studied in \cite{bhattacharya2015design}, where a deterministic primal-dual algorithm with $O(f^2)$ approximation and $O(f \log(m+n))$ (amortized) update time was presented.\footnote{For brevity, in what follows we shall not make the distinction between amortized and worst-case update time.} 
Later on, a deterministic $O(\log n)$-approximation algorithm with $O(f\log n)$ update time
was given in \cite{gupta2017online}. 
This work of \cite{gupta2017online} essentially ``dynamizes'' the greedy algorithm in the high-frequency regime, namely $f = \Omega(\log n)$. In a recent work \cite{solomon2023dynamic}, the authors improved the approximation to $(1+\epsilon)\ln n$ with $O\brac{\frac{f\log n}{\epsilon^5}}$ amortized update time.

All other previous works, which we survey next, focus on the low-frequency regime of $f = O(\log n)$, and they all essentially dynamize  the primal-dual algorithm. A deterministic $O(f^3)$-approximation algorithm with $O(f^2)$   update time was achieved in \cite{gupta2017online,bhattacharya2017deterministic}.
The first $O(f)$ approximation was achieved in \cite{abboud2019dynamic}, where the authors proposed a randomized
 $(1+\epsilon)f$-approximation algorithm with   update time $O(\frac{f^2\log n}{\epsilon})$; this algorithm works for unweighted instances only (where $c_s \equiv 1$ for all $s$) and it assumes an oblivious adversary. 
This result was subsumed by \cite{bhattacharya2019new}, where a deterministic $(1+\epsilon)f$-approximation algorithm for weighted instances was presented, with update time of $O(\frac{f}{\epsilon^2}\log(Cn))$.

The works of \cite{abboud2019dynamic,bhattacharya2019new} with $(1+\epsilon)f$-approximation incur a slack of $\log n$ on the update time.
Two subsequent works remove the dependency on $\log n$: \cite{bhattacharya2021dynamic} gave a deterministic
$(1+\epsilon)f$-approximation algorithm with update time $O(\frac{f^2}{\epsilon^3} + \frac{f}{\epsilon^2}\log C)$, 
while \cite{assadi2021fully} gave a randomized $f$-approximation algorithm with update time $O(f^2)$, but it assumes an oblivious adversary, and it only applies to unweighted instances.

To summarize, in the low frequency regime of $f=O(\log n)$, no  previous work achieved update time of $o(f^2)$, even allowing randomization against an oblivious adversary and even for approximation larger than $O(f)$.
For the high frequency regime of $f = \Omega(\log n)$, the only previous work achieves update time $O(f \log n)$ \cite{gupta2017online};
interestingly, at the intersection of the two regimes, i.e., $f = \Theta(\log n)$, the state-of-the-art results coincide (ignoring the dependencies on $\eps$ and $C$): approximation  $\Theta(f) = \Theta(\log n)$ with amortized update time $O(f^2) = O(f \log n) = O(\log^2 n)$.

A fundamental question left open by previous works is whether one can break the quadratic-in-$f$ update time barrier, ideally to achieve an update time of $O_{\eps,C}(f)$ (ignoring the dependencies on $\eps$ and $C$), i.e., linear in the time needed to explicitly specify an update. 
\begin{question} \label{q1}
Is there $O(f)$-approximation (or $O(\log n)$-approximation) algorithm for set cover with update time $o(f^2)$?
Further, it is possible to achieve approximation approaching $f$ (or $\ln n$) with update time approaching $O(f)$?
\end{question}

\paragraph{Perspective: The Quest Towards Optimal Update Time.~}
The quest towards {\em constant update time} algorithms for basic graph problems is an important research agenda in the field of dynamic graph algorithms, which has attracted a lot of research attention over the past decade 
~\cite{PelegS16,Solomon16,bhattacharya2017deterministic,BhattacharyaGM17,gupta2017online,
SolomonW18,
bhattacharya2019deterministically, Henzinger020,
bhattacharya2021dynamic,assadi2021fully,BhattacharyaGKL22,BCPS23}.
This research agenda coincides with the quest towards linear-time graph algorithms in the static sequential setting, since any constant update time algorithm (that uses at most linear time during preprocessing) gives rise to a linear-time static algorithm.
Of course, not every dynamic graph problem admits a constant update time solution, even if the respective static problem admits a linear running time, and graph connectivity is a prime example \cite{patrascu2006logarithmic,puatracscu2011don}. 

The set cover problem is equivalent to the vertex cover problem in hypergraphs, where the {\em rank} of the hypergraph is the
frequency $f$ of the set-system. One can generalize any graph problem for hypergraphs, with one significant caveat:
In dynamic hypergraphs, the time needed to explicitly specify an edge update is no longer constant, but rather $O(f)$. One may consider implicit updates instead (switching an edge ``on'' and ``off''), which can be carried out in constant time, but even for implicit updates there are conditional lower bounds on the update time that are not far from $\Omega(f)$, albeit only for sufficiently high frequency; refer to \cite{abboud2019dynamic} for details. 
Nonetheless, even ignoring such conditional lower bounds,
the $O(f)$ time bound to explicitly specify an update in rank-$f$ hypergraphs seems the natural generalization of constant update time in simple graphs for hypergraphs, and is thus a  natural time barrier.

To the best of our knowledge, the previous work on the dynamic set cover problem provides the {\em first systematic study on any dynamic hypergraph problem}. 
Moreover, we are not aware of any nontrivial hypergraph problem that is solved within update time $O(f)$.
Consequently, whether  it is possible to fully resolve \Cref{q1} --- and obtain the {\em first update time of $O(f)$ for any rank-$f$ hypergraph problem} (with a reasonably good approximation) --- seems to be of major importance.

\subsection{Our result} \label{sec:our}
Our main result, which resolves \Cref{q1} in the affirmative, is summarized in the following theorem;
\Cref{comp} provides a concise comparison between our and previous results.

\begin{theorem}\label{main-result}
	For any set system $(\univ, \set)$ (with $\univ = \emptyset$ initially) that undergoes a sequence of element insertions and deletions, where the frequency is always bounded by $f$, and for any $\epsilon \in (0, 0.1)$, there are dynamic algorithms that maintain a $(1+\epsilon)f$-approximate minimum set cover with the following   amortized update time bounds.
	\begin{itemize}[leftmargin=*]
		\item Expected $O\brac{\frac{f}{\epsilon^2}\log^*f  + \frac{f}{\epsilon^3} + \frac{f}{\epsilon^2}\log C}$, via a randomized algorithm {\em against an adaptive adversary}.
		\item Deterministic $O\brac{\frac{1}{\epsilon}f\log f + \frac{f}{\epsilon^3} + \frac{f}{\epsilon^2}\log C}$.
	\end{itemize}
\end{theorem}

\noindent{\bf Remark.} For our deterministic algorithm, we shall assume for simplicity that the length of the update sequence is at least  
$\frac{1}{\epsilon}m\log(Cn)$. In this way, during preprocessing (before the first element is inserted to $\univ$), we   prepare a data structure of size $O(\frac{1}{\epsilon}m\log(Cn))$. (The same is done implicitly in previous works 
whose amortized update time is independent of $n$ \cite{bhattacharya2019deterministically,bhattacharya2021dynamic,assadi2021fully}.) 
In these algorithms (including ours), all elements $e\in \univ$ are assigned a level value $0\leq \lev(e)\leq \ceil{\log_{1+\epsilon}(Cn)}+1$, and for each set $s\in \set$, we maintain a list of all elements $E_i(s) = \{e\in s, \lev(e) = i \}$ (in our algorithm, sets $A_i(s)$ and $P_i(s)$, which are defined a bit differently). 
Since the pointer to each set $E_i(s)$ needs to be retrieved in $O(1)$ time given the index $i$, 
we maintain an array of length $O(\log_{1+\epsilon}(Cn))$ storing all the pointers, even if some sets $E_i(s)$ might be empty. 
(For our randomized algorithm, we can simply use dynamic hash tables \cite{dietzfelbinger1994dynamic}.)

We emphasize two points regarding our randomized algorithm. 
\begin{itemize}[leftmargin=*]
\item It works against an adaptive adversary; this is the first randomized algorithm 
for dynamic set cover that does not make the assumption of an oblivious adversary. 
\item Assuming element updates are specified {\em explicitly}, the update time $O_{\epsilon,C}(f \log^*f)$ exceeds the time needed to specify an update by a $\log^*f$ factor.
This slack of $\log^* f$ factor is interesting in its own right --- {\bf we are not aware of any problem for which the
state-of-the-art dynamic algorithm admits a slack of}   ${\log^*f = O(\log^* n)}$ {\bf from optimality}.
(A notable example where such a slack was studied is for the Disjoint-set data structure,
where a highly influential line of work improved the $O(\log^* n)$   bound to an inverse-Ackermann bound, later shown to be tight.) 
\end{itemize}

\begin{table}
	\begin{tabular}{|c|c|c|c|c|}
		\hline
		reference	&	approximation	&	update time	&	deterministic?	&	weighted?\\\hline
		\cite{gupta2017online}	&	$O(\log n)$	&	$O(f\log n)$	&	yes		&	yes	\\\hline
		\cite{solomon2023dynamic}	&	$(1+\epsilon)\ln n$	&	$O\brac{\frac{f\log n}{\epsilon^5}}$	&	yes	&	yes\\\hline\hline
		\cite{bhattacharya2015design}	&	$O(f^2)$	&	$O(f\log(m+n))$	&	yes	&	yes	\\\hline
		\cite{gupta2017online,bhattacharya2017deterministic}	&	$O(f^3)$	&	$O(f^2)$	&	yes	&	yes\\\hline
		\cite{abboud2019dynamic}	&	$(1+\epsilon)f$	&	$O\brac{\frac{f^2}{\epsilon}\log n}$	&	oblivious	&	no\\\hline
		\cite{bhattacharya2019new}	&	$(1+\epsilon)f$	&	$O\brac{\frac{f}{\epsilon^2}\log(Cn)}$	&	yes	&	yes	\\\hline
		\cite{bhattacharya2021dynamic}	&	$(1+\epsilon)f$	&	$O\brac{\frac{f^2}{\epsilon^3} + \frac{f}{\epsilon^2}\log C}$	&	yes	&	yes	\\\hline
		\cite{bhattacharya2021dynamic}	&	$(1+\epsilon)f$	&	$O\brac{f\log^2(Cn)/\epsilon^3}$ (wc)	&	yes	&	yes\\\hline
		\cite{assadi2021fully}	&	$f$	&	$O\brac{f^2}$	&	oblivious	&	no\\\hline \hline
		\textbf{new}	&	$(1+\epsilon)f$	&	$O\left(\frac{f}{\epsilon^2}\log^*f  + \frac{f}{\epsilon^3} + \frac{f}{\epsilon^2}\log C\right)$	&	adaptive	&	yes\\\hline
		\textbf{new}	&	$(1+\epsilon)f$	&	$O\left(\frac{1}{\epsilon}f\log f + \frac{f}{\epsilon^3} + \frac{f}{\epsilon^2}\log C\right)$	&	yes	&	yes	\\\hline
	\end{tabular}
	\caption{Summary of results on dynamic set cover. ``wc'' stands for ``worst case''.}\label{comp}
\end{table}

\subsection{Technical and Conceptual Contribution}
Our algorithm builds upon the {primal-dual} framework from \cite{bhattacharya2015design, bhattacharya2019deterministically,
bhattacharya2019new,bhattacharya2021dynamic}. In the primal-dual framework, all sets in $s\in\set$ are assigned to levels $\lev(s)$ numbered from $0$ to $L = \ceil{\log_{1+\epsilon}(Cn)}+1$. For each element $e\in\univ$, its level $\lev(e)$ is defined as the maximum level of any set it belongs to, namely $\lev(e) = \max_{s\ni e}\{\lev(s) \}$. This hierarchical partition of sets and elements into levels defines weights for elements and sets:
Each element $e$ is assigned a weight $\wts(e) = (1+\epsilon)^{-\lev(e)}$, and the weight $\wts(s)$ of each set $s$ is given as the total weight of elements in it, namely $\wts(s) = \sum_{e\in s}\wts(e)$. A set $s$ is called \emph{tight} if $\wts(s)\geq c_s / (1+\epsilon)$. 
The primal-dual framework maintains a hierarchical partition into levels as above, aiming to satisfy the following invariants.
\begin{itemize}
	\item $\wts(s) \le c_s, \forall s\in\set$. 
	\item All sets on level $>0$ are tight.
\end{itemize}

If both invariants are met, then weak duality implies that the set $T\subseteq \set$ of all tight sets provides a $(1+\epsilon)f$-approximate set cover, i.e., $c(T) \le (1+\eps)f
\cdot c(\set^*)$.

\paragraph{Local approach.}  To dynamically maintain an approximate minimum set cover via the primal-dual framework, it is 
perhaps most natural to employ the so-called {\em local approach}: 
Each time an element is inserted or deleted, the algorithm will perform some {\em local} ``fixing'' steps ``around the update'' to recover both invariants, to restore a valid and up-to-date hierarchical partition (including up-to-date induced weights). 
This local approach, which was implemented in \cite{bhattacharya2015design}, has two drawbacks: (1) The update time is $O(f\log(m+n))$, which in particular depends on $m, n$, and (2) the approximation ratio is $O(f^2)$ rather than $O(f)$. To shave the $\log n$ factor in the update time, \cite{bhattacharya2019deterministically} studied the special case of   vertex cover, and introduced a new analysis of the local approach that improves the update time to $O(1)$. Although this new analysis of the local approach   generalizes for set cover, it does not fix the second drawback of approximation   $O(f^2)$.

\paragraph{Global approach.}   To obtain a $(1+\epsilon)f$-approximation, the subsequent works \cite{bhattacharya2019new,abboud2019dynamic} adopted a {\em global approach} to maintain the primal-dual hierarchical partition. Basically, instead of recovering the invariants persistently after every element update, the global approach only handles the updates in the following lazy manner.

For each insertion of some element $e$, if we insist that $\wts(e) = (1+\epsilon)^{-\lev(e)}$, then $\wts(s)$ for some sets $s \ni e$ might exceed $c_s$; to satisfy the first invariant, we would have to raise the level of such sets, which might set off a long cascade of level changes of elements and sets. The lazy approach would be to simply assign the largest possible weight $\wts(e) = (1+\epsilon)^{-l}$ without violating any constraints $\wts(s)\leq c_s,  s \ni e$. In this way, we have relaxed the requirement that $\wts(e)$ is equal to $(1+\epsilon)^{-\lev(e)}$ by assigning it a smaller weight $(1+\epsilon)^{-\ilev(e)}$ for some {\em intrinsic} level $\ilev(e)$. This relaxation naturally partitions all existing elements into two categories: (1) \emph{active} elements $e$ where $\wts(e) = (1+\epsilon)^{-\lev(e)}$, and (2) \emph{passive} elements $e$ where $\wts(e) = (1+\epsilon)^{-\ilev(e)} < (1+\epsilon)^{-\lev(e)}$.

For each deletion of some element $e$, we simply ignore it, and when deletions have accumulated to a large extent,
a rebuild procedure is invoked, which rebuilds a carefully chosen ``prefix'' of the primal-dual hierarchical partition. 
Roughly speaking, when the approximation of the current set cover might exceed $(1+\epsilon)f$, the algorithm of
\cite{bhattacharya2019new} looks for the lowest level $k$ such that the fraction of deleted elements on levels $\leq k$ is large. 
Then the entire primal-dual hierarchy from levels $0$ to $k$ is rebuilt by first moving all existing elements on levels $\leq k$ to level $k+1$ and then pushing them downward using a discretized water-filling procedure. This ensures that for any element $e$ that remains passive,
the gap $\ilev(e) - \lev(e)$ decreases.
It can be shown that the runtime of the rebuild procedure is $O(f|A_{\leq k}| + f|P_{\leq k}|)$, where $A_{\leq k}, P_{\leq k}$ are the sets of active and passive elements that lied on levels $\leq k$ before the rebuild, respectively. For the amortized analysis,   the term $f|A_{\leq k}|$ can be charged to the deletions that have accumulated, and the term $f|P_{\leq k}|$ can be charged (via a potential function analysis) to the decrease of gaps $\ilev(e) - \lev(e), e\in P_{\leq k}$. Using the fact that the gap $\ilev(e) - \lev(e)$ is bounded by $O(\log n)$, an amortized update time of $O(f\log n)$ is   derived.

\paragraph{Combining local and global approaches.} 
To shave the $\log n$ factor while preserving a $(1+\epsilon)f$ approximation, \cite{bhattacharya2021dynamic} combines the local approach with the global approach in the following way. For insertion $e$, they assign the true weight $\wts(e) = (1+\epsilon)^{-\lev(e)}$, and apply the local approach from \cite{bhattacharya2019deterministically} to fix the violated constraints of the first invariant, if any. For deletion $e$, they follow the same rebuild procedure from \cite{bhattacharya2019new}. Now there is no dependency on $\log n$, since every element is always active (and the gap $\ilev(e) - \lev(e)$ does not exist at all).

Alas, this approach incurs a quadratic dependency on $f$. Indeed, in the analysis of \cite{bhattacharya2021dynamic}, which uses a potential function $\Phi(\cdot)$, each newly inserted element $e$ adds roughly $\wts(e)\cdot f(1+\epsilon)^{\lev(s)}$ units to the potential $\Phi(s)$ of element $s\ni e$, and summing over all up to $f$ sets $s\ni e$, the total potential increase could be as large as $f^2$.

\subsubsection{Our Approach}
\paragraph{A careful balance between local and global approaches.} 
To improve over previous works, and in particular to bypass the quadratic-in-$f$ time barrier in \cite{bhattacharya2021dynamic}, 
we seek a better balance between the local and global approaches. On the one hand, to avoid the quadratic-in-$f$ potential increase due to an element insertion, we will still allow $e$ to be passive, so that we can avoid the heavy cost that is incurred by the local approach to fix the violated constraints. On the other hand, we do not want $e$ to be {\em too passive}, so that $e$ does not participate in too many instances of rebuilding before it becomes active, as this might blow up the update time by a factor of $\log n$.
To express this idea in terms of levels, we would like to balance two contradictory requirements:
the first is that the gap $\ilev(e) -\lev(e)$ would be large, while the second is that the gap $\ilev(e) - \lev(e)$ would be small.

To optimize the balance,
we need to overcome several highly nontrivial technical hurdles. Our resulting algorithm is inherently different than the previous ones, and so is our analysis. 
We next sketch the core idea of the argument (ignoring most of the technical issues that arise).
When an element $e$ is inserted, we will assign $\ilev(e) = \lev(e) + \log_{1+\epsilon}f$, which bounds the gap $\ilev(e) - \lev(e)$ by $O(\log f)$.
On the one hand, we can show that the total potential increase due to fixing the violated constraints would be smaller by a factor of $f$,
as compared to \cite{bhattacharya2019deterministically};
to fix the violated constraints, we basically follow the same local approach as in previous works (with several important modifications, which we skip here). On the other hand, if there are no violated constraints with respect to the intrinsic level $\ilev(e) = \lev(e) + \log_{1+\epsilon}f$ assigned to $e$, we can make sure that the total time spent on $e$ would be $O(f\log f)$.
More specifically, the algorithm will carefully make sure that the gap $\ilev(e) - \lev(e)$ never increases, which is a key technical challenge that the algorithm and analysis must face. Moreover, each time the passive element $e$ participates in a call to the rebuild procedure, the gap $\ilev(e) - \lev(e)$ will decrease by at least one. Therefore $e$ can participate in at most $\log_{1+\epsilon}f$ calls to the rebuild procedure, which we show ultimately implies that the total time spent on $e$ is $O(f\log f)$.

\paragraph{Going below $O(f \log f)$ update time: sampling and bootstrapping.}
To go below $O(f\log f)$ update time, let us take a closer look at the rebuild procedure. For each passive element $e\in P_{\leq k}$, in previous works, one had to scan all the sets $s\ni e$ to test whether $e$ can be activated on level $k+1$
(whether $\wts(s) -\wts(e) +  (1+\epsilon)^{-k-1} \leq c_s$ is not violated for all $s \in e$), which takes time $O(f)$. The worst-case performance of the algorithm occurs when such tests always fail, so that one always pays $O(f)$ time to decrease the gap $\ilev(e) - \lev(e)$ by one. 
To improve the runtime, we would like to be able to decrease this gap {\em exponentially}, i.e., from $d = \ilev(e) - \lev(e)$ to $\log d$;
Alas, this is not always possible.
To overcome this hurdle,
our key insight is to only sample $O(f / \log f)$ sets $s\ni e$ and test whether $e$ can be activated
with respect to all sampled sets (whether $\wts(s) -\wts(e) +  (1+\epsilon)^{-k-1} \leq c_s$ is not violated for all sampled sets). 
If there are at least $10\log^2f$ {\em witness} sets $s$ for which the test is violated, then one of them will be sampled with good probability, and in that case we have shaved off a $\log f$ factor from the time needed to process $e$ due to the rebuild procedure. 
Otherwise, we will push down the intrinsic level of $e$ from level $k+1+\log_{1+\epsilon} f$ to level $k+1+2\log_{1+\epsilon}\log_{1+\epsilon}f$, which increases $\wts(e)$ to $\frac{1}{\log^2_{1+\epsilon}f}(1+\epsilon)^{-k-1}$, and then apply the local approach to fix the violated constraints. A crucial observation is that we know that the total number of violations is bounded by $\log_{1+\epsilon}^2f$, which is exponentially smaller than the trivial bound $f$, and so we can bound the potential increase by $O(f)$ instead of $O(f^2)$. We demonstrate that by a careful repetition of this observation, the gap $\ilev(e) - \lev(e)$ can be reduced {\em exponentially} in $O(f)$ time, which ultimately leads to the desired update time of $O(f\log^*f)$.

\paragraph{Summary.~}
The starting point of our work is the aforementioned dynamic primal-dual algorithms for set cover.
However, to break the quadratic-in-$f$ time barrier, and further to achieve the near-optimal  (up to the $\log^* f$ slack factor) update time, 
we had to deviate significantly from previous works.  
The facts that our approach provides (1) the only randomized set cover algorithm that works against an adaptive adversary,
and (2) a rare example of achieving optimal time to within a $\log^* n$ factor  --- 
may serve as some ``evidence'' for the novelty of our algorithm and its analysis. 

\section{Preliminaries}
\begin{definition}
    For any real values $\epsilon \in (0, 1), y \geq 1$ and integer $\eta\geq 1$, inductively define: $$\clog_{1 + \epsilon}^{(\eta)}y = 5\cdot\log_{1 + \epsilon}\brac{\clog_{1 + \epsilon}^{(\eta-1)}y}$$
    where $\clog_{1 + \epsilon}^{(0)}y = y$, and define $\clog_{1 + \epsilon}^*(y)$ to be the minimum value of index $\eta$ such that $\clog_{1 + \epsilon}^{(\eta)}(y)\leq \frac{200}{\epsilon^2}$.    
\end{definition}
The following lemma shows that $\clog_{1 + \epsilon}^*(y)$ is well-defined.
\begin{lemma}\label{lm:iterated-log}
    $5\cdot\log_{1 + \epsilon} y \le \sqrt{\frac{\epsilon}{5}} \cdot  y$ for any $y \ge \frac{200}{\epsilon^2}$.
\end{lemma}
\begin{proof}
    First, notice that $\ln (1 + \epsilon) \ge \frac{1}{2} \epsilon$ for $\epsilon \in (0, 1)$, so $5\cdot\log_{1 + \epsilon} y \le \frac{10}{\epsilon} \ln y$. Thus, it suffices to show that $\ln y \le \frac{\epsilon\sqrt{\epsilon}}{10\sqrt{5}} y$.
    Let $y = \frac{200}{\epsilon^2}(1 + x)$ for some $x \ge 0$. Then 
    \[\ln y = \ln \frac{200}{\epsilon^2}(1 + x) = 4 \ln \frac{\sqrt{5}(1 + x)^{1/4}}{\sqrt{\epsilon}} + \ln 8 \le \frac{4\sqrt{5}}{\sqrt{\epsilon}}(1 + x)^{1/4} \le \frac{4\sqrt{5}}{\sqrt{\epsilon}}(1 + x) = \frac{\epsilon\sqrt{\epsilon}}{10\sqrt{5}} y,\]
    where the first inequality is due to $\ln z \le z - 1$ for any $z > 0$ and $\ln 8 \le  4$.
\end{proof}
\Cref{lm:iterated-log} implies that $\clog_{1+\epsilon}(y) = O(\log^* y)$, since applying $\clog_{1+\epsilon}$ three times either results in something bounded by $\frac{200}{\epsilon^2}$, or decreases the argument exponentially, i.e. $\clog_{1 + \epsilon}^{(3)} (y) \le \epsilon / 5 \cdot 5\log_{1 + \epsilon} y \le 2 \ln y$.

\subsection{Primal-dual framework}

 We will always assume that $f > \frac{\log C}{\epsilon}$, since otherwise we will simply apply the algorithm from \cite{bhattacharya2021dynamic}. For each element $e\in \univ$, we assume all the sets $s$ containing $e$ are stored as an array, not a linked list, so that we can take uniformly random samples from all these sets in $O(1)$ time. This assumption is valid because only the elements are dynamic, while all sets are static.

We will follow the primal-dual framework from \cite{bhattacharya2021dynamic,bhattacharya2019new,bhattacharya2019deterministically}. However, there is a tiny difference: instead of aiming to satisfy $\wts(s) \le c_s, \, \forall s \in \set$, we aim to satisfy $\wts(s) < c_s, \, \forall s \in \set$. This is not crucial for the approximation guarantee, but this simplifies the algorithm and the analysis.

Let $\epsilon \in (0, 0.1)$ be a constant. Define $L = \ceil{\log_{1+\epsilon}(Cn)}+1$. Each set $s\in S$ is assigned a level $\lev(s)\in [L]$. The base level of a set is defined as $\bs(s) = \floor{\log_{1+\epsilon}1/c_s}$.

Each element $e$ will be assigned a level $\lev(e) = \max_{s\ni e}\{\lev(s)\}$ and weight $\wts(e)$, and $\wts(s) = \sum_{e\in s}\wts(e)$ denotes the total weight of $s\in \set$. In addition, for every set $s\in \set$, we also maintain a \emph{dead weight} $\phi(s)$, and let  $\wts^*(s) = \wts(s) + \phi(s)$ be the \emph{composite weight}.

\begin{definition}
	A set $s$ is called \emph{tight}, if $\wts^*(s) \geq \frac{c_s}{1+\epsilon}$, and \emph{slack} otherwise.
\end{definition}

\subsection{Basic data structures} \label{basicds}

During the dynamic algorithm, we will not keep track of the value of $\lev(e)$ for all elements. Instead, we will maintain a \emph{lazy level} $\zlev(e)$. In addition, we will maintain an \emph{intrinsic level} $\ilev(e)$, which defines the weight of an element: $\wts(e) = (1 + \epsilon)^{-\ilev(e)}$. Because of that, all elements have two categories: \emph{active} and \emph{passive}.
\begin{itemize}[leftmargin=*]
    \item \textbf{Active.} 
    If an element $e$ is {\em active}, then the value of $\lev(e)$ will be correctly maintained. For such elements we will have $\zlev(e) = \ilev(e) = \lev(e)$, and so  $\wts(e) = (1+\epsilon)^{-\lev(e)}$. Let $A_i\subseteq \univ$ be the set of active elements on level $i$. For each set $s$ and each level index $i$, our algorithm explicitly maintains a list $A_i(s)\subseteq A_i$ which is the set of active elements in $s$ on level $i$.

    \item \textbf{Passive.} 
    If an element $e$ is {\em passive}, due to runtime issues, we might not always keep track of the value $\lev(e)$ all the time. Instead, we can only maintain a \emph{lazy} level $\zlev(e)\leq \lev(e)$ which is refreshed to $\lev(e)$ once in a while. The \emph{intrinsic} level $\ilev(e)$ will satisfy $\lev(e)< \ilev(e)\leq \zlev(e) + \left\lceil\log_{1+\epsilon}\max\{f, \frac{2C}{\epsilon} \}\right\rceil$.
    
    Let $P_i\subseteq \univ$ be the set of all passive elements whose intrinsic levels are $i$. For each set $s$ and each intrinsic level $i$, our algorithm explicitly maintains a list $P_i(s)\subseteq P_i$ which is the set of passive elements in $s$ on intrinsic level $i$. In contrast, we will not maintain a list for the set of passive elements in set $s$ on lazy level $i$ (since that would be too time-consuming), hence we are unable to enumerate all the passive elements in $s$ on lazy level $i$.
\end{itemize}

For a set $s$ and an index $i\geq \lev(s)$, the weight of $s$ at level $i$ is defined as:
\begin{equation}
\label{eqlevel}
\begin{aligned}
	\wts(s, i) &= \sum_{\text{active }e\in s}(1+\epsilon)^{-\max\{i, \max_{t\mid e\in t\neq s} \lev(t) \}} + \sum_{\text{passive }e\in S}(1+\epsilon)^{-\max\{i, \ilev(e)\}}\\
	&= \sum_{e\in s}\min\left\{\wts(e), (1+\epsilon)^{-\max\{i, \max_{t\mid e\in t\neq s} \lev(t) \}}\right\}
\end{aligned}
\end{equation}
In other words, $\wts(s, i)$ is the weight of $s$ if it were raised to level $i$. So by definition, $\wts(s) = \wts(s, \lev(s))$, and 
\begin{equation}
\label{eqlevelup}
    \wts(s, \lev(s)+1) = \wts(s) - |A_{\lev(s)}(s)|\cdot \epsilon(1+\epsilon)^{-\lev(s)-1},
\end{equation}
which can be computed in $O(1)$ time once we know $\wts(s)$ and $|A_{\lev(s)}(s)|$. We assume that all powers of $1+\eps$ can be computed in constant time; one way of implementing this efficiently is to compute all these powers at the outset in $O(L)$ time.

Throughout the algorithm, let $T\subseteq \set$ be the set of all tight sets, and let $\phi = \sum_{s\in \set}\phi(s)$ be the total dead weight. For each $i$, let $S_i \subseteq \set$ be the set of sets at level $i$, let $T_i\subseteq T$ be the set of tight sets at level $i$, let $E_i$ be the set of elements such that $\zlev(e) = i$ (note that $E_i$ contains all active elements such that $\lev(e) = i$, since for them $\zlev(e) = \lev(e)$). Define $\phi_i = \sum_{s \in S_i} \phi(s)$, the total dead weight of sets on level $i$; similarly, we can define notations $\phi_{\leq i}, S_{\leq i}, T_{\leq i}, E_{\leq i}$.

Each set $E_i, S_i, T_i$ will be maintained as a linked list, and we store all pointers to lists $\{S_i\}_{0\leq i\leq L}$, $\{T_i\}_{0\leq i\leq L}$, $\{E_i\}_{0\leq i\leq L}$ as three arrays of length $L+1$. When the values of $\lev(s), \zlev(e), \wts(s), \wts(e)$ change for a set $s$ or an element $e$, we can update the lists and the values $\phi_i, \wts(E_i), \wts(S_i), \wts(T_i)$ accordingly in constant time.

\paragraph{Iterating over nonempty sets.} In the algorithm, we want to be able to access nonempty sets from $\{E_i\}_{0\leq i\leq L}$, $\{S_i\}_{0\leq i\leq L}$, $\{T_i\}_{0\leq i\leq L}$ efficiently in the increasing order by $i$. If we were maintaining them in doubly linked lists, we would not be able to update the lists in constant time whenever we update the values of $\zlev(e), \lev(s)$. To cope with that, we rely on the way our algorithm update these values. First, notice that we can update the lists in constant time if we increase these values by one. Another idea is to maintain the doubly linked lists only for levels above $\ceil{\log_{1+\epsilon}C}+1$, so we can update them in constant time whenever we set $\lev(s)$ to some $k \le \ceil{\log_{1+\epsilon}C}+1$. We use the same idea for $\{E_i\}_{0\leq i\leq L}$, but we also store $E_i$ whenever $T_i \ne \emptyset$, even if $E_i = \emptyset$. That way, we are able to update the list when we set $\zlev(e)$ to $\max_{s \ni e} \{ \lev(s) \}$.
This allows us to compute quantities $\phi_{\le i}$, $c(T_{\le i})$ and $\wts(E_{\le i})$ more efficiently.

\begin{observation} \label{linkedlist}
This linked list data structure allows us to compute quantities $\phi_{\leq i}, c(T_{\leq i})$ in $O\brac{\left|T_{\leq i}\setminus T_{\leq \ceil{\log_{1+\epsilon}C}+1}\right| + \frac{\log C}{\epsilon}}$ time, and enumerate elements from $E_{\le i}$ or compute $\wts(E_{\leq i})$ in $O\brac{\left|T_{\leq i}\setminus T_{\leq \ceil{\log_{1+\epsilon}C}+1}\right| + \frac{\log C}{\epsilon} + |E_{\le i}|}$ time.
\end{observation}

\paragraph{Implicit zeroing.}\label{implicit-zeroing}
We need a fast data structure for the following operation.
\begin{itemize}[leftmargin=*]
	\item Given a level index $0\leq i\leq  L$, we want to assign $\lev(s),\phi(s)\leftarrow 0$ for all $s \in S_i$, and we need to do this in constant time.
\end{itemize}

Updating the lists $T_i, S_i$ or sums $\phi_i, \phi_0, c(T_i), c(S_i)$ can be done in constant time. However, this task is impossible if we want to explicitly update all the values $\phi(s),\lev(s)\leftarrow 0$ for all $s \in S_i$. So, we have to zero out each individual value $\lev(s), \phi(s)$ in an implicit way. To do this, for each set $s\in\set$, we will associate it with a time stamp $\tm(s)$ which indicates the latest time point when the value of $\lev(s)$ or $\phi(s)$ is explicitly updated. Then, create an array $\aux$ of length $L + 1$, where each entry $\aux[i]$ stores the time point $t$ of the latest zeroing operation to level $i$. 
Each time we want to access the values of $\lev(s), \phi(s)$, compare $\tm(s)$ and $\aux[\lev(s)]$. If $\tm(s) > \aux[\lev(s)]$, we know that $s$ did not suffer the latest zeroing out on level $\lev(s)$, and hence $\lev(s), \phi(s)$ are referring to their current values; otherwise, $s$ must have undergone a zeroing operation implicitly. In this case, explicitly set $\lev(s), \phi(s)$ to $0$, and update $\tm(s)$ accordingly. Here we have implicitly assumed that the time values can be stored in a single word; otherwise, we would rebuild the entire dynamic set cover data structure and reset the time to zero.

Zeroing out the levels of $\lev(s)$ may also affect the levels of other elements. But in our algorithm, we will apply implicit zeroing in a careful manner, so that the levels of sets and elements are consistent.

\subsection{Approximation guarantees}

\begin{invariant}\label{inv}
	During the algorithm, we will maintain the following invariants.
	\begin{enumerate}[(1)]
		\item For any set $s$, $\wts(s, \lev(s)+1) < c_s$. \
		\item All sets at level at least $1$ are tight.
		\item It always holds that $\phi\leq \epsilon \left(c(T) + f\cdot \wts(\univ)\right)$.
	\end{enumerate}
\end{invariant}

\begin{corollary}\label{wts-bound}
    If \Cref{inv}(1) holds, then $\wts(s) < (1 + \epsilon)c_s$ for all sets $s \in \set$. 
\end{corollary}
\begin{proof}
    Notice that $\wts(s) \ge |A_{\lev(s)}(s)| \cdot (1 + \epsilon)^{-\lev(s)}$, since every $e \in A_{\lev(s)}(s)$ has $\wts(e) = (1 + \epsilon)^{-\lev(s)}$. Thus, \[\wts(s) - |A_{\lev(s)}(s)| \cdot \epsilon(1 + \epsilon)^{-\lev(s) - 1} \ge \wts(s) - \frac{\epsilon}{1 + \epsilon}\wts(s) = \frac{\wts(s)}{1 + \epsilon}.\] Plugging this into \Cref{eqlevel}, we get $\wts(s) \le (1 + \epsilon) \cdot \wts(s, \lev(s) + 1) < (1 + \epsilon)c_s$.
\end{proof}

\begin{lemma}[\cite{bhattacharya2021dynamic}]\label{ele-wts}
	If \Cref{inv} holds, then $\wts(s) < (1+\epsilon) c_s$, and $\wts(\univ) \leq (1+\epsilon)\cdot \opt$, where $\opt = c(\set^*)$ is the total weight of an optimal set cover $\set^*$.
\end{lemma}
\begin{proof}
	By \Cref{wts-bound}, $\wts(s) < (1+\epsilon) c_s$. Furthermore, we have
	$$\wts(\univ) = \sum_{e\in \univ}\wts(e)\leq \sum_{s\in \set^*}\sum_{e\in s}\wts(e)\leq (1+\epsilon)\cdot c(\set^*) = (1+\epsilon)\cdot \opt$$
	This first inequality relies on the fact that $\set^*$ is a valid set cover.
\end{proof}

\begin{lemma}[\cite{bhattacharya2021dynamic}]
	If \Cref{inv} holds and the collection of tight sets $T$ is a set cover, then $T$ is a $(1+5\epsilon)f$-approximate set cover.
\end{lemma}
\begin{proof}
    By the definition of a tight set, we have $\wts^*(s) = \wts(s) + \phi(s) \geq \frac{c_s}{1+\epsilon}$. Then, the cost of $T$ is bounded by
	 $$\begin{aligned}
		c(T) \leq (1+\epsilon)\cdot \sum_{s\in T}(\wts(s) + \phi(s)) &\leq (1+\epsilon)\cdot \wts(\set) + (1+\epsilon)\cdot\phi\\
		&\leq (1+\epsilon)f\cdot \wts(\univ) + \epsilon(1+\epsilon)\cdot c(T) + \epsilon(1+\epsilon)f\cdot \wts(\univ)\\
		&\leq (1+\epsilon)^2f\cdot \wts(\univ) + \epsilon(1+\epsilon)\cdot c(T)
	\end{aligned}$$
	As $\epsilon\in(0, 0.1)$ and by \Cref{ele-wts}, we have
	$$c(T) \leq \frac{(1+\epsilon)^2f}{1 - \epsilon(1+\epsilon)}\cdot \wts(\univ)\leq (1+5\epsilon)f\cdot \opt$$
\end{proof}

\subsection{Glossary}
Some of the notations used are summarized in \Cref{table:notation-summary} (placed in the last page for convenience). 

\section{Algorithm description}
We will describe subroutines $\del(e)$, $\ins(e)$, $\fix(e, l)$, and $\reset(k)$ which constitute the main update algorithm, whose pseudocode is given in \Cref{alg:main}. At the beginning of the algorithm, we assume $\univ$ is empty, and so all sets in $\set$ are initialized on level $0$. When an element $e$ is deleted from $\univ$, we will call subroutine $\del(e)$ to deal with it; if an element $e$ is inserted, then we will call $\ins(e)$. 

After that, we check if \Cref{inv}(3) is violated. If so, we find the smallest index $k$ such that $\phi_{\leq k} > \epsilon\cdot \left(c(T_{\leq k}) + f\cdot \wts(E_{\leq k})\right)$ and then invoke subroutine $\reset(k)$;
this is repeated until \Cref{inv}(3) holds.
\begin{algorithm}
    \caption{\textsf{DynamicSetCover}}\label{alg:main}
    initialize $\lev(s) = 0, \forall s\in \set$\;
    \ForEach{element update $e$}{
        \If{$e$ is deleted}{
            $\del(e)$\;
        }\Else{
            $\ins(e)$\;
        }
        \While{\Cref{inv}(3) is violated}{\label{main-while}
            find the smallest $k$ such that $\phi_{\leq k} > \epsilon\left(c(T_{\leq k}) + f\cdot \wts(E_{\leq k})\right)$\;\label{find-k}
            $\reset(k)$\;
        }
    }
\end{algorithm}

To implement \cref{find-k} which finds the smallest index $k$ such that $\phi_{\leq k} > \epsilon\cdot \left(c(T_{\leq k}) + f\cdot \wts(E_{\leq k})\right)$, start with $k = 0$ and each time increase $k$ to the next index $k\leftarrow k^\prime$ where $T_{k^\prime}\neq \emptyset$ or $E_{k^\prime} \ne \emptyset$, using the doubly linked list data structure, and check if $\phi_{\leq k} > \epsilon\cdot \left(c(T_{\leq k}) + f\cdot \wts(E_{\leq k})\right)$. In this way, the runtime of locating the smallest $k$ would be 
$O\brac{|T_{\leq k}\setminus T_{\leq \ceil{\log_{1+\epsilon}C}+1}| + \frac{\log C}{\epsilon} + |E_{\le k}|}$ (similarly to  \Cref{linkedlist}); note that the amount of time spent per level is constant, since we have maintained the quantities per each level separately, and we just need to sum the quantities for prefixes of levels.

We will make sure that each of the $\ins$, $\del$ and $\reset$ subroutines does not increase the gap $\ilev(e) - \zlev(e)$ for any passive element $e$. That property will allow us to bound the total time spent on element $e$.

\subsection{Deletion}
We handle deletions in the same way as \cite{bhattacharya2021dynamic}; refer to \Cref{alg:del} for the pseudocode. When an element $e$ is deleted, the algorithm subtracts, for each set $s \ni e$, the value of $\wts(e)$ from its weight $\wts(s)$, and compensates for the loss by increasing the dead weight $\phi(s)$ by $\wts(e)$, if $s$ was tight. 
This $\del(e)$ subroutine takes $O(f)$ time.

Besides, we also need to specify how to maintain the underlying data structures after an element deletion. If $e$ is active, then we go over all sets $s\ni e$ and remove $e$ from the linked list $A_{\ilev(e)}(s)$; if $e$ is passive, then we go over all sets $s\ni e$ and remove $e$ from the linked list $P_{\ilev(e)}(s)$. This operation takes time $O(f)$.

\begin{algorithm}
    \caption{$\del(e)$}\label{alg:del}
    \ForEach{set $s\ni e$}{
        $\wts(s)\leftarrow \wts(s) - \wts(e)$\;
        \If{$s$ was tight}{
            $\phi(s)\leftarrow \phi(s) + \wts(e)$\;
        }
    }
\end{algorithm}

As for the invariants, since $\del(e)$ does not increase any weight $\wts(s)$, \Cref{inv}(1) is preserved. \Cref{inv}(2) is also preserved due to the way we modify the dead weights. \Cref{inv}(3) might have been violated due to the increases of dead weights, but it will be restored by the while loop on \cref{main-while} of \Cref{alg:main}.

\subsection{Insertion}\label{subsec:insert}
\paragraph{High-level idea.} When inserting an element $e$, we aim to satisfy the constraint $\forall s \ni e,\,\wts(s) < c_s$. We try to make the newly inserted element $e$ active at level $\lev(e) = \max_{s\ni e}\{\lev(s) \}$ if possible. If not, to make sure $e$ is covered by a tight set, we try to make it passive at the lowest possible intrinsic level up to $l = \max_{s\ni e}\{\lev(s) \} + \left\lceil\log_{1+\epsilon}\max\{f, \frac{2C}{\epsilon} \}\right\rceil$, since we want $\ilev(e) \le \zlev(e) + \left\lceil\log_{1+\epsilon}\max\{f, \frac{2C}{\epsilon} \}\right\rceil$. This may still be impossible; in that case, we invoke $\fix(e, l)$. 

Thus, we invoke $\fix(e, l)$  when we can't make the element passive on intrinsic level $l \le \max_{s\ni e}\{\lev(s) \} + \ceil{\log_{1+\eps} f}$. Our amortized analysis employs a potential function. As we will show later (the full analysis is given in \Cref{potfun}, this is just for intuition), the \emph{potential increase} and thus the amortized cost of $\fix(e, l)$ is bounded by roughly $f^2 \cdot (1 + \epsilon)^{-d}$, where $d = \ilev(e) - \lev(e)$. Hence by placing the element $\ceil{\log_{1 + \epsilon} f}$ levels higher than $\lev(e)$, we can guarantee that the {\em potential increase} is roughly $f$. We note that we call to $\fix(e, l)$ only when the gap $\ilev(e) - \lev(e)$ is at least $\log_{1 + \epsilon} \frac{2C}{\epsilon}$ (hence the reason for taking the maximum of $f$ and $\frac{2C}{\epsilon}$); that restriction, as we will show later, guarantees that the invariants are preserved.

\subsubsection{Description of the $\ins$ subroutine} Upon an insertion $e$, assign $\zlev(e) \gets \max_{s\ni e}\{\lev(s) \}$ (thus making $\zlev(e) = \lev(e)$). Define $l = \max_{s\ni e}\{\lev(s) \} + \left\lceil\log_{1+\epsilon}\max\{f, \frac{2C}{\epsilon} \}\right\rceil$ and $F = \{s\ni e\mid \wts(s) + (1+\epsilon)^{-l} \ge c_s \}$. Next, we branch into two cases:

\begin{enumerate}[(1),leftmargin=*]
    \item If $F = \emptyset$, then it is possible to insert $e$ at an intrinsic level $\ilev(e) \le l$ without violating $\wts(s) < c_s$ for any $s \ni e$. In that case, we compute the smallest index $h\geq \zlev(e)$ such that $\forall s \ni e,\, \wts(s) + (1+\epsilon)^{-h} < c_s$ and set $\ilev(e) \gets h$. After that, we update the weights of sets by going over each set $s \in e$ and setting $\wts(s) \gets \wts(s) + \wts(e)$. We also add $e$ to $A_h(s)$ if $e$ is active (i.e., $\ilev(e) = \zlev(e)$), or $P_h(s)$ if $e$ is passive, for each $s \ni e$ (which is omitted in the pseudocode). 
    
    \textbf{Computing $h$ in $O(f)$ time.} To compute $h$ in $O(f)$ time, we can first compute the minimum value of the gap $c_s - \wts(s)$, and then use binary search over the interval $[\zlev(e), l]$ to find $h$, which takes time $O\left(\log \left(\log_{1+\epsilon}\max\{f, \frac{2C}{\epsilon} \}\right)\right)$. Note that $\log_{1+\epsilon}\frac{2C}{\epsilon} = O\left(\frac{\log C}{\epsilon} + \frac{1}{\epsilon^2}\right)$, since $\log_{1+\epsilon} x = O\left(\frac{\log x}{\epsilon}\right)$ for any $x > 0, \epsilon \in (0, 1)$. Recall that we assume $f > \frac{\log C}{\epsilon}$. Then the time to find $h$ is $O\left(\log \left(\frac{\log f}{\epsilon} + \frac{\log C}{\epsilon} + \frac{1}{\epsilon^2}\right)\right) = O(f)$, where the last transition holds due to $\frac{1}{\epsilon} \le \frac{\log C}{\epsilon} < f$. This operation will appear again in the $\reset$ subroutine.
    
    \item If $F \neq \emptyset$, then it is impossible to insert $e$ at an intrinsic level $\ilev(e) \le l$ without violating $\wts(s) < c_s$ for some $s \ni e$. Inserting $e$ at intrinsic level $\ilev(e) = l$ may violate \Cref{inv}(1). Hence we apply subroutine $\fix(e, l)$, which will make $e$ passive at intrinsic level $l$ or higher, but will keep the gap $\ilev(e)-\zlev(e)$ (and the gap $\ilev(e) - \lev(e)$ as well, since it also keeps $\zlev(e) = \lev(e)$) equal to $\left\lceil\log_{1+\epsilon}\max\{f, \frac{2C}{\epsilon} \}\right\rceil$, and make sure that all the invariants are satisfied.
\end{enumerate}

\begin{algorithm}
    \caption{$\ins(e)$}\label{alg:ins}
    \SetKw{Let}{let}
    assign $\zlev(e) \gets \max_{s\ni e}\{\lev(s) \}$\;
    \Let{$l = \max_{s\ni e}\{\lev(s) \} + \left\lceil\log_{1+\epsilon}\max\{f, \frac{2C}{\epsilon} \}\right\rceil$}\;
    \Let{$F = \{s\ni e\mid \wts(s) + (1+\epsilon)^{-l} \ge c_s \}$}\;
    \If{$F = \emptyset$}{
        compute the smallest index $h\geq \zlev(e)$ such that $\wts(s) + (1+\epsilon)^{-h} < c_s, \forall s\ni e$\;
        $\ilev(e) \gets h$\;
        \ForEach{$s \ni e$}{
            $\wts(s) \gets \wts(s) + \wts(e)$\;
        }
    }\Else{
        $\fix(e, l)$\;
    }
\end{algorithm}

It is left to show that $\ins$ maintains a valid set cover and preserves the invariants. We defer the proof of the following theorem to \Cref{sec:fix-properties}, since our argument relies on the properties of the $\fix$ subroutine, which we have not stated yet.
\begin{restatable}{theorem}{insinv}
    \label{ins-inv}
    After the call to $\ins(e)$, $T$ is a set cover and \Cref{inv}(1)(2) are maintained.
\end{restatable}

\subsection{Fixing levels}\label{subsec:fixlevel}
\paragraph{High-level idea.} When the subroutine $\fix(e, l)$ is called, we assign the intrinsic level of a passive or new element $e$ to $l$. However, this update can increase $\wts(s)$ for $s \ni e$, and hence could violate \Cref{inv}(1). To restore it, the subroutine then goes over each set $s \ni e$ and repeatedly raises $s$; that is, increases $\lev(s)$ and increases the levels of elements in $s$ if needed, until the invariant is satisfied. 
Additionally, it makes sure that the gap $d = \ilev(e) - \lev(e)$ remains the same, by increasing $\ilev(e)$ whenever $\lev(e)$ increases. The fact that the gap does not decrease is crucial for the amortized runtime analysis.

The exact definition of the potential functions used in the amortized runtime analysis is provided later in \Cref{potfun}. Intuitively, each unit of ``excess weight'', that is $\max\{\wts(s) - c_s, 0\}$, has a potential cost, which depends $\lev(s)$. Increasing $\lev(s)$ by one increases the potential cost of one unit of weight by a factor of $1 + \epsilon$. If we had kept $\ilev(e)$ unchanged, $\lev(s)$ could become close to $\ilev(e)$, which would make the potential increase too large. So our rule here is to keep the gap $d=\ilev(e) - \lev(e)$ the same by increasing $\ilev(e)$ whenever $\lev(e)$ increases.

Let $F$ be the set of set $s \ni e$, for which $\wts(s) \ge c_s$ after making the intrinsic level of $e$ to be $l$. Note that only sets from $F$ need to be raised to restore \Cref{inv}(1). As we will show later in \Cref{sec:runtime-analysis}, the amortized runtime of $\fix$ depends on the size of $F$ and the size of the gap $d = \ilev(e) - \zlev(e)$, and is roughly $|F| \cdot f \cdot (1 + \epsilon)^{-d}$. Therefore, the smaller $F$ is, the smaller we can make the gap to achieve the same runtime. Later we will use this property to improve the runtime of the $\reset$ subroutine.

Whenever we raise $s$ one level up, levels of some elements in $s$ may also increase. Since we maintain the gap $\ilev(e) - \lev(e)$, we may also need to increase $\ilev(e)$, which decreases $\wts(e)$ by a factor of $1 + \epsilon$. The tightness is maintained for $s$, since we raise $s$ only when $\wts(s, \lev(s) + 1) \ge c_s$, so $\wts(s) \ge c_s / (1 + \epsilon)$ as a result. However, this could also decrease $\wts(s^\prime)$ for some other sets $s^\prime \ne s$, and hence $s'$ may become slack. To avoid that, for each element $e^\prime \ne e$ that decreased its weight due to the raise of $s$, we compensate the loss of $\wts(s^\prime)$ incurred by $e^\prime$ by increasing the dead weight $\phi(s^\prime)$. As for the losses due to increases of $\ilev(e)$, we need to address them more carefully. Before describing how to implement this approach, there are some technical challenges we would like to explain.

\paragraph{Technical challenges.} 

If we enumerate all sets $s\ni e$ in an arbitrary order, and increase $\ilev(e)$ in each round, then the tightness of some previously visited sets $s^\prime\ni e$ might be violated. To avoid this issue, in the original algorithm of \cite{bhattacharya2021dynamic}, they enumerated the sets $s\ni e$ by increasing order of difference $c_s - \wts(s)$, which already takes $O(f\log f)$ time; we note that one cannot use a linear-time approximate sorting algorithm for this task. For our deterministic algorithm we can afford to spend the sorting time of $O(f\log f)$, but the update time of our randomized algorithm is asymptotically smaller than that; to circumvent the $\Omega(f \log f)$ sorting time overhead, we will make do with an arbitrary ordering of these sets, and restore tightness of sets $s' \ni e$ by increasing their dead weights.

Each time the intrinsic level $\ilev(e)$ of $e$ changes, we may need to update the weights of sets $s^\prime\ni e$, which already takes $O(f)$ time. However, in some cases, the decrease in potential would not be enough to cover that runtime cost. So in the worst case, the runtime cost could be as high as $O(f^2)$. To avoid this, we will update the contribution of $\wts(e)$ to the weights $\wts(s^\prime)$ of other sets $s^\prime\ni e$ in a lazy manner. More specifically, we will do so only when a set $s^\prime\ni e$ is being enumerated~--- which is when we actually need the updated value $\wts(s^\prime)$.
	
The decrease of $\wts(e)$ might also violate the tightness of some yet unvisited sets $s\ni e$. However, compensating these losses by increasing $\phi(s)$ might be too costly in terms of the potential increase. In practice, the algorithm will make sure that $\ilev(e)$ is not higher than it was before the call (if $e$ existed before the call). We will prove later that this allows us to maintain tightness for all sets.

\subsubsection{Description of the $\fix$ subroutine} 
Next, let us describe the algorithm more formally; see \Cref{alg:fix} for the pseudocode of the $\fix$ subroutine. Let $\ilev^\old(e)$ and $\wts^\old(e)$ be the intrinsic level and the weight of element $e$ right before the execution of $\fix(e, l)$; if $e$ is a newly inserted element, then define $\ilev^\old(e) = \infty$ (and hence $\wts^\old(e) = 0$). 

The $\fix(e, l)$ subroutine assumes that $\zlev(e) = \lev(e)$ at the beginning of the call and $\lev^\old(e) < l \le \ilev^\old(e)$.
The algorithm aims to maintain \Cref{inv}(1). We also note that the algorithm does not explicitly maintain \Cref{inv}(2). The reason for that is that the subroutine is called from the $\reset$ subroutine, where \Cref{inv}(2) may be (temporarily) violated. However, the algorithm makes sure that tight sets remain tight, and we will show later that if some conditions hold, no slack set get raised. For the same reasons, we allow passive elements to violate $\zlev(e') \le \lev(e')$. However, for active elements we still have $\ilev(e') = \zlev(e') = \lev(e')$, and for passive elements we still have $\lev(e')< \ilev(e')\leq \zlev(e') + \left\lceil\log_{1+\epsilon}\max\{f, \frac{2C}{\epsilon} \}\right\rceil$.

The algorithm makes sure that $\ilev(e) \le \ilev^\old(e)$ during its execution; we will rely on this property in the proof of correctness. Intuitively, every time we raise a set during $\fix$, some active elements can also increase their level, and hence lose weight. We will compensate each such loss by the increase of dead weights. However, we will not compensate the loss due to increases of $\ilev(e)$, but the inequality $\ilev(e) \le \ilev^\old(e)$ will guarantee that $\wts(e)$ is no smaller than it was at the beginning of the call, so the tightness will be preserved.

Next, let us describe the steps the algorithm makes.

\paragraph{Changing the intrinsic level of $e$ to $l$.} First, make $e$ passive at intrinsic level $l$ by setting $\ilev(e) \gets l$ and then setting $\wts(s) \gets \wts(s) - \wts^\old(e) + \wts(e)$ for each set $s \ni e$. If $e$ is not a freshly inserted element, then remove $e$ from $P_{l}(s)$ for each $s \ni e$; we will add it to the relevant sets in the end, when we finalize the changes, where we have calculated the final value of $\ilev(e)$. 
Keep a record $d = l - \zlev(e)$ and compute $F = \{s\ni e\mid \wts(s)\ge c_s \}$.

During the algorithm, we will gradually increase $\ilev(e)$ (but never above 
$\ilev^\old(e)$). However, updating all weights and the relevant data structures takes $O(f)$ time. To avoid spending $O(f)$ time for such updates each time $\ilev(e)$ increases, we will update the data structures only once, during the finalization step.

\paragraph{Raising sets $s \ni e$.} During the algorithm, we will make sure that the lazy level $\zlev(e)$ is always equal to the actual level $\lev(e)$, and the gap $\ilev(e) - \zlev(e)$ is always equal to $d$. The algorithm then goes over each set $s\ni e$ and processes it the following way: First, it updates the weight of $s$ according to the current value of $\wts(e)$, since
throughout the algorithm execution we may have changed $\wts(e)$ due to increases of $\ilev(e)$, from $(1+\epsilon)^{-l}$ to some possibly lower weight. To update $\wts(s)$ accordingly, we set  $\wts(s)\leftarrow \wts(s) - (1+\epsilon)^{-l} + \wts(e)$.

Next, we repair \Cref{inv}(1) for $s$ by increasing $\lev(s)$. Recall that $\wts(s, \lev(s)+1)$ can be computed in constant time, given $\wts(s)$ and $|A_{\lev(s)}(s)|$, by \Cref{eqlevelup}. 

\subparagraph{Raising $s$ to level $\min\{\bs(s),\zlev(e)\}$.} 
If $\lev(s)$ is below $\min\{\bs(s),\zlev(e)\}$ and we still have $\wts(s, \lev(s)+1) \ge c_s$, then we leverage the fact that there are no elements at intrinsic levels below $\bs(s)$ (the proof of that is given later in \Cref{fix-empty}). Because of that, we can raise $s$ to $\min\{\bs(s),\zlev(e)\}$ in a single shot.
Since $\wts(s, \lev(s)+1) \ge c_s$, we have $\wts(s) \ge c_s$, so $s$ is tight regardless of the value of $\phi(s)$. Thus, we zero out $\phi(s)$. Then, we set $\lev(s) \gets \min\{\bs(s), \lev(s)\}$. Now there can be passive elements $e' \in s$ such that $\ilev(e') = \lev(s)$, so we need to make them active (note that there are no elements with $\ilev(e') < \lev(s)$, since $\ilev(e')$ must be at least $\bs(s)$). To activate elements in $P_{\lev(s)}(s)$, go over each $e'\in P_{\lev(s)}(s)$ and set $\zlev(e') \gets \lev(s)$, and for every set $s' \ni e'$, move $e'$ from $P_{\lev(s)}(s')$ to $A_{\lev(s)}(s')$.

\subparagraph{Raising $s$ by one level.} At this point, \Cref{inv}(1) may be still violated for $s$. Thus, we repeatedly increase the level of $s$ by one in a while loop, until $\wts(s, \lev(s)+1) < c_s$ is satisfied. During each iteration, the algorithm does the following steps:

\begin{enumerate}[(1),leftmargin=*]
    \item Set $\phi(s) \gets 0$; again, we have $\wts(s) \ge c_s$, so we can safely zero out $\phi(s)$ without breaking the tightness.
    Define $k = \lev(s)$ and increase $\lev(s)$ by one. 

    \item \textbf{Raising elements from $A_k(s)$:} The increase of $\lev(s)$ leads to the increase of levels of elements from $A_k(s)$. Thus, we go over all elements $e' \in A_k(s)$ and raise them to level $k + 1$. To do so, set $\ilev(e'), \zlev(e') \leftarrow k + 1$ (making $\wts(e') = (1 + \epsilon)^{-k - 1}$), and for each $s' \ni e'$, move $e'$ from $A_k(s')$ to  $A_{k + 1}(s')$ and update $\wts(s')$ by decreasing it by $\epsilon(1 + \epsilon)^{-k-1}$.

    After that, for each $s^\prime\ni e^\prime, s^\prime\neq s$, to restore tightness on set $s^\prime\ni e^\prime$, we increase its dead weight $\phi(s^\prime)$ by $\epsilon(1+\epsilon)^{-k-1}$. 

    \item \textbf{Raising $e$:} If $\zlev(e) = k$, then $\lev(e)$ becomes $k + 1$ due to the increase of $\lev(s)$, so we need to update the levels of $e$. To do it, increase both $\zlev(e)$ and $\ilev(e)$ by one and update $\wts(s)$ accordingly, by setting $\wts(s) \leftarrow \wts(s) - \epsilon(1 + \epsilon)^{-k-1-d}$. 
    
    \item \textbf{Activating passive elements:} Since we have increased $\lev(s)$ by one, for elements $e' \in P_{\lev(s)}(s)$ we now have $\ilev(e) = \lev(e)$, so we need to activate them. We do it the same way we did when raising $s$ to level $\min\{\bs(s),\zlev(e)\}$.
\end{enumerate}
After the while loop terminates, the algorithm keeps a record of the current value of $l_s\leftarrow \ilev(e)$ for the current set $s$ to update $\wts(s)$ to the actual value in the end.

\paragraph{Finalizing the changes.} When all sets in $s\ni e$ have been enumerated, since we are being lazy on updating the contribution of $e$ whenever $\ilev(e)$ changes, we need to update the relevant data structures accordingly. To do so, go over all sets $s \ni e$ again to update the weight by setting $\wts(s)\leftarrow \wts(s) - (1+\epsilon)^{-l_s} + \wts(e)$. This update of $\wts(s)$ could decrease it, so $s$ might have become slack. Hence, to restore the tightness for $s$, we increase its dead weight $\phi(s)$ by $(1+\epsilon)^{-l_s} -\wts(e)$ if $s \in F$. We will show later that if $s \notin F$, then we did not enter the while loop for $s$, and so every loss of $\wts(s)$ due to raises of active elements was compensated by the increase of $\phi(s)$, and we never decreased $\phi(s)$. Since we are maintaining $\ilev(e) \le \ilev^\old(e)$, the change of $\wts(e)$ could only increase $\wts(s)$. Therefore, if $s$ was tight, it remains tight.
Finally, add $e$ to $P_{\ilev(e)}(s)$.

\begin{algorithm}
    \caption{$\fix(e, l)$}\label{alg:fix}
    \SetKw{Let}{let}
    \tcp{We assume that $\lev^\old(e) < l \le \ilev^\old(e)$ and $\zlev(e) = \lev(e)$}
    $\ilev(e) \gets h$\;
    \ForEach{$s \ni e$}{
        $\wts(s) \gets \wts(s) -\wts^\old(e) + \wts(e)$\;
    }
    remove $e$ from $P_{\ilev^\old(e)}(s)$ for each $s \ni e$ if $e$ is not a freshly inserted element\;
    \Let{$F = \{s\ni e\mid \wts(s) \ge c_s \}$}\;
    \Let{$d = \ilev(e) - \zlev(e)$}\;
    \ForEach(\tcp*[f]{raising sets $s \ni e$}){$s\ni e$}{\label{ln:fix-main-for}
        update the weight $\wts(s) \gets \wts(s) - (1+\epsilon)^{-l} + \wts(e)$\; \tcp{refresh $\wts(s)$ according to up-to-date $\wts(e)$}
        \If{$\wts(s, \lev(s)+1) \geq c_s$ {\bf and} $\lev(s) < \min\{\bs(s),\zlev(e)\}$}{\label{ln:fix-below}
            $\phi(s) \gets 0$\;
            $\lev(s) \gets \min\{\bs(s),\zlev(e)\}$\;
            activate all passive elements in $P_{\lev(s)}(s)$\; \label{ln:activate-below}
        }
        \While{$\wts(s, \lev(s)+1) \geq c_s$}{\label{ln:fix-while}
            $\phi(s) \gets 0$\;\label{ln:zero-phi}
            \Let{$k = \lev(s)$}\;
            $\lev(s) \gets k+1$\;
            \ForEach{$e^\prime \in A_k(s)$}{
                raise $e^\prime$ to level $k + 1$\; \label{ln:raise-active}
                \ForEach{$s^\prime \ni e^\prime$, $s^\prime \ne s$}{
                    $\phi(s^\prime) \gets \phi(s^\prime) + \epsilon(1+\epsilon)^{-k-1}$\;
                }
            }
            \If{$\zlev(e) = k$}{
                assign $\zlev(e) \gets k + 1$, $\ilev(e) \gets k + 1 + d$\; 
                $\wts(s) \gets \wts(s) - \epsilon(1+\epsilon)^{-k-1-d}$\;
            }
            activate all passive elements in $P_{k + 1}(s)$\; \label{ln:activate-k+1}
        }
        $l_s \gets \ilev(e)$\;
    }
    \ForEach(\tcp*[f]{finalizing changes}){$s \ni e$}{\label{ln:fix-final-for}
        $\wts(s) \gets \wts(s) - (1+\epsilon)^{-l_s} + \wts(e)$\; \tcp{refresh $\wts(s)$ according to up-to-date $\wts(e)$}
        \If{$s\in F$}{
            $\phi(s) \gets \phi(s) + (1+\epsilon)^{-l_s} - \wts(e)$\;
        }
        add $e$ to $P_{\ilev(e)}(s)$\;
    }
\end{algorithm}

\begin{figure}
    \centering
    \begin{tikzpicture}
  [thick,scale=0.7,
  active/.style={circle, draw, fill=black, inner sep=0pt, minimum width=6pt},
  passive/.style={circle, draw, fill=white, inner sep=0pt, minimum width=6pt},
  shadow/.style={rectangle, draw, pattern=north east lines, inner sep=0pt, minimum width=6pt, minimum height=6pt}],
	\begin{scope}
		\draw (2, 2) node(1)[active, label=30: {$e_1$}] {};
		\draw (3.5, 3) node(2)[passive, label=30: {$e_2$}] {};
		\draw (6.5, 4) node(3)[passive, label=30: {$e_3$}] {};
		\draw (6.5, 3) node(4)[passive, label=30: {$e_4$}] {};
		\draw (8, 2) node(5)[active, label=30: {$e_5$}] {};
    	\draw (9.5, 2) node(6)[active, label=30: {$e_6$}] {};
		\draw (5, 7) node(e)[passive, label=-30: {$e$}] {};
		\draw (5, 2) node(ze)[shadow, label=30: {$e$}] {};
	\end{scope}

	\begin{scope}[xshift=11cm]
		\draw (2, 3) node(1_2)[active, label=30: {$e_1$}] {};
        \node (1_2_shadow) at (2, 2) {};
		\draw (3.5, 3) node(2_2)[active, label=30: {$e_2$}] {};
		\draw (6.5, 4) node(3_2)[passive, label=30: {$e_3$}] {};
		\draw (6.5, 3) node(4_2)[active, label=30: {$e_4$}] {};
		\draw (8, 3) node(5_2)[active, label=30: {$e_5$}] {};
    	\draw (9.5, 2) node(6_2)[active, label=30: {$e_6$}] {};
		\draw (5, 8) node(e_2)[passive, label=-30: {$e$}] {};
		\draw (5, 3) node(ze_2)[shadow, label=30: {$e$}] {};
	\end{scope}
	
	\draw (0, 10) node {\textbf{level}};
	\draw (-1.5, 2) node {$k$};
	\draw (-1.5, 3) node {$k + 1$};
	\draw (-1.5, 4) node {$k + 2$};
	\draw (-1.5, 5.5) node {$\vdots$};
	\draw (-1.5, 7) node {$k + d$};
	\draw (-1.5, 8) node {$k + d + 1$};

	\draw (5, -0.5) node {Before};
	\draw (15, -0.5) node {After};
    
	\begin{scope}[on background layer]
		\draw [line width = 0.5mm] (-2, 0) to (22, 0);
		\draw [->, line width = 0.5mm] (0, -1) to (0, 9.5);
		\draw [line width=0.5mm, decorate, decoration=zigzag] (11, 0) to (11, 9.5);

		\draw [gray, dashed, thick] (0, 2) to (22, 2);
		\draw [gray, dashed, thick] (0, 3) to (22, 3);
		\draw [gray, dashed, thick] (0, 4) to (22, 4);
		\draw [gray, dashed, thick] (0, 7) to (22, 7);
		\draw [gray, dashed, thick] (0, 8) to (22, 8);

		\begin{scope}[label distance=1mm]
			\colorlet{darkgreen}{green!60!black};
			\node[draw=red, rounded corners, thick, fill=red, fill opacity=0.1, inner ysep=2.5mm, fit=(1) (2) (3) (4) (5) (e), label={[red]below:$s$}] {};
    		\node[draw=blue, rounded corners, thick, fill=blue, fill opacity=0.1, fit=(3) (4) (5) (6) (e), label={[blue]below:$s_2$}] {};
			\node[draw=darkgreen, rounded corners, thick, fill=green, fill opacity=0.1, inner xsep=2mm, fit=(1) (2), label={[darkgreen]below:$s_1$}] {};

			\node[draw=red, rounded corners, thick, fill=red, fill opacity=0.1, inner ysep=2.5mm, fit=(1_2) (2_2) (3_2) (4_2) (5_2) (e_2), label={[red]below:$s$}] {};
    		\node[draw=blue, rounded corners, thick, fill=blue, fill opacity=0.1, fit=(3_2) (4_2) (5_2) (6_2) (e_2), label={[blue]below:$s_2$}] {};
			\node[draw=darkgreen, rounded corners, thick, fill=green, fill opacity=0.1, inner xsep=2mm, fit=(1_2) (2_2) (1_2_shadow), label={[darkgreen]below:$s_1$}] {};
		\end{scope}

		\draw [line width = 0.2mm, style={decorate, decoration=snake}] (e) to node[right] {$d$} (ze);
		\draw [line width = 0.2mm, style={decorate, decoration=snake}] (e_2) to node[right] {$d$} (ze_2);
	\end{scope}
\end{tikzpicture}
    \caption{An illustration of a single iteration of the while loop of the $\fix(e, l)$ subroutine,
    which raises set $s$ by one level.
    Depicted by black circles are three active elements $e_1, e_5, e_6$, and  depicted by white circles are four passive elements $e, e_2, e_3, e_4$ (note that $e$ is passive).
    The small dashed rectangle represents $\zlev(e)$, which is equal to $\lev(e)$. There are three sets: $s = \{e, e_1, e_2, e_3, e_4, e_5\}$ colored red, $s_1 = \{e_1, e_2\}$ colored green, and $s_2 = \{e, e_3, e_4, e_5, e_6\}$ colored blue. The level of all three sets is $k$. The left and right parts of the figure illustrate the states right before and after the iteration, respectively. During the iteration, $e_1$ and $e_5$ were raised from level $k$ to level $k + 1$ and they remain active. Elements $e_2$ and $e_4$ became active, since they were passive elements at level $k + 1$, and since they belong to $s$. Since the lazy level of $e$ was $k$, both its lazy and intrinsic levels got raised by one, so the gap between them remains equal to $d$. The level of $s$ was raised to $k + 1$.}
    \label{fixlevel-fig}
\end{figure}

\subsubsection{Key properties}\label{sec:fix-properties}
Recall that for any call to $\fix(e, l)$, we assume \Cref{inv}(1) holds, for every passive element we have $\lev(e')< \ilev(e')\leq \zlev(e') + \left\lceil\log_{1+\epsilon}\max\{f, \frac{2C}{\epsilon} \}\right\rceil$, and that $\zlev(e) = \lev(e)$ and $\lev^\old(e) < l \le \ilev^\old(e)$. 
We use the super-script ``$\old$'' to denote the values of the variables right before the execution of $\fix(e, l)$ started (e.g., $\lev^\old(s), \wts^\old(s)$). Before we proceed to proving the properties after the call to $\fix(e, l)$, we state some auxiliary observations and claims about what happens during the execution of it.

\begin{observation}\label{obs:fix-monotonicity}
    For any set $s'$ and any element $e'$, the values of $\lev(s')$, $\ilev(e')$ and $\zlev(e')$ can only increase, and the values of $\wts(s')$ and $\wts(e')$ can only decrease during the raising sets $s \ni e$ and finalizing changes steps. 
\end{observation}

\begin{observation}\label{obs:fix-no-changes}
    If $\ilev^{\old}(e) = l$, then the call to $\fix(e, l)$ does not make any changes.
\end{observation}

\begin{observation}\label{obs:fix-gap}
    After raising a set $s$ to level $\min\{\bs(s),\zlev(e)\}$ or raising $s$ by one level, we have $\ilev(e) - \zlev(e) = d$ and $\zlev(e) = \lev(e)$.
\end{observation}

\begin{claim}\label{fix-enter-F}
    For any set $s \notin F$, we have $\wts(s, \lev(s) + 1) < c_s$ during the raising sets $s \ni e$ step.
\end{claim}
\begin{proof}
    If $s \notin F$, then $\wts(s) < c_s$ after setting the intrinsic level of $e$ to $l$. By \Cref{obs:fix-monotonicity}, $\wts(s)$ could only decrease after that, so the entering condition $\wts(s, \lev(s) + 1) \ge c_s$ is never satisfied.
\end{proof}

\begin{observation}\label{fix-lev-inc}
    $\lev(s)$ strictly increases after each iteration of the while loop.
\end{observation}

\begin{claim}\label{fix-empty}
    During the call to $\fix(e, l)$,
    \[\bigcup_{i=0}^{\bs(s)-1}A_i(s) ~=~\bigcup_{i=0}^{\bs(s)-1}P_i(s)\setminus \{e\} ~=~ \emptyset.\]
\end{claim}
\begin{proof}
    Suppose for contradiction that there is an element $e' \in s, e' \ne e$ with $\ilev^\old(e') \le \bs(s) - 1$. Thus, $\wts^\old(e') \ge (1 + \epsilon)^{-\bs(s) + 1}$. But then $\wts^\old(s, \lev^\old(s) + 1) \ge \wts^\old(e') / (1 + \epsilon) \ge (1 + \epsilon)^{-\bs(s)} \ge c_s$, which violates \Cref{inv}(1)~--- a contradiction. Thus, \[\bigcup_{i=0}^{\bs(s)-1}A^\old_i(s) =\bigcup_{i=0}^{\bs(s)-1}P^\old_i(s)\setminus \{e\}= \emptyset.\]
    During the call, any element level is non-decreasing by Observation \ref{obs:fix-monotonicity}; therefore, 
    \[\bigcup_{i=0}^{\bs(s)-1}A_i(s) =\bigcup_{i=0}^{\bs(s)-1}P_i(s)\setminus \{e\}= \emptyset.\]
\end{proof}

Next, we show properties after the call to $\fix(e, l)$. First, we state the following observation.

\begin{observation}\label{obs:fix-main}
    After the call to $\fix(e, l)$, the following holds:
    \begin{enumerate}[(1)]
        \item Element $e$ is passive with $\zlev(e) = \lev(e)$, and the gap $\ilev(e) - \zlev(e)$ is equal to $d = l - \lev^\old(e)$.
        \item For any passive element $e^\prime \ne e$, either $\ilev(e^\prime)$ and $\zlev(e^\prime)$ remain the same (and hence the gap $\ilev(e^\prime) - \zlev(e^\prime)$ remains the same), or $e^\prime$ becomes active with $\ilev(e^\prime) \ge \ilev^\old(e^\prime)$.
        \item Levels of sets $s \in e$ could have only increased; levels of other sets remain the same.
        \item For any set $s^\prime$ such that $e \notin s^\prime$, the value of $\wts(s^\prime)$ could have only decreased.
    \end{enumerate}
\end{observation}

Recall that for any passive element $e^\prime$ we must have $\zlev(e^\prime) \le \lev(e^\prime)$. Since $\fix$ can be invoked from the $\reset$ subroutine, where the inequality $\zlev(e^\prime) \le \lev(e^\prime)$ may be temporarily violated, we do not assume that it holds before the call to $\fix(e, l)$. However, we still want to maintain $\lev(e')< \ilev(e')\leq \zlev(e') + \left\lceil\log_{1+\epsilon}\max\{f, \frac{2C}{\epsilon} \}\right\rceil$, which we show we do in the following claim.

\begin{claim}\label{fix-lev<ilev}
    After the call to $\fix(e, l)$, for any passive element $e'$, we have $\lev(e')< \ilev(e')\leq \zlev(e') + \left\lceil\log_{1+\epsilon}\max\{f, \frac{2C}{\epsilon} \}\right\rceil$.
\end{claim}
\begin{proof}
    For element $e$ this follows from \Cref{obs:fix-main}(1). Recall that we have $\lev^\old(e')< \ilev^\old(e')\leq \zlev^\old(e') + \left\lceil\log_{1+\epsilon}\max\{f, \frac{2C}{\epsilon} \}\right\rceil$ for every passive element $e' \ne e$. Observe that $\lev(e')$ can change only when $\lev(s)$ changes for some $s \ni e$, which can happen only during raising $s$ to level $\min\{\bs(s),\zlev(e)\}$ or raising $s$ by one level. In the former case, we have $\ilev(e') \ge \bs(s)$ by \Cref{fix-empty}, and if $\ilev(e')$ becomes equal to $\bs(s)$, then $e'$ gets activated at \cref{ln:activate-below}. In the latter case, $\ilev(e')$ can become equal to $\lev(e')$ due to the increase of $\lev(s)$ by one; in that case, $e'$ gets activated as well at \cref{ln:activate-k+1}.
\end{proof}

The $\fix$ subroutine does not explicitly maintain \Cref{inv}(2). However, we stated that our goal is to make sure that tight sets remain tight and that slack sets are not raised. The following claims show that.

\begin{claim}\label{slack-no-rise}
    If $d\geq \log_{1+\epsilon}\frac{2C}{\epsilon}$, then $\lev(s)$ has not changed for any slack set $s$ after the call to $\fix(e, l)$.
\end{claim}
\begin{proof}
    Since $s$ was slack, we have $\wts^{\old}(s) < c_s / (1+\epsilon)$. Note that $l \ge d$. Therefore, after changing the intrinsic level of $e$ to $l$, we have \[\wts(s) < \frac{c_s}{1+\epsilon} + (1+\epsilon)^{-l} \leq c_s.\]
    The last inequality holds since $l \ge d\geq \log_{1+\epsilon}\frac{2C}{\epsilon} \ge \log_{1+\epsilon}\frac{2}{\epsilon c_s} \ge \log_{1+\epsilon}\frac{1+\eps}{\epsilon c_s}$, for $\eps \le 1$.

    Therefore, $s \notin F$. By \Cref{fix-enter-F}, we have $\wts(s, \lev(s) + 1) < c_s$, so we do not raise $s$.
\end{proof}

\begin{claim}\label{fix-lev+1-e}
    During the execution of $\fix(e, l)$, for any set $s$, the value of $\wts(s, \lev(s) + 1) - \wts(e)$ can only decrease.
\end{claim}
\begin{proof}
    By \Cref{obs:fix-gap}, we have $\ilev(e) > \lev(e)$ during the execution, so from \Cref{eqlevel}: \[\wts(s, \lev(s) + 1) - \wts(e) = \sum_{e' \in s, e' \ne e} \min\left\{\wts(e'), (1 + \epsilon)^{-\max\{\lev(s) + 1,  \max_{t\mid e'\in t\neq s} \lev(t)\}}\right\}.\]
    
    Observe that changing the intrinsic level of $e$ to $l$ does not affect $\wts(e')$ for $e' \ne e$, so $\wts(s, \lev(s) + 1) - \wts(e)$ remains the same. By \Cref{obs:fix-monotonicity}, $\lev(s)$ can only increase, and for every $e' \in s$, $\wts(e')$ can only decrease, which can only decrease $\wts(s, \lev(s) + 1) - \wts(e)$.
\end{proof}

\begin{claim}\label{fix-ilev-upper}
    During the execution of $\fix(e, l)$, we always have $\ilev(e)\leq \ilev^\old(e)$. 
\end{claim}
\begin{proof}
    After changing the intrinsic level of $e$ to $l$, we have $\ilev(e) = l \le \ilev^\old(e)$. Observe that $\ilev(e)$ can increase only during an iteration of the while loop, and can only increase by one.
    Suppose for contradiction that at the beginning of some iteration of the while loop we had $\ilev(e) = \ilev^\old(e)$, and thus $\wts(e) = \wts^\old(e)$. By \Cref{fix-lev+1-e}, $\wts^\old(s, \lev^\old(s) + 1) - \wts^\old(e) \ge \wts(s, \lev(s) + 1) - \wts(e)$, so $\wts^\old(s, \lev^\old(s) + 1) \ge \wts(s, \lev(s) + 1)$. By \Cref{inv}(1), $\wts^\old(s, \lev^\old(s) + 1) < c_s$, thus $\wts(s, \lev(s) + 1) < c_s$, which contradicts the entering condition of the while loop.
\end{proof}

\begin{claim}\label{fix-tight}
    Any tight set remains tight after the call to $\fix(e, l)$.
\end{claim}
\begin{proof}
    First, observe that during changing the intrinsic level of $e$ to $l$, weights of sets can only increase, since $l \le \ilev^\old(e)$. Thus, this does not break tightness for any set.

    Consider a tight set $s$. Observe that the tightness for it can be violated only during raising $s$ (the for loop at \cref{ln:fix-main-for}) or during the finalization step for $s$ (the for loop at \cref{ln:fix-final-for}). This is so, because during raising some other set $s'$, $\wts(s)$ could change only due to raising elements from $A_k(s')$. However, in that case, the decrease of $\wts(s)$ is immediately compensated by the increase of $\phi(s)$. Therefore, if $e \notin s$, then $s$ remains tight after the call to $\fix(e, l)$.

    Next, consider the case $e \in s$. If  we do not raise $s$ to level $\min\{\bs(s),\zlev(e)\}$ or enter the while loop, then $\wts^*(s)$ could decrease only due to changes of $\wts(e)$. However, we have $\wts(e) \ge \wts^\old(e)$ by \Cref{fix-ilev-upper}, so $\wts^*(s)$ is no smaller than it was at the beginning of the call. Therefore, $s$ remains tight.
    
    If we do not enter the while loop, but raise $s$ to level $\min\{\bs(s),\zlev(e)\}$, then observe that we have $\wts(s, \lev(s) + 1) \ge c_s$, and hence $\wts(s) \ge c_s$. If we enter the while loop for $s$, consider the last iteration of it. In that case, we also have $\wts(s) \ge c_s$ at the beginning of that iteration. During it, elements from $A_k(s)$ are raised one level up, and $e$ may also raise one level up, thus decreasing their weights by a factor of $1 + \epsilon$. Then, $\wts(s)$ decreases by at most a factor of $1 + \epsilon$, so $\wts(s) \ge c_s / (1 + \epsilon)$ as a result of this iteration. As we have shown before, $\wts^*(s)$ remains the same during the subsequent iterations. Since we entered the while loop, $s$ must be in $F$, according to \Cref{fix-enter-F}. Observe that this implies that at the finalization step, the decrease of $\wts(s)$ is compensated by the increase of $\phi(s)$, so $s$ remains tight.
\end{proof}

Finally, recall that $\fix$ aims to maintain \Cref{inv}(1). The following theorem shows that this indeed holds.

\begin{theorem}\label{fix-invs}
    \Cref{inv}(1) holds after the call to $\fix(e, l)$.
\end{theorem}
\begin{proof}
    After changing the intrinsic level of $e$ to $l$, \Cref{inv}(1) could be violated only for sets containing $e$. The algorithm raises each set $s \ni e$ until the condition of \Cref{inv}(1) is satisfied for $s$. By \Cref{obs:fix-monotonicity}, the weights of sets can only decrease. Therefore, if the condition of \Cref{inv}(1) was satisfied for a set at some point during the call, it is satisfied until the end of the call. It follows that \Cref{inv}(1) holds for every set after the call to $\fix(e, l)$. 
\end{proof}
We are ready to prove \Cref{ins-inv}, which shows the correctness of the $\ins$ subroutine, as was stated in \Cref{subsec:insert}.
\insinv*
\begin{proof}
    We start by analyzing the case $F = \emptyset$. The analysis splits into two subcases.
    \begin{itemize}[leftmargin=*]
        \item The first subcase is that $F = \emptyset$ and $h>\lev(e)$. Let $\wts^{\old}(s)$ be the weight of a set $s$ before the insertion of $e$. By the minimality of $h$, there exists $s \ni e$ such that $\wts^{\old}(s) + (1 + \epsilon)^{-h + 1} \ge c_s$. Dividing both sides by $1 + \epsilon$, we get $\frac{1}{1 + \epsilon}\wts^{\old}(s) + (1 + \epsilon)^{-h} \ge \frac{c_s}{1 + \epsilon}$. The left-hand side is at most $\wts(s) = \wts^{\old}(s) + (1 + \epsilon)^{-h}$. Therefore, $\wts(s) \ge c_s/(1 + \epsilon)$, so $s$ is tight and the new element $e$ is covered by $T$. 

        \item The second subcase is that $F = \emptyset$ and $h = \lev(e)$. If $\lev(e) > 0$, then $e$ is covered by a set $s$ with $\lev(s) > 0$. This set is tight by \Cref{inv}(2). Otherwise, $\wts(e) = 1$ and thus any set $s \ni e$ would have $\wts(s) \ge 1$, and so is tight, since $c_s \le 1$.
    \end{itemize}

    In both subcases, \Cref{inv}(1)(2) hold, since weights do not decrease, and $\wts(s) < c_s$ for all $s\ni e$.

    Next, let us consider the case where $F \neq \emptyset$. For any $s\in F$, before $e$ was inserted, we had $\wts(s) \ge c_s - (1+\epsilon)^{-l} \ge c_s - \frac{\epsilon}{2C} \geq c_s / (1+\epsilon)$. Hence, all sets in $F$ are tight before we invoke $\fix(e, l)$. By \Cref{fix-tight}, sets from $F$ are still tight after the call to $\fix(e, l)$, so $e$ is covered by a tight set. \Cref{inv}(1) holds by \Cref{fix-invs}. To show that \Cref{inv}(2) is also maintained, consider a set $s$. If it was slack before the insertion, we had $\lev(s) = 0$. After the call to $\fix(e, l)$, it remains at level 0 by \Cref{slack-no-rise}. If $s$ was tight before the insertion, then $s$ remains tight by \Cref{fix-tight}.
\end{proof}

\subsection{Rebuilding} \label{subsec:rebuild}
The $\reset$ subroutine is invoked whenever \Cref{inv}(3) is violated. Let $k$ be the smallest index such that $\phi_{\leq k} > \epsilon\left(c(T_{\leq k}) + f\cdot \wts(E_{\leq k})\right)$; by \Cref{linkedlist}, such an index $k$ can be found in time
\[O\left(\frac{\log C}{\epsilon} + \left|T_{\le k} \setminus T_{\leq \ceil{\log_{1+\epsilon}C}+1}\right| + \left|E_{\leq k}\right|\right).\]

\paragraph{High-level idea.} In the $\reset$ procedure we want to eliminate the dead weights of sets up to level $k$, and then fix the cover and violated invariants by changing levels of elements and sets. If all elements were active (which is the case in the rebuild subroutine of \cite{bhattacharya2021dynamic}), we could just raise every set and element below level $k + 1$ to level $k + 1$, and then invoke the subroutine $\water$ for them (refer to \Cref{waterfill}), which would restore the cover and the invariants by lowering their levels. 

However, passive elements pose additional challenges. First, there can be a passive element $e$ with $\zlev(e) \le k$ and $\ilev(e) > k + 1$, which we call a ``dirty'' element. If we just apply $\water$, without considering them, that will cause issues. Invoking $\water$ does not guarantee that every set becomes tight, so it might be the case that after the execution of $\water$, every set that contains $e$ is slack, so $e$ is not covered by $T$. Another problem is that we want to bound the time incurred by all the instances of $\reset$ on any passive element $e$. To do so, we would like to argue that the gap $\ilev(e) - \zlev(e)$ decreases after each call to $\reset$. But even if $e$ is contained in a tight set after the call to $\water$, its $\lev(e)$ could decrease as a result, which could force the decrease of $\zlev(e)$ (since we maintain $\lev(e) \ge \zlev(e)$), and hence increase the gap $\ilev(e) - \zlev(e)$.

To overcome these problems, we do the following. First, we raise every set with $\lev(s) \le k$ to level $k + 1$. After that, we raise every ``clean'' element to level $k + 1$; that is, active elements with $\lev(e) \le k$ and passive with $\ilev(e) \le k + 1$ (note that we need to make passive elements active). Then we raise $\zlev(e)$ to $k + 1$ for every ``dirty'' element, thus decreasing the gap $\ilev(e) - \zlev(e)$. Raising clean elements could decrease weights of some sets. As a result, some elements could stop being covered by a tight set. For clean elements, the cover and the invariants will be repaired by invoking $\water$ on them. 
To repair the cover for dirty elements, we
try to decrease $\ilev(e)$ as much as possible if $e$ is not already covered by a tight set, while maintaining $\wts(s) < c_s$ for each $s \ni e$ (and make $e$ active, if necessary). However, it may be still not enough to fix the cover for $e$. In that case, we activate $e$ on level $k + 1$, we add it to the set of ``clean'' elements.  
Finally, we apply $\water$ on the set of clean elements, to fix the cover and the invariants for sets that contain them.

However, this approach is not enough by itself, since we cannot afford to go over all sets below level $k+1$ to raise them to level $k+1$. To overcome this hurdle, we observe that many sets are ``unnecessary'' to restore a valid cover, and would be dropped to level $0$ after applying $\water$ anyway. Hence we can drop every set below level $k+1$ to level 0 efficiently using the aforementioned idea of \emph{implicit zeroing} (see \Cref{basicds}), and after that $\water$ will be applied only on a subset of ``necessary'' sets below level $k+1$, which we show how to efficiently find.

To bound the runtime, our argument relies on the decrease of the gap $\ilev(e) - \zlev(e)$ after each instance of $\reset$.
So far, processing a dirty element required scanning all sets $s \ni e$, and the gap $\ilev(e) - \zlev(e)$ decreased by at least one as a result.
To improve the runtime further, we would like to handle dirty elements more efficiently. 
Our idea here is to try to avoid scanning all the sets $s \ni e$, which takes time $O(f)$. This might be not always possible; in that case, we try to decrease the gap $\ilev(e) - \zlev(e)$ significantly. Suppose we want to increase $\wts(e)$ up to some $\delta$ without violating $\wts(s) < c_s$ for any $s \ni e$. There may be some witness set $s$ for which $\wts(s) - \wts(e) + \delta \ge c_s$. If $\delta$ is close enough to $\wts(e)$, namely, $\delta - \wts(e) \le \epsilon \cdot c_s / (1 + \epsilon)$, then $s$ must be tight, since $\wts(s) \ge c_s + \wts(s) - \delta \ge c_s / (1 + \epsilon)$. Therefore, if there is such a set, then $e$ is covered by a tight set. If there are many witness sets, we can find one by sampling sets $s \ni e$ uniformly at random, which allows us to avoid scanning all sets $s \ni e$;
upon failure, we can basically proceed as in the deterministic algorithm in this case, since the sampling fails with small probability.

The other extreme case is when there is no witness set, and then we can increase $\wts(e)$ up to at least $\delta$, and hence decrease the gap $\ilev(e) - \zlev(e)$. The interesting case is when there are a few witness sets, and so the sampling is likely to fail; however, in this case we can apply $\fix$ to decrease the gap. Recall that the amortized runtime cost of invoking $\fix(e, l)$ is roughly $|F| \cdot f \cdot (1 + \epsilon)^{-d}$, where $d = l - \lev(e)$ and $F$ is the set of sets $s \ni e$, for which $\wts(s) \ge c_s$ after setting $\ilev(e)$ to $l$. Notice we take $l$ such that $(1 + \epsilon)^{-l} \le \delta$, then every set in $F$ must be a witness set. As we will argue later, if we define $\delta$ carefully, we can achieve a significant decrease of the gap $\ilev(e) - \zlev(e)$, while making sure that the amortized runtime cost stays linear in $f$.

\subsubsection{Description of the $\reset$ subroutine}
Next, let us describe the algorithm more formally; refer to \Cref{alg:reset} for the pseudocode of the $\reset$ subroutine, and to \Cref{alg:procdet,alg:procrand} for the pseudocodes of the auxiliary $\procdet$ and $\procrand$ subroutines, which are called from within the $\reset$ subroutine.

\begin{definition}
    Consider the moment right before the call to $\reset(k)$. An element $e\in E_{\leq k}$ (i.e., $\zlev(e) \le k$) is called \emph{clean}, if $\ilev(e) \le k  + 1$, and is called \emph{dirty} otherwise.
\end{definition}

First, go over each element in $E_{\le k}$ and compute the set $D$ of all dirty elements and the set $E$ of all clean elements among them. 
We initially define set $S \leftarrow \emptyset$. During the steps of our algorithm, \Cref{inv}(2) may get violated for some sets. 
In the end of the subroutine, during the post-processing step, we will call $\water$ (see \cref{ln:call-water} of \Cref{alg:reset}), which will restore the invariants for some of the sets, but not necessarily for all sets.
Hence we try to collect all possibly affected sets to $S$, so we are able to deal with them later. 

Next, apply implicit zeroing (refer to \Cref{implicit-zeroing} for the definition) on all levels in $[0, k]$. Note that this operation may decrease $\lev(e)$ for some elements in $E_{\le k}$, without modifying $\ilev(e)$ and $\zlev(e)$, thus possibly making them incorrect. However, this will be repaired when we process clean and dirty elements.

\paragraph{Processing clean elements.} Let us first process all clean elements the following way. We move each clean element to level $k + 1$, together with all sets containing it; note that all such sets must be below level $k + 1$. More specifically, for each clean element $e$, we set $\ilev(e), \zlev(e) \gets k + 1$ and go over sets $s\ni e$ to assign $\lev(s)\leftarrow k+1$, add $s$ to $S$ and update $\wts(s)$ for each $s \ni e$ accordingly. Besides, we also need to update the sets $A_i(s)$ and $P_i(s)$ (which is omitted in the pseudocode).

This repairs the levels of clean elements, but for dirty elements, their lazy levels can be still incorrect for now (i.e., $\zlev(e) > \lev(e)$). Also, by processing clean elements, we could have only decreased the weights of sets. We will make sure that (composite) weights only increase during the subsequent steps, so sets that are tight at this moment will remain tight during the subsequent steps. 

Note that after processing clean elements, we have $\ilev(e) \ge k + 1$ for every element. During the subsequent steps, we may raise some sets to level $k + 1$, but because of that, we will not need to raise any element. 

\paragraph{Processing dirty elements.} Next, we describe how we handle dirty elements. The way we do it depends on whether our algorithm uses randomness or not, which is controlled by the global flag \textsf{deterministic}. 

First, initialize set $E' \gets \emptyset$. During processing a dirty element, it could become active at level $k + 1$. The algorithm will make sure that every set above level $k + 1$ is tight; however, sets at level $k + 1$ are not necessarily tight. Thus, we collect such elements to $E'$, so we are able to repair the cover for them later, by invoking $\water$ on them.

There are two different ways to process an element $e \in D$: one of them uses randomness, while the other does not. We define each of them in the respective subroutine. Refer to \Cref{alg:procdet} for the pseudocode of the deterministic subroutine $\procdet(e)$ and to \Cref{alg:procrand} for the randomized subroutine $\procrand(e)$. To simplify the notation, we assume that $k, E', S$ are passed to them as arguments implicitly.

To process dirty elements, go over each element $e\in D$. If the flag \textsf{deterministic} is set, then we simply call $\procdet(e)$. Otherwise, we call $\procrand(e)$ if $e$ is still passive; if not, continue to the next iteration. We need to do this, since in the randomized version, a dirty element could become active during a previous iteration. 
Refer to \Cref{fig:reset} for an illustration of processing dirty elements. 

Next, we describe the subroutines $\procdet$ and $\procrand$.

\begin{figure}
    \centering
    \begin{tikzpicture}[thick,scale=0.7]
	\draw (-8, 10.2) node[inner sep=0pt, minimum width=6pt, label=\textbf{level}] {};
	\draw (-9.2, 5.1) node[inner sep=0pt, minimum width=6pt, label=$k+1$] {};
	\draw (-9.2, 4) node[inner sep=0pt, minimum width=6pt, label=$k$] {};
	
	\draw (-6, 3) node(1)[circle, draw, fill=black, inner sep=0pt, minimum width=6pt, label=0: {$e_1$}] {};
	\draw (-6, 7) node(2)[circle, draw, fill=gray, inner sep=0pt, minimum width=6pt, label=0: {$e_1$}] {};
	
	\draw (-9.2, 2.5) node[inner sep=0pt, minimum width=6pt, label=$\zlev(e_1)$] {};
	\draw (-9.2, 6.5) node[inner sep=0pt, minimum width=6pt, label=$\ilev(e_1)$] {};
	
	\draw (-3, 2) node(3)[circle, draw, fill=black, inner sep=0pt, minimum width=6pt, label=0: {$e_2$}] {};
	\draw (-3, 8) node(4)[circle, draw, fill=gray, inner sep=0pt, minimum width=6pt, label=0: {$e_2$}] {};
	
	\draw (-9.2, 1.5) node[inner sep=0pt, minimum width=6pt, label=$\zlev(e_2)$] {};
	\draw (-9.2, 7.5) node[inner sep=0pt, minimum width=6pt, label=$\ilev(e_2)$] {};
	
	\draw (-0, 1) node(5)[circle, draw, fill=black, inner sep=0pt, minimum width=6pt, label=0: {$e_3$}] {};
	\draw (-0, 9) node(6)[circle, draw, fill=gray, inner sep=0pt, minimum width=6pt, label=0: {$e_3$}] {};
	
	\draw (-9.2, 0.5) node[inner sep=0pt, minimum width=6pt, label=$\zlev(e_3)$] {};
	\draw (-9.2, 8.5) node[inner sep=0pt, minimum width=6pt, label=$\ilev(e_3)$] {};
	
	\draw (2, 5.5) node[circle, draw, fill=black, inner sep=0pt, minimum width=6pt, label=0: {$e_4$}] {};
	\draw (3.5, 5.5) node[circle, draw, fill=black, inner sep=0pt, minimum width=6pt, label=0: {$e_5$}] {};
	\draw (5, 5.5) node[circle, draw, fill=black, inner sep=0pt, minimum width=6pt, label=0: {$e_6$}] {};
	\draw (6.5, 5.5) node[circle, draw, fill=black, inner sep=0pt, minimum width=6pt, label=0: {$e_7$}] {};
	\draw (8, 5.5) node[circle, draw, fill=black, inner sep=0pt, minimum width=6pt, label=0: {$e_8$}] {};

	\begin{scope}[on background layer]
		\draw [line width = 0.5mm] (-10, 0) to (10, 0);
		\draw [->, line width = 0.5mm] (-8, -1) to (-8, 10);
		\draw [dashed, line width = 0.5mm, color=red] (-10, 5) to (10, 5);
		
		\draw [dashed, line width = 0.4mm, color=orange] (1) to (-8, 3);
		\draw [dashed, line width = 0.4mm, color=orange] (2) to (-8, 7);
		\draw [dashed, line width = 0.4mm, color=orange] (3) to (-8, 2);
		\draw [dashed, line width = 0.4mm, color=orange] (4) to (-8, 8);
		\draw [dashed, line width = 0.4mm, color=orange] (5) to (-8, 1);
		\draw [dashed, line width = 0.4mm, color=orange] (6) to (-8, 9);
		
		\draw [line width = 0.2mm, style={decorate, decoration=snake}] (1) to (2);
		\draw [line width = 0.2mm, style={decorate, decoration=snake}] (3) to (4);
		\draw [line width = 0.2mm, style={decorate, decoration=snake}] (5) to (6);
		
		\draw [decorate,
		decoration = {brace}] (1.5, 6.2) -- (9, 6.2);
		\draw (5, 6.5) node[label={clean elements}]{};
	\end{scope}
\end{tikzpicture}
    \caption{In this illustration, the dirty elements are  $e_1, e_2, e_3$, and the clean elements, which have been raised to level $k+1$, are $e_4, e_5, e_6, e_7, e_8$. The $\reset(k)$ procedure reduces the gaps $\ilev(e_i) - \zlev(e_i), \forall i\in \{1, 2, 3\}$. In the deterministic algorithm, the gaps decrease by at least one; in the randomized algorithm, the gaps decrease exponentially (i.e., from $d$ to $\log d$).} 
    \label{fig:reset}
\end{figure}

 \begin{algorithm}
    \caption{$\reset(k)$}\label{alg:reset}
    \SetKw{Let}{let}
    \SetKw{Continue}{continue}
    \SetKwData{Deterministic}{deterministic}
    \Let{$D$ be the set of all dirty elements, and $E$ be the set of all clean elements}\;
    $S \gets \emptyset$\;
    apply implicit zeroing for all levels in $[0, k]$\;\label{ln:impl-zeroing}
    \ForEach(\tcp*[f]{processing clean elements}){$e \in E$}{\label{ln:process-clean}
        assign $\lev(s) \gets k + 1$ and add $s$ to $S$ for each $s \ni e$\;
        \ForEach{$s \ni e$}{
            $\wts(s) \gets \wts(s) - \wts(e) + (1 + \epsilon)^{-k - 1}$\;
        }
        $\ilev(e),\zlev(e) \gets k + 1$\;
    }
    $E' \gets \emptyset$\;
    \ForEach(\tcp*[f]{processing dirty elements}){$e\in D$}{\label{ln:process-dirty}
        \If{\Deterministic}{
            $\procdet(e)$\;\label{ln:proc-call-1}
        }\Else{
            \If{$e$ is still passive}{\label{ln:check-passive}
                $\procrand(e)$\;
            }
        }
    }
    \tcp{post-processing}
    \Let{$\widehat{E}$ be the elements from $E \cup E'$ that are not contained in a tight set}\;
    set $\lev(s) \leftarrow 0$ for every slack set $s \in S$\;\label{ln:zero-slack}
    \Let{$\widehat{S}$ be the collection of all the sets that contain elements from $\widehat{E}$}\;
    \Let{$\widehat{k} = \min\left\{k+1, \ceil{\log_{1+\epsilon}\frac{2C\cdot|\widehat{E}|}{\epsilon}}\right\}$}\;
    move all elements and sets in $\widehat{E}, \widehat{S}$ to level $\widehat{k}$\;\label{ln:move-down}
    $\water(\widehat{k}, \widehat{S}, \widehat{E})$\;\label{ln:call-water}
\end{algorithm}

\paragraph{The deterministic subroutine $\procdet(e)$ (for processing an element $e \in D$).} 
In the $\procdet(e)$ subroutine, we want to raise $\zlev(e)$ to at least $k + 1$ and make sure that $e$ will be covered by a tight set. Due to the way we process clean elements, some sets on positive levels might become slack, and hence $e$ could become not covered by $T$. In that case, we can try to decrease $\ilev(e)$ (but not below $k + 1$) until some set $s \ni e$ becomes tight. If this is not possible, then we can make $e$ active at level $k + 1$ and add it to $E'$ and all sets $s \ni e$ to $S$; for them, the cover will be repaired in the end, using the $\water$ subroutine. Later in the runtime analysis, we will rely on the fact that each dirty element strictly decreases its gap $\ilev(e) - \zlev(e)$. 

First, we check if there is a tight set $s \ni e$. If there is such a set $s$, then we assign $\lev(s) \gets \max\{k + 1,\lev(s)\}$ (thus, raising it to level $k + 1$ if it was below) and set $\zlev(e)$ to $\lev(s)$, making it at least $k + 1$ and restoring $\zlev(e) \le \lev(e)$.

Otherwise, if all sets $s \ni e$ are slack, we call the $\decrease(e)$ subroutine, which tries to restore the cover for $e$ by decreasing $\ilev(e)$ as much as possible, adding $e$ to $E'$ and all sets $s \ni e$ to $S$ upon failure.

\begin{algorithm}
    \caption{$\procdet(e)$}\label{alg:procdet}
    \SetKw{Let}{let}
    \If{there is a tight set $s \ni e$}{
        set $\lev(s) \gets \max\{k + 1, \lev(s)\}$\;
        update $\zlev(e) \gets \lev(s)$\;
    }\Else{
        $\decrease(e)$\;
    }
\end{algorithm}

\subparagraph{The $\decrease(e)$ subroutine.} The $\decrease(e)$ subroutine tries to decrease $\ilev(e)$ (and hence increase $\wts(e)$) as much as possible, while maintaining $\wts(s) < c_s$ for each $s \ni e$. In order to maintain the correctness of levels, it does not decrease $\ilev(e)$ below $\lev(e)$ and $k + 1$. Because of that, it is not necessary that at least one of the sets $s \ni e$ becomes tight as a result of decreasing $\ilev(e)$. However, this can only happen when $\ilev(e)$ reaches $k + 1$. In that case, we add $e$ to $E'$ and move each set $s \ni e$ to level $k + 1$ and add it to $S$.

To simplify the notation, we assume that $k, E', S$ are passed to the subroutine implicitly. Also, we assume that $\wts(s) < c_s$ for each $s \ni e$ at the beginning of the call. Refer to \Cref{alg:decrease} for the pseudocode of the $\decrease(e)$ subroutine. The subroutine makes the following steps.

First, compute $l = \max\left\{k + 1, \max_{s\ni e}\{\lev(s)\}\right\}$, which is the lower bound for $\ilev(e)$. Find the smallest index $h \in [l, \ilev(e)]$ such that $\forall s \ni e,\, \wts(s) - \wts(e) + (1+\epsilon)^{-h} < c_s$. Note that we can find $h$ using binary search in $O(f)$ time, if $\ilev(e) \le \zlev(e) + \left\lceil\log_{1+\epsilon}\max\{f, \frac{2C}{\epsilon} \}\right\rceil$, as we argued in the description of the $\ins$ subroutine. We will later show that this condition, i.e., $\ilev(e) \le \zlev(e) + \left\lceil\log_{1+\epsilon}\max\{f, \frac{2C}{\epsilon} \}\right\rceil$ is preserved during the execution of $\reset(k)$. After that, we update $\ilev(e)$. To do it, we first set $\wts(s) \gets \wts(s) - \wts(e) + (1 + \epsilon)^{-h}$ for each $s \ni e$, and then set $\ilev(e) \gets h$, making $\wts(e) = (1 + \epsilon)^{-h}$. 

After that, all sets $s \ni e$ may be still slack. Note that this can only happen if all sets $s \ni e$ are at level $k + 1$ or below, since we make sure that sets above level $k + 1$ are tight. In that case, we add $e$ to $E'$, assign $\lev(s) \gets k + 1$ and add $s$ to $S$ for each $s \ni e$, and update $\zlev(e)$ to $\lev(e)$ by setting $\zlev(e) \gets k + 1$.

Otherwise, if there is a tight set $s \ni e$, then set $\lev(s) \gets \max\{k + 1, \lev(s)\}$ and update $\zlev(e)$ to $\lev(e)$ by setting $\zlev(e) \gets \max_{s \ni e}\{\lev(s)\}$. As a result, $\zlev(e) \ge k + 1$ after the call.

Besides, we also need to update the sets $A_i(s)$ and $P_i(s)$ (which is omitted in the pseudocode). To do it, move $e$ from $P_{\ilev^\old(e)}(s)$ to $P_{\ilev(e)}(s)$ (or $A_{\ilev(e)}(s)$ if $e$ has become active) for each $s \ni e$, where $\ilev^\old(e)$ is the value of $\ilev(e)$ at the beginning of the call.

\begin{algorithm}\caption{$\decrease(e)$}\label{alg:decrease}
    \SetKw{Let}{let}
    \Let{$l = \max\left\{k + 1, \max_{s\ni e}\{\lev(s)\}\right\}$}\;
    find the smallest index $h \in [l, \ilev(e)]$ such that $\forall s \ni e,\, \wts(s) - \wts(e) + (1+\epsilon)^{-h} < c_s$\;
    \ForEach{$s\ni e$}{
        $\wts(s) \gets \wts(s) - \wts(e) + (1+\epsilon)^{-h}$\;
    }
    $\ilev(e) \gets h$\;
    \If{all $s \ni e$ are slack}{
        add $e$ to $E'$\;
        assign $\lev(s) \gets k + 1$ and add $s$ to $S$ for each $s \ni e$\;
        $\zlev(e) \gets k + 1$.
    }\Else{
        find a tight set $s \ni e$\;
        set $\lev(s) \gets \max\{k + 1, \lev(s)\}$\;
        update $\zlev(e) \gets \max_{s\ni e}\{\lev(s)\}$\;
    }
\end{algorithm}

\paragraph{The randomized subroutine $\procrand(e)$ (for processing an element $e \in D$).} In the $\procrand(e)$ subroutine, we try to avoid spending $O(f)$ time on scanning all the sets $s \ni e$, and upon failure, we aim to decrease the gap $\ilev(e) - \zlev(e)$ significantly. To do so, we consider the set $\widehat{F} = \{ s \ni e\,|\, \wts(s) - \wts(e) + \delta \ge c_s\}$ for some $\delta > \wts(e)$ which we define later. That is, the set of sets $s \ni e$ such that $\wts(s) \ge c_s$ if we increase $\wts(e)$ to $\delta$. 

Later we will show that every set in $\widehat{F}$ is tight (see \Cref{witness-tight}). Accordingly, if $\widehat{F} \ne \emptyset$, then $e$ is already covered by $T$. However, computing $\widehat{F}$ explicitly would take time $O(f)$. We can test whether $\widehat{F} \ne \emptyset$, without computing it explicitly, by sampling sets $s \ni e$ randomly. With some probability, the random sampling may fail to find a set $s \in \widehat{F}$; in that case, we compute $\widehat{F}$ explicitly. Notice that if $\widehat{F}$ is large enough, the probability that we will proceed to computing $\widehat{F}$ is low. Otherwise, if $\widehat{F}$ is small, the expected runtime can be as large as $\Omega(f)$. However, in that case, we argue that we can decrease the gap $\ilev(e) - \zlev(e)$ significantly, so there can be only a few such ``expensive'' calls. More specifically, if $\widehat{F} = \emptyset$, we can call $\decrease(e)$, which will decrease $\ilev(e)$, so as a result, $\wts(e)$ becomes at least $\delta$ (or $e$ becomes active). Otherwise, if $\widehat{F}$ is small but not empty, we can apply the $\fix$ subroutine to decrease the gap significantly, which we argue would be cheap enough in that case (we prove that later in \Cref{lm:fix-pot-increase}).

For the randomized version, we assume $f > \frac{2C}{\epsilon}$. Otherwise, if $f \le \frac{2C}{\epsilon}$, we can simply use the deterministic version, for which the runtime bound is not worse in that case.

Recall that we check if $e$ is still passive at \cref{ln:check-passive} of \Cref{alg:reset}, before processing it. This is necessary in the randomized version of $\reset$, since $e$ could have been activated when processing some previous element from $D$ (note that in that case we do not have the problem with incorrect $\zlev(e)$ anymore). 

First, we check if $\ilev(e) - k - 1 \le \max\{\frac{200}{\epsilon^2}, 1 + 2\log_{1 + \epsilon} \frac{2C}{\epsilon}\}$. If that is the case, we fall back to calling $\procdet(e)$.

Otherwise, the $\procrand(e)$ subroutine performs the following steps. First, find the index $\eta$ such that\footnote{Let $\clog^{(\eta)}$ denote the $\clog$ function, iterated for $\eta$ times.}
\[\clog_{1+\epsilon}^{(\eta+1)}f\leq \ilev(e) - k-1\leq \clog_{1+\epsilon}^{(\eta)}f\,.\]
We will argue later in \Cref{handle-eta} that such an index exists. Note that we can find the index in constant time by computing it at the outset for every value of $\ilev(e) - k - 1$.
Next, define
\[\delta = \min\left\{\left(\frac{1}{\clog^{(\eta)}_{1+\epsilon}f}\right)^4, \brac{\frac{\epsilon}{2C}}^2\right\}\cdot (1+\epsilon)^{-k-1}\,.\] 

\subparagraph{Random sampling a witness set.} Repeatedly take uniformly random samples of sets $s\ni e$ for $50\ceil{\frac{f}{\clog_{1+\epsilon}^{(\eta)}f}}$ times, and check if $\wts(s) -\wts(e) +\delta \ge c_s$. If there exists a set $s$ for which $\wts(s) -\wts(e) +\delta \ge c_s$, then $s$ is tight (which we prove later in \Cref{witness-tight}), and hence $e$ is covered by a tight set.
To update the lazy level of $e$, assign $\lev(s)\leftarrow \max\{k+1, \lev(s)\}$ and $\zlev(e)\leftarrow \lev(s)$. After that, exit the subroutine. 

\subparagraph{Handling the element when the random sampling step fails.} Reaching this step means we have failed to find a witness set. In that case, compute the set $\widehat{F}$ of all $s\ni e$ such that $\wts(s) - \wts(e) + \delta \ge c_s$, or in other words, the set of all witness sets.
We will prove later in \Cref{delta-increase} that $\delta > \wts(e)$.
After that, the execution splits into three branches, depending on the size of $\widehat{F}$:

\begin{enumerate}[(1),leftmargin=*]
    \item If $|\widehat{F}| = 0$, then it means that we can safely increase $\wts(e)$ to $\delta$ without breaking $\wts(s) < c_s$ for any set $s \ni e$. Hence we call $\decrease(e)$.
    
    \item If $0<|\widehat{F}| \leq \left(\clog_{1+\epsilon}^{(\eta)}f\right)^2$, then we use the $\fix$ subroutine to decrease the gap $\ilev(e) - \zlev(e)$ significantly. We will later show that the amortized cost of invoking it in this case is small enough. 

    Go over each $s \ni e$ to assign $\lev(s) \leftarrow \max \{ k + 1, \lev(s) \}$ if $s$ is tight. Note that there is at least one tight set $s \ni e$, since sets in $\widehat{F}$ are tight. After that, update $\zlev(e)\leftarrow \max_{s\ni e}\{\lev(s) \}$, since $\fix$ requires $\zlev(e) = \lev(e)$. Then, define $l = \min\{\ilev(e), \zlev(e) + \widehat{d}\}$, where:
    \[\widehat{d} = \left\lceil\log_{1+\epsilon}\max\left\{\left(\clog^{(\eta)}_{1+\epsilon}f\right)^4, \brac{\frac{2C}{\epsilon}}^2\right\}\right\rceil.\]
    Invoke the subroutine $\fix(e, l)$ and apply the implicit zeroing for sets on level 0. This is because some sets at level 0 can gain dead weight due to the call to $\fix$; however, as we will show later, these sets are not needed for the cover, so we can safely zero the dead weights for them.

    \item Otherwise, $|\widehat{F}| > \left(\clog_{1+\epsilon}^{(\eta)}f\right)^2$. Recall that every set in $\widehat{F}$ is tight. Thus, we pick an arbitrary set $s \in \widehat{F}$ and assign $\lev(s)\leftarrow \max\{k+1, \lev(s)\}$ and $\zlev(e)\leftarrow \lev(s)$.
\end{enumerate}

\begin{algorithm}
    \caption{$\procrand(e)$}\label{alg:procrand}
    \SetKw{Let}{let}
    \If{$\ilev(e) - k - 1 \le \max\{\frac{200}{\epsilon^2}, 1 + 2\log_{1 + \epsilon} \frac{2C}{\epsilon}\}$}{
        $\procdet(e)$\;
    }\Else{
        find index $\eta$ such that $\clog_{1+\epsilon}^{(\eta+1)}f\leq \ilev(e) - k-1\leq \clog_{1+\epsilon}^{(\eta)}f$\;
        \Let{$\delta = \min\left\{\left(\frac{1}{\clog^{(\eta)}_{1+\epsilon}f}\right)^4, \brac{\frac{\epsilon}{2C}}^2\right\}\cdot (1+\epsilon)^{-k-1}$}\;
        \For{$50\ceil{f / \clog_{1+\epsilon}^{(\eta)}f}$ times}{
            uniformly sample $s\ni e$\;
            \If{$\wts(s) - \wts(e) + \delta \ge c_s$}{
                set $\lev(s) \gets \max\{k + 1, \lev(s)\}$\;
                update $\zlev(e) \gets \lev(s)$\;
                \Return{}
            }
        }
        compute $\widehat{F} = \{s\ni e\mid \wts(s) - \wts(e) + \delta \ge c_s \}$\;
        \If{$|\widehat{F}| = 0$}{
            $\decrease(e)$\;
        }\ElseIf{$0 < |\widehat{F}| \leq \left(\clog_{1+\epsilon}^{(\eta)}f\right)^2$}{\label{small-F}
            \Let{$\widehat{d} = \left\lceil\log_{1+\epsilon}\max\left\{\left(\clog^{(\eta)}_{1+\epsilon}f\right)^4, \brac{\frac{2C}{\epsilon}}^2\right\}\right\rceil$}\;
            assign $\lev(s) \gets \max \{ k + 1, \lev(s) \}$ for each tight $s \ni e$\;
            update $\zlev(e)\leftarrow \max_{s\ni e}\{\lev(s)\}$\;
            \Let{$l = \min\{\ilev(e), \zlev(e) + \widehat{d}\})$}\;
            $\fix(e, l)$\;
            apply implicit zeroing for sets at level 0\;
        }\Else{
            take an arbitrary $s \in \widehat{F}$\;
            set $\lev(s) \gets \max\{k + 1, \lev(s)\}$\;
            update $\zlev(e) \gets \lev(s)$\;
        }
    }
\end{algorithm}

\paragraph{Post-processing.} After we have processed elements from $D$, some elements from $E \cup E'$ may have become covered by tight sets. Consequently, as a post-processing step, we define $\widehat{E}$ as the set of elements from $E \cup E'$ not covered by a tight set. We will argue later that every set at a level greater than $k + 1$ is tight, so every element in $\widehat{E}$ is at level $k + 1$. Note that there may be sets in $S$, which are slack and at positive levels, so the condition of \Cref{inv}(2) is violated for them. To restore it, we assign $\lev(s) \leftarrow 0$ for every such a set $s \in S$.
Next, we define $\widehat{S}$ to be  the collection of all sets that contain elements from $\widehat{E}$; note that $\widehat{S}$ can be computed in $O(f|\widehat{E}|)$ time, by going over each element from $\widehat{E}$. 

For the rest, we need to use the following subroutine $\water$ from \cite{bhattacharya2021dynamic,bhattacharya2019deterministically}. 

\begin{lemma}[\cite{bhattacharya2021dynamic,bhattacharya2019deterministically}]\label{waterfill}
    There is a deterministic subroutine $\water(\widehat{k}, \widehat{S}, \widehat{E})$ which takes as input a collection of sets $\widehat{S}$ and a collection of active elements $\widehat{E}$ on level $\widehat{k}$, such that for each $s\in \widehat{S}$, we have $\wts(s)<c_s, \phi(s) = 0$. In the end, all elements in $\widehat{E}$ are still active, and the subroutine places each set $s\in \widehat{S}$ at level $\lev(s)$ such that (1) $\wts(s) < c_s, \phi(s) = 0$, and (2) if $\lev(s) > 0$ then $\wts(s) \ge \frac{c_s}{1+\epsilon}$. The runtime is $O(f|\widehat{E}| +\widehat{k})$.
\end{lemma}

We cannot directly apply $\water(k+1, \widehat{S}, \widehat{E})$, since the runtime would depend on $k$. Instead, we follow the idea from \cite{bhattacharya2021dynamic} and move all elements and sets in $\widehat{E}, \widehat{S}$ to level $\widehat{k} = \min\left\{k+1, \ceil{\log_{1+\epsilon}\frac{2C\cdot|\widehat{E}|}{\epsilon}}\right\}$. The following lemma claims that directly moving elements and sets to level $\widehat{k}$ keeps $\wts(s) < c_s$ for every set $s \in \widehat{S}$, which is required in the conditions of \Cref{waterfill}. So in the final step, we can safely invoke $\water(\widehat{k}, \widehat{S}, \widehat{E})$. 

\begin{lemma}[\cite{bhattacharya2021dynamic}]
    After moving sets and elements in $\widehat{S},\widehat{E}$ to level $\widehat{k}$, each set $s\in \widehat{S}$ has weight less than $c_s$.
\end{lemma}

\subsubsection{Key properties}
Now we will prove some properties of $\reset$ which are common for both the deterministic and the randomized versions. As we did for the $\fix$ subroutine, we use the superscript ``old'' to denote the values of the variables right before the execution of $\reset(k)$. 
Recall that we assume that at the beginning of the call, \Cref{inv}(1)(2) hold and every element is in the correct state; i.e. every active element satisfies $\ilev(e) = \zlev(e) = \lev(e)$, and every passive element satisfies $\zlev(e) \le \lev(e)$ and $\lev(e) < \ilev(e) \le \zlev(e) + \left\lceil\log_{1+\epsilon}\max\{f, \frac{2C}{\epsilon} \}\right\rceil$.

\paragraph{Properties after processing clean elements.}
First, we show some properties after processing clean elements (see the for loop at \cref{ln:process-clean} of \Cref{alg:reset}).
Similarly to $T_{\le i}$ and $E_{\le i}$, let $T_{> i}$ be the set of tight sets $s$ such that $\lev(s) > i$ and $E_{> i}$ be the set of elements $e$ such that $\zlev(e) > i$.
\begin{observation}\label{obs:clean-passive}
    During the implicit zeroing and processing clean elements, every clean element becomes active at level $k + 1$; for every other element, $\zlev(e)$ and $\ilev(e)$ remain the same.
\end{observation}

\begin{claim}\label{clean-lvl-change}
    During the implicit zeroing and processing clean elements, only sets with $\lev^\old(s) \le k$ are affected, and for them $\lev(s)$ becomes $0$ or $k + 1$.
\end{claim}
\begin{proof}
    After the implicit zeroing, all such sets are at level 0.
    Observe that only sets that contain clean elements are affected by processing clean elements. Since $\lev^\old(e) \le k$ for each clean element, such sets are initially at level $k$ or below, and all of them are raised to level $k + 1$. 
\end{proof}

Since sets at level $k + 1$ or above are not affected, we can make the following two corollaries. 

\begin{corollary}\label{clean-R-covered}
    After processing clean elements, elements from $E^\old_{> k}$ are covered by $T_{> k}$.
\end{corollary}

\begin{corollary}\label{clean-slack}
    After processing clean elements, any slack set $s$ either has $\lev(s) = 0$, or $s \in S$ and $\lev(s) = k + 1$.
\end{corollary}

\begin{claim}\label{clean-elem-ilev}
    After processing clean elements, for each element $e$, we have $\ilev(e) \ge k + 1$ if $e$ is active, and $\ilev(e) > k + 1$ if $e$ is passive.
\end{claim}
\begin{proof}
    Observe that by \Cref{obs:clean-passive}, it holds for all elements that are not clean. Every clean element $e$ becomes active and $\ilev(e)$ becomes $k + 1$ during processing clean elements. 
\end{proof}

\begin{claim}\label{clean-elements}
    After processing clean elements, for each passive element holds $\lev(e) < \ilev(e) \le \zlev(e) + \left\lceil\log_{1+\epsilon}\max\{f, \frac{2C}{\epsilon} \}\right\rceil$, and for each active element holds $\ilev(e) = \zlev(e) = \lev(e)$.
\end{claim}
\begin{proof}
    Observe that it holds for every clean element. By \Cref{obs:clean-passive}, the conditions hold for every other element, for which $\lev(e)$ has remained the same. By \Cref{clean-lvl-change}, $\lev(e)$ could have changed only for dirty elements. 
    Consider a dirty element $e$. By \Cref{clean-elem-ilev}, we have $\ilev(e) > k + 1$. By \Cref{clean-lvl-change}, if $\lev(e)$ has changed, then $\lev(e) \le k + 1$. Thus, the inequality $\lev(e) < \ilev(e)$ holds.
\end{proof}

\begin{claim}\label{clean-wts}
    After processing clean elements, we have $\wts(s) < c_s$ for every $s \notin T^\old_{> k}$.
\end{claim}
\begin{proof}
    Observe that $\wts(s) = \wts^\old(s, k + 1)$ after processing clean elements, by \Cref{eqlevel}. By \Cref{inv}(1), $\wts^\old(s, \lev^\old(s) + 1) < c_s$. Since $\lev^\old(s) \le k$, we get $\wts(s) = \wts^\old(s, k + 1) \le \wts^\old(s, \lev^\old(s) + 1) < c_s$.
\end{proof}
\begin{claim}\label{clean-below}
    After processing clean elements, we have $\wts(s) < c_s$, $\lev(s) = 0$ and $\phi(s) = 0$ for each set $s$ such that $\lev(s) < k + 1$.
\end{claim}
\begin{proof}
    By \Cref{clean-lvl-change}, only sets with $\lev^\old(s) < k + 1$ can have $\lev(s) < k + 1$ after processing clean elements. 
    For such sets, we have $\lev(s) = 0$ and $\phi(s) = 0$ after the implicit zeroing. Observe that during processing clean elements, $\lev(s)$ can only increase to $k + 1$ and $\phi(s)$ remains the same. Finally, we have $\wts(s) < c_s$ by \Cref{clean-wts}.
\end{proof}

\begin{claim}\label{clean-inv}
    \Cref{inv}(1) holds after processing clean elements.
\end{claim}
\begin{proof}
    By \Cref{clean-lvl-change}, we only need to prove it for sets $s$ such that $\lev^\old(s) \le k$. Observe that for them, we have $\wts(s) < c_s$ by \Cref{clean-wts}. Therefore, \Cref{inv}(1) holds.
\end{proof}

\paragraph{Properties during processing dirty elements.}
Next, we show some properties after each iteration of the for loop that processes dirty elements (see the for loop at \cref{ln:process-dirty} of \Cref{alg:reset}). 

First, we state some properties that we will show are maintained after processing each element $e \in D$.

\begin{property}\label{hp-inv}
    \Cref{inv}(1) holds.
\end{property}

\begin{property}\label{hp-elem-ilev}
    For each element $e'$, we have $\ilev(e') \ge k + 1$ if $e'$ is active, and $\ilev(e') > k + 1$ if $e'$ is passive.
\end{property}

\begin{property}\label{hp-elem-correctness}
    Each passive $e'$ satisfies $\lev(e') < \ilev(e') \le \zlev(e') + \left\lceil\log_{1+\epsilon}\max\{f, \frac{2C}{\epsilon} \}\right\rceil$ if it is passive, and satisfies $\ilev(e') = \zlev(e') = \lev(e')$ if it is active.
\end{property}

\begin{property}\label{hp-below}
    Any set $s$ such that $\lev(s) < k + 1$, has $\wts(s) < c_s$, $\lev(s) = 0$ and $\phi(s) = 0$.
\end{property}

Next, we state some claims about the iterations of the for loop. We will prove all of them simultaneously.
\begin{claim}\label{handle-properties}
    Processing an element $e \in D$ maintains \Cref{hp-inv,hp-elem-ilev,hp-below,hp-elem-correctness}.
\end{claim}
Note that all the conditions of \Cref{hp-inv,hp-elem-ilev,hp-below,hp-elem-correctness} hold after processing clean elements by \Cref{clean-inv,clean-elem-ilev,clean-below,clean-elements}.

Recall that due to the implicit zeroing, the condition $\zlev(e) \le \lev(e)$ can be temporarily violated for a dirty element $e$. The following claim states that this condition is repaired after processing $e$.
\begin{claim}\label{handle-zlev}
    If an element $e \in D$ has not become active, then $\zlev(e) \le \lev(e)$ after processing it.
\end{claim}

The following claim shows that processing a dirty element is ``local'', in a sense, that it does not affect the lazy and intrinsic levels of other passive elements (except possibly making them active).
\begin{claim}\label{handle-local}
    During processing an element $e \in D$, every passive element $e' \ne e$ either becomes active or $\zlev(e')$ and $\ilev(e')$ remain the same.
\end{claim}

\begin{claim}\label{handle-tight}
    Any tight set stays tight after processing an element $e \in D$.
\end{claim}

\begin{claim}\label{handle-slack}
    Any slack set $s$ either has $\lev(s) = 0$, or $s \in S$ and $\lev(s) = k + 1$.
\end{claim}\
Recall that the purpose of $S$ was to collect all sets for which \Cref{inv}(2) might have been violated. \Cref{handle-slack} shows that $S$ indeed contains all such sets. Also, notice that \Cref{clean-slack} implies that the conditions of \Cref{handle-slack} hold after processing clean elements.

\begin{claim}\label{handle-lev-inc}
    During processing an element $e \in D$, levels of sets can only increase.
\end{claim}

\begin{claim}\label{handle-cover}
    If during processing an element $e \in D$, it was not added to $E'$, then $e$ is covered by a set $s \in T_{> k}$ after processing $e$.
\end{claim}

\begin{claim}\label{handle-eta}
    If during processing an element $e \in D$ we call $\procrand(e)$ and reach the computation of $\eta$, then there exists $\eta$ such that \[\clog_{1+\epsilon}^{(\eta+1)}f\leq \ilev(e) - k-1\leq \clog_{1+\epsilon}^{(\eta)}f\,.\]
\end{claim}

Before we proceed to the proof of \Cref{handle-properties,handle-zlev,handle-local,handle-tight,handle-lev-inc,handle-cover,handle-slack}, let us first prove two lemmas about the properties of $\delta$ defined in the algorithm.

\begin{lemma}\label{witness-tight}
    During a call to $\procrand(e)$, each set $s$ such that $\wts(s) - \wts(e) + \delta \ge c_s$ is tight.
\end{lemma}
\begin{proof}
    For such $s$ we have $\wts(s) - \wts(e) \ge c_s - \delta \ge c_s - \frac{\epsilon}{2C} \ge \frac{c_s}{1 + \epsilon}$, so $s$ is tight.
\end{proof}

\begin{lemma}\label{delta-increase}
    During a call to $\procrand(e)$, we have $\delta > \wts(e)$.
\end{lemma}
\begin{proof}
    Observe that if we enter the case where we compute $\delta$, we have $\ilev(e) - k - 1 > 1 + 2\log_{1 + \epsilon} \frac{2C}{\epsilon}$. By definition, $\delta \ge \brac{\frac{\epsilon}{2C}}^2 \cdot (1 + \epsilon)^{-k-1}$. Recall that $\wts(e) = (1 + \epsilon)^{-\ilev(e)}$. Thus, \[\wts(e) = (1 + \epsilon)^{-\ilev(e)} < (1 + \epsilon)^{-2\log_{1 + \epsilon}\frac{2C}{\epsilon} - k - 1} = \brac{\frac{\epsilon}{2C}}^2 \cdot (1 + \epsilon)^{-k-1} \le \delta.\]
\end{proof}

\begin{proof}[Proof of \Cref{handle-properties,handle-zlev,handle-local,handle-tight,handle-lev-inc,handle-cover,handle-slack,handle-eta}]
    Consider the iteration that processes an element $e \in D$ and assume the claims hold for each of the previous iterations. Observe that if $e$ is already active, then the algorithm makes no changes, so \Cref{handle-zlev,handle-local,handle-lev-inc,handle-tight,handle-slack,handle-properties} hold vacuously. Observe that $e$ could become active only due to the call to $\fix$ for some previously processed $e' \in D$. At the iteration for $e'$, we had $\ilev(e) > k + 1$ by \Cref{hp-elem-ilev}. By \Cref{obs:fix-main}(2), $e$ becomes active with $\lev(e) > k + 1$ after that call, and hence there is a set $s \ni e$ with $\lev(s) = \lev(e) > k + 1$, which must be tight by \Cref{handle-slack}. This set has been tight since then by \Cref{handle-tight}, and its level could have only increased by \Cref{handle-lev-inc}.
    
    Otherwise, $e$ is passive, and we call $\procdet(e)$ or $\procrand(e)$. First, we prove \Cref{handle-eta}. We need to show that $\ilev(e) - k - 1$ is not too large and not too small. According to \Cref{obs:clean-passive} and \Cref{handle-local}, $\ilev(e) = \ilev^\old(e)$ and $\zlev(e) = \zlev^\old(e) \le k$, since $e$ is dirty. Therefore, $\ilev(e) - k - 1$ is bounded by $\ilev^\old(e) - \zlev^\old(e) - 1 \le \left\lceil\log_{1+\epsilon}\max\{f, \frac{2C}{\epsilon} \}\right\rceil - 1$. Recall that in the randomized version we assume $f > \frac{2C}{\epsilon}$. Thus, $\ilev(e) - k - 1 \le \left\lceil\log_{1+\epsilon}f\right\rceil - 1 \le 5 \log_{1 + \epsilon} f$ and $\eta$ is at least one. Observe that we enter the branch where we compute $\eta$ only when $\ilev(e) - k - 1 > \frac{200}{\epsilon^2}$. By \Cref{lm:iterated-log}, $\clog_{1 + \epsilon}(y) \le \sqrt{\epsilon/5} \cdot y$ for any  $y \ge \frac{200}{\epsilon^2}$. Hence, such $\eta$ exists.
    
    To prove the remaining claims, consider the following three cases and their subcases:
    \begin{enumerate}[(1),leftmargin=*]
        \item We call $\procdet(e)$ or call $\procrand(e)$ and fall back to calling $\procdet(e)$. Then we branch into two cases:
        \begin{enumerate}[leftmargin=*]
            \item \label{case-a} There is a tight set $s \ni e$. Observe that if $\lev(s) \ge k + 1$, then $\lev(s)$ remains the same. If $\lev(s) < k + 1$, then we have $\wts(s) < c_s$. Therefore, \Cref{hp-inv,hp-below} are preserved. Observe that \Cref{hp-elem-ilev} is preserved as well. Notice that this does not break \Cref{hp-elem-correctness}, since for any active element $e'$, we have $\lev(e') \ge k + 1$, thus $\lev(e')$ is not affected; for any passive element $e'$, we have $\ilev(e') > k + 1$, and $\lev(e')$ can increase only to $k + 1$. By \Cref{handle-local} and \Cref{obs:clean-passive}, $\zlev(e) = \zlev^\old(e) \le k$ before the iteration, so $\zlev(e)$ can only increase. Therefore, \Cref{hp-elem-correctness} is preserved, and hence \Cref{handle-properties} holds.
            \Cref{handle-zlev,handle-local,handle-tight,handle-lev-inc,handle-cover,handle-slack} can be easily verified by the description of the algorithm.

            \item \label{case-b} All sets $s \ni e$ are slack. In that case, we call $\decrease(e)$. Observe that in the $\decrease(e)$ subroutine, steps before the branching can only increase $\wts(s)$ and do not break the tightness for any set. Thus, \Cref{handle-tight} holds. Observe that \Cref{handle-zlev,handle-local,handle-cover,handle-slack} hold as well.

            Suppose we enter the case where all sets $s \ni e$ are slack. Then it must be $\lev(s) \le k + 1$ for all $s \ni e$, since each slack set is either at level 0, or at level $k + 1$. Observe that this also implies that levels of sets can only increase, and thus \Cref{handle-lev-inc} holds.

            Since initially we had $\wts(s) < c_s$ for sets $s \ni e$, and due to the way we update the weights, we have $\wts(s) < c_s$ for all $s \ni e$ in the end, and hence \Cref{inv}(1) holds, so \Cref{hp-inv} is preserved. Since levels of sets can only increase, and they increase to $k + 1$, \Cref{hp-below} is preserved. Observe that $\zlev(e)$ becomes equal to $\lev(e)$ and $\ilev(e) \ge \zlev(e) \ge k + 1$, so \Cref{hp-elem-correctness,hp-elem-ilev} are preserved as well. 
        \end{enumerate}

        \item We call $\procrand(e)$ and do not fall back to $\procdet(e)$, and the random sampling step finds a witness set $s$.
        In that case, $s$ is tight by \Cref{witness-tight}, and so the proof is identical to case \ref{case-a}.
        
        \item We call $\procrand(e)$ and do not fall back to $\procdet(e)$, and the random sampling step fails. In that case, we compute $\widehat{F}$ and enter the branching. Consider each branch separately:
        \begin{enumerate}[leftmargin=*]
            \item If $|\widehat{F}| = 0$, then we call $\decrease(e)$. By \Cref{delta-increase}, $\delta > \wts(e)$, so we have $\wts(s) < c_s$ for all sets $s \ni e$. The rest of the proof is identical to case \ref{case-b} (with the exception that not all sets $s \ni e$ are necessary slack).
            
            \item If $0 < |\widehat{F}| \leq \left(\clog_{1+\epsilon}^{(\eta)}f\right)^2$, then we call $\fix(e, l)$. 
            Recall that $\fix(e, l)$ assumes the following: $\lev(e) < l \le \ilev(e)$, we have $\zlev(e) = \lev(e)$, \Cref{inv}(1) holds and for every passive element $e^\prime$ we have $\ilev(e^\prime) > \lev(e^\prime)$. 
            By the same argument as in case \ref{case-a}, all the properties hold after setting $\lev(s) \gets \max\{k + 1, \lev(s)\}$ for all $s \in e$ and updating $\zlev(e) \gets \max_{s \ni e}\{\lev(s)\}$. Hence the last three assumptions hold. By \Cref{hp-elem-correctness}, we have $\ilev(e) > \lev(e)$. Thus, by the definition of $l$, $\lev(e) < l \le \ilev(e)$, and hence the first assumption holds as well.

            The properties are preserved by the call to $\fix(e, l)$ and the subsequent implicit zeroing of sets at level 0. First, notice that if $l = \ilev(e)$, then the call makes no changes by \Cref{obs:fix-no-changes}. Otherwise, \Cref{hp-inv} is preserved by \Cref{fix-invs}. 
            \Cref{hp-elem-ilev} is preserved by \Cref{obs:fix-main}(2). 
            \Cref{hp-elem-correctness} is preserved by \Cref{obs:fix-main}(1)(2) and by \Cref{fix-lev<ilev}. Next, we show that \Cref{hp-below} is preserved. By \Cref{obs:fix-main}(3), only sets $s \ni e$ can change their level, and their levels can only increase. By \Cref{slack-no-rise}, only tight set can change their level. Observe that all tight sets $s \ni e$ are at level at least $k + 1$ before the call to $\fix(e, l)$. Due to the implicit zeroing, $\phi(s) = 0$ for all sets at level 0. Therefore, \Cref{hp-below} is preserved.

            \Cref{handle-zlev,handle-local,handle-tight,handle-lev-inc} hold by \Cref{obs:fix-main} and \Cref{fix-tight}.
            \Cref{handle-slack} holds, since $\fix$ does not raise slack sets by \Cref{slack-no-rise}. 
            Finally, we show that \Cref{handle-cover} holds. Observe that $\widehat{F} \ne \emptyset$ by the entering condition. By \Cref{witness-tight}, every set in $\widehat{F}$ is tight, and it remains tight after the call to $\fix(e, l)$ by \Cref{fix-tight}.
            
            \item If $|\widehat{F}| \geq \left(\clog_{1+\epsilon}^{(\eta)}f\right)^2$, then we take an arbitrary set $s \in \widehat{F}$. By \Cref{witness-tight}, $s$ is tight, so the proof is identical to case \ref{case-a}.
        \end{enumerate}
    \end{enumerate}
\end{proof}

\begin{corollary}\label{iter-zlev<=k}
    At the beginning of the iteration for $e \in D$, if $e$ is still passive, then $\zlev(e) = \zlev^\old(e) \le k$ and $\ilev(e) = \ilev^\old(e)$.
\end{corollary}
\begin{proof}
    Since $e$ is dirty, $\zlev^\old(e) \le k$, and by \Cref{obs:clean-passive}, $\ilev(e)$ and $\zlev(e)$ remain the same after processing clean elements. By \Cref{handle-local}, they remain the same until this iteration.
\end{proof}

Next, let us show that $\widehat{E}$ collects all elements that are not covered by a tight set. We prove the following corollaries.

\begin{corollary}\label{cover-other}
    After processing dirty elements, every element $e \in E^\old_{>k}$ is covered by $T_{> k}$.
\end{corollary}
\begin{proof}
    By \Cref{clean-lvl-change}, $e$ is covered by $T^\old_{> k}$ after processing clean elements, and all these sets are tight. During processing dirty elements, sets from $T^\old_{>k}$ remain tight, and their level can only increase by \Cref{handle-tight,handle-lev-inc}.
\end{proof}

\begin{corollary}\label{cover}
    After processing dirty elements, every element $e\in D \setminus E'$ is covered by $T_{> k}$.
\end{corollary}
\begin{proof}
    By \Cref{handle-cover}, $e$ is covered by a set $s \in T_{>k}$ after the iteration that processes it. By \Cref{handle-tight,handle-lev-inc}, $s$ is tight and has $\lev(s) \ge k + 1$ after the subsequent iterations.
\end{proof}

\paragraph{Properties after $\reset(k)$.}
First, we make the following observation about the post-processing steps.
\begin{observation}\label{obs:post-processing}
    During the post-processing steps, $\ilev(e)$ and $\zlev(e)$ remain the same for each passive element $e$.
\end{observation}

Next, we show that we apply $\water$ on all elements not covered by a tight set.
\begin{claim}\label{slack-level}
    Before elements from $\widehat{E}$ and sets from $\widehat{S}$ are moved to level $\widehat{k}$ and $\water(\widehat{k}, \widehat{S}, \widehat{E})$ is invoked, $\widehat{E}$ contains all the elements that are currently not covered by a tight set, and $\widehat{S}$ contains all the sets that contain elements from $\widehat{E}$. Every other element $e \notin \widehat{E}$ is covered by $T_{> k}$.
\end{claim}
\begin{proof}
    Every element $e \in E^\old_{>k} \cup (D \setminus E')$ is covered by $T_{>k}$ by \Cref{cover-other,cover}. Notice that $T_{>k} \cap \widehat{S} = \emptyset$, since sets in $\widehat{S}$ are slack. 
    Therefore, if $e$ is not covered by a tight set, then $e \in (E \cup E')$. Recall that $\widehat{E}$ is the set of elements from $E \cup E'$ that are not covered by a tight set, and $\widehat{S}$ is the collection of all the sets that contain elements from $\widehat{E}$. Observe that whenever we process an element $e \in E$ or add $e$ to $E'$, we add all sets $s \ni e$ to $S$. Therefore, if $e \in (E \cup E') \setminus \widehat{E}$, then there is a tight set $s \ni e$. 
\end{proof}

Now we are ready to prove that the $\reset$ subroutine maintains \Cref{inv}(1)(2) and a valid cover. 

\begin{theorem}\label{reset-inv}
    \Cref{inv}(1)(2) hold after the call to $\reset(k)$.
\end{theorem}
\begin{proof}
    After processing dirty elements, \Cref{inv}(1) holds by \Cref{handle-properties}. Observe that the post-processing steps do not change $\wts(s)$ for $s \notin \widehat{S}$. Since only slack sets are dropped to level 0 during the post-processing, \Cref{inv}(1) holds for $s \notin \widehat{S}$. For any $s \in \widehat{S}$, \Cref{inv}(1) will hold after the call to $\water(\widehat{k}, \widehat{S}, \widehat{E})$ by \Cref{waterfill}. 

    After processing dirty elements, all slack sets at levels above 0 are in $S$ by \Cref{handle-slack}. All slack sets in $S$ are dropped to level 0 at \cref{ln:zero-slack}, so \Cref{inv}(2) holds for them after that. By the definitions of $\widehat{E}$ and $\widehat{S}$, it must be $\widehat{S} \subseteq S$. For sets from $\widehat{S}$, \Cref{inv}(2) can be violated after moving them to level $\widehat{k}$, but it will be restored afterward, due to the call to $\water(\widehat{k}, \widehat{S}, \widehat{E})$, by \Cref{waterfill}.
\end{proof}

\begin{theorem}\label{reset-cover}
    $T$ is a valid set cover for $\univ$ after the call to $\reset(k)$.
\end{theorem}
\begin{proof}
   By \Cref{slack-level}, $\widehat{E}$ contains all elements that are not in any tight set, and $\widehat{S}$ contains all the sets that contain such elements. Every element from $\univ \setminus \widehat{E}$ is covered by a tight set from $\set \setminus \widehat{S}$, since all the sets in $\widehat{S}$ are slack. Then, after the call to $\water(\widehat{k}, \widehat{S}, \widehat{E})$, all elements will be covered by tight sets by \Cref{waterfill}.  
\end{proof}

During the post-processing steps, we drop slack sets from $S$ to level 0, which could potentially decrease $\lev(e)$, and hence violate $\zlev(e) \le \lev(e)$ for passive elements, or $\zlev(e) = \lev(e)$ for active elements. In the following claim, we show that at the end of the execution of the $\reset(k)$ subroutine, no element will be in an incorrect state.

\begin{theorem}
    After the call to $\reset(k)$, for each passive element $e$, we have $\zlev(e)\leq \lev(e)$ and $\lev(e) < \ilev(e) \le \zlev(e) + \left\lceil\log_{1+\epsilon}\max\{f, \frac{2C}{\epsilon} \}\right\rceil$; for each active element $e$, we have $\ilev(e) = \zlev(e) = \lev(e)$.
\end{theorem}
\begin{proof}
    Notice that by \Cref{handle-properties}, we have that after processing dirty elements. 
    
    The inequality $\lev(e) < \ilev(e) \le \zlev(e) + \left\lceil\log_{1+\epsilon}\max\{f, \frac{2C}{\epsilon} \}\right\rceil$ holds for every passive element $e$, since the post-processing steps can only decrease $\lev(e)$, and they do not affect $\ilev(e)$ and $\zlev(e)$ by \Cref{obs:post-processing}. 

    Next, we prove the inequality $\zlev(e) \le \lev(e)$ for every passive element $e$. If $e \in E^\old_{>k}$, then $\lev(e)$ and $\zlev(e)$ are the same after the implicit zeroing and processing clean elements by \Cref{obs:clean-passive}. If $e \in D \setminus E'$, then by \Cref{handle-zlev}, $\zlev(e) \le \lev(e)$ after the iteration that processes $e$. By \Cref{handle-lev-inc,handle-local}, $\ilev(e)$ remains the same and $\lev(e)$ could only increase during the subsequent iterations. During the post-processing steps, dropping slack sets from $S$ to level 0 at \cref{ln:zero-slack} might break the inequality. However, by \Cref{cover-other,cover}, $e$ is covered by $T_{> k}$, and by \Cref{handle-slack}, every slack set that was dropped was at level $k + 1$. Therefore, $\lev(e)$ is not affected.
    
    Finally, we prove that $\ilev(e) = \zlev(e) = \lev(e)$ for every active element $e$. By \Cref{hp-elem-ilev,hp-elem-correctness}, it holds after processing dirty elements and $\lev(e) \ge k + 1$. If $e \in \widehat{E}$, then the equality will hold after the call to $\water$. Otherwise, $e$ is covered by $T_{> k}$, and as we have shown before, $\lev(e)$ is not affected by dropping slack sets from $S$ to level 0 at \cref{ln:zero-slack}.
\end{proof}

\section{Update time analysis}\label{sec:runtime-analysis}
\subsection{Potential functions} \label{potfun}

Following \cite{bhattacharya2021dynamic}, we define the \emph{up}, \emph{down} and \emph{lift potentials}. 
In addition, we introduce two new types of potential, which we call the \emph{passive potential} and the \emph{clean potential};
the definitions of these five types of potentials are given below, along with some intuitive explanations behind the definitions.
We define the \emph{total potential} of the set system, denoted by $\Phi$,  as the sum of all types of potentials across all elements and sets.
\begin{itemize}[leftmargin=*]
    \item \textbf{Up potential.} Define  $\alpha_i = 2f\left(\frac{3}{\epsilon^3} + \frac{\log C}{\epsilon^2}\right)(1+\epsilon)^{i+1}$, for any index $0\leq i\leq L$. Then, the \emph{up potential} of $s$ is defined as $\Phi_{\up}(s) = \max\{\wts(s)- c_s, 0\}\cdot \alpha_{\lev(s)}$.
    
    \emph{Intuition.} This type of potential, which is gained by element insertions, is used to cover the costs of raising active elements by the $\fix$ subroutine. Whenever we violate Invariant \ref{inv}(1), the process of raising active elements releases a sufficiently large amount of potential to cover the costs. Notice that potential release due to the loss of one unit of weight can cover the increase of down potential due to increase of $f$ units of dead weight.
    \item \textbf{Down potential.} Define $\beta_i = \frac{2}{\epsilon^2} (1 + \epsilon)^{i + 1}$ for any index $0\leq i\leq L$.
    Then, the \emph{down potential} of $s$ is defined as $\Phi_{\down}(s) = \phi(s)\cdot \beta_{\lev(s)}$.
    
    \emph{Intuition.} This type of potential is used to cover the costs of the $\reset$ subroutine. We gain it whenever we lose weight of some elements due to element deletions or raising levels of elements. Whenever Invariant \ref{inv}(3) is violated, we have gained large enough down potential to cover the costs of the $\reset$ subroutine, which aims at restoring the invariant.
    
    \item \textbf{Lift potential.} Each set $s$ has a \emph{lift potential} of $\Phi_{\lift}(s) = L - \max\{\lev(s), \bs(s)\}$.
    
    \emph{Intuition. } This type of potential is used to cover the costs of raising levels of sets above their base levels due to the $\fix$ subroutine. Initially, we have the maximum amount $L$ of that potential for every set. After that, it is restored due to the $\reset$ subroutine. It is needed to cover the cases where we do not raise any active elements, and hence do not decrease the up potential.
    
    \item \textbf{Passive potential.} Each element $e$ has a \emph{passive potential} $\Phi(e) = f$ if $e$ is passive, and $\Phi(e) = 0$ otherwise.
    
    \emph{Intuition.} This type of potential is used to cover the costs of activating an element. This potential is gained by element insertions.

    \item \textbf{Clean potential.} Each set $s$ such that $\phi(s) \ne 0$ has a \emph{clean potential} 
    $\Phi_\clean(s) = \frac{1}{\epsilon^2} + \frac{\log C}{\epsilon}$, and $\Phi_\clean(s) = 0$ otherwise.

    \emph{Intuition.} This potential is used to cover the costs associated with the $\reset$ subroutine. Sometimes, the decrease of down potential is not enough to cover the $O(\frac{1}{\epsilon^2} + \frac{\log C}{\epsilon})$ term in the runtime cost. But in that case, we argue it can be covered by the decrease of clean potential.
\end{itemize}

Initially, when $\univ = \emptyset$, by definition of our potential functions, the set system has potential at most $mL = O(m\log(Cn))$. This is because the lift potential is at most $L$ for each set, $\wts(s), \phi(s) = 0$ for each set $s \in \set$, and there are no passive elements. 
We note that the amortized update time (associated with potential function $\Phi$) 
of an element update 
is defined as the sum of the potential change $\Delta\Phi$ due to the update and the actual time spent by the update algorithm; for technical convenience, we shall assume that any unit of potential can cover $O(1)$ units of time. 
Consequently, in amortized analysis via the potential function method, we would like cheap operations (such as element deletions in our case) to increase the total potential, but not by too much, whereas costly operations (such as the $\reset$ subroutine) should decrease the total potential in roughly the same amount as their actual running time.

\subsection{Deletion}
By the algorithm's description, the $\del(e)$ subroutine takes $O(f)$ time. Hence, it suffices to bound the potential increase following an execution of the $\del(e)$ subroutine. 

Observe that the up potential may only decrease, and the lift and passive potentials for elements $e^\prime \ne e$ remain unchanged; $\Phi(e)$ may only decrease.
To bound the increase in down potential, note that for each $s\ni e$ that was tight before the deletion of element $e$, its dead weight increases by at most $\wts(e)$. We know that $\lev(e) \ge \lev(s)$; if $e$ is passive, then we also have $\ilev(e) > \lev(e) \ge \lev(s)$. Therefore, $\wts(e) \le (1 + \epsilon)^{-\lev(s)}$, and the increase of $\Phi_\down(s)$ is at most
\[\wts(e) \cdot \beta_{\lev(s)} \le (1+\epsilon)^{-\lev(s)}\cdot \frac{2}{\epsilon^2} \cdot (1+\epsilon)^{\lev(s)+1}\leq \frac{2(1+\epsilon)}{\epsilon^2}.\]
The clean potential may also increase due to the increase of $\phi(s)$; the increase of $\Phi_\clean(s)$ is bounded by $\frac{1}{\epsilon^2} + \frac{\log C}{\epsilon}$.
Hence, the total increase of $\Phi$ after the deletion of $e$ (and the amortized runtime cost) is  $O\left(\frac{f}{\epsilon^2} + \frac{f \log C}{\epsilon}\right)$.

\subsection{Fixing levels}
Recall that we assume \Cref{inv}(1) holds before the execution of $\fix(e, l)$ started. As we did for the properties, use the super-script ``$\old$'' to denote the values of the variables before the execution of $\fix(e, l)$ started. 
We only need to consider the case where $l$ is strictly smaller than the old value of $\ilev^\old(e)$; otherwise, since \Cref{inv}(1) held before, none of the while loops on \cref{ln:fix-while} would be triggered, and the algorithm would make no changes by \Cref{obs:fix-no-changes}.  The total amortized runtime of the procedure in this case would be $O(f)$.

During the call, we may activate passive elements in $P_{\lev(s)}(s)$ for some set $s$ (see \cref{ln:activate-k+1,ln:activate-below} in the pseudocode of the $\fix$ subroutine). The activation of a passive element $e^\prime$ takes time $O(f)$. However, this runtime cost can be charged to the clearance of its passive potential $\Phi(e^\prime)$; indeed, note that the passive potential of any passive element is at least $f$ and that of any active elements is 0.

For the rest, we will only be concerned with other steps in the subroutine. Consider the moment just before an iteration of the while loop (or before entering the branching at \cref{ln:fix-below}). Let $\wts^\nw(s)$, $\Phi^\nw_\up(s)$ and $\lev^\nw(s)$ be the weight, the up potential and the level of $s$ right after this iteration (or the branching at \cref{ln:fix-below}) respectively. Define $\widehat{\wts}(s) = \wts(s) - \wts(e)$ and $\widehat{\wts}^\nw(s) = \wts^\nw(s) - \wts^\nw(e)$. 

First, we bound the potential increase due to raising $s$ to level $\min\{\bs(s),\zlev(e)\}$ (i.e., when we enter the if statement at \cref{ln:fix-below}). It is easy to see that $\Phi_\down + \Phi_\clean$ does not increase, since we zero out $\phi(s)$. However, $\Phi_\up$ may increase.
\begin{claim}\label{fix-raise-base}
     If $s \in F$, then the increase of $\Phi_\up(s)$ due to raising $s$ to level $\min\{\bs(s), \zlev(e)\}$ is at most $\wts^\nw(e) \cdot \alpha_{\lev^\nw(s)} - \wts(e) \cdot \alpha_{\lev(s)}$.
\end{claim}
\begin{proof}
    Let $k = \lev(s)$ and $k' = \lev^\nw(s)$. By \Cref{inv}(1) and \Cref{fix-lev+1-e}, we have $\wts(s, \lev(s) + 1) - \wts(e) \le \wts^\old(s, \lev^\old(s) + 1) < c_s$.
    By \Cref{fix-empty}, there are no elements in $A_k(s)$, so from \Cref{eqlevelup}, we have $\wts(s, \lev(s) + 1) = \wts(s)$. Hence $\widehat{\wts}(s) = \wts(s) - \wts(e) < c_s$. Observe that no weight of an element is changed, so $\wts(s) = \wts^\nw(s)$ and $\widehat{\wts}(s) = \widehat{\wts}^\nw(s)$. By the entering condition, $\wts(s) = \wts^\nw(s) = \wts(s, \lev(s) + 1) \ge c_s$. Therefore,
    \begin{align*}
        \Phi_\up^\nw(s) - \Phi_\up(s) &= (\widehat{\wts}^\nw(s) + \wts^\nw(e) - c_s)\cdot \alpha_{k'} - (\widehat{\wts}(s) + \wts(e) - c_s) \cdot \alpha_k \\
        &=(\widehat{\wts}(s) - c_s) \cdot (\alpha_{k'} - \alpha_k) + \wts^\nw(e) \cdot \alpha_{k'} - \wts(e) \cdot \alpha_k \\
        &\le \wts^\nw(e) \cdot \alpha_{k'} - \wts(e) \cdot \alpha_k.
    \end{align*}
\end{proof}
Next, we analyze the potential increase due to a single iteration of the while loop. Consider the iteration of the for loop that raises a set $s \ni e$. Note that if $s \notin F$, then we do not enter the while loop by \Cref{fix-enter-F}. So the interesting case is when $s \in F$.  We are going to prove some claims about that iteration. Recall that $k$ is the level of $s$ at the beginning of an iteration of the while loop.

\begin{lemma}\label{inv1}
    At the beginning of each iteration of the while loop, we have $\wts(s) - \wts(e) - |A_k(s)|\cdot\epsilon (1+\epsilon)^{-k-1}< c_s$.
\end{lemma}
\begin{proof}
    Since we assumed that \Cref{inv}(1) has held before the execution of $\fix$, we have $\wts^\old(s, \lev^\old(s) + 1) < c_s$. By \Cref{fix-lev+1-e}, $\wts(s, k + 1) - \wts(e) \le \wts^\old(s, \lev^\old(s) + 1) - \wts^\old(e)$. Thus, $\wts(s, k + 1) - \wts(e) < c_s$. Since $\wts(s, k + 1) = \wts(s) - |A_k(s)| \cdot \epsilon(1 + \epsilon)^{-k-1}$ by \Cref{eqlevelup}, we get the desired inequality.
\end{proof}

\begin{claim}\label{fix-raise-one-up}
    If $s \in F$, then the increase of $\Phi_\up(s)$ after a single iteration of the while loop is at most 
    \begin{equation}\label{eq:fix-iter-up-pot}
        -2f \cdot |A_k(s)|\cdot \left(\frac{3}{\epsilon^2} + \frac{\log C}{\epsilon}\right)  + \wts^\nw(e) \cdot \alpha_{k + 1} - \wts(e) \cdot \alpha_k.
    \end{equation}
    For any other set $s^\prime \ne s$ the up potential $\Phi_\up(s^\prime)$ does not increase.
\end{claim}
\begin{proof}
    If $\wts^\nw(s) \le c_s$,  then $\Phi_\up^\nw(s) = 0$. So by the while loop condition that $\wts(s, k+1) \ge c_s$ and by \Cref{eqlevelup}, we know that the up potential change is equal to 
    \[-\Phi_\up(s) = -(\wts(s) - c_s)\cdot \alpha_k\leq -\left(\wts(s) - \wts(s, k+1)\right)\cdot \alpha_k\leq -2f \cdot |A_k(s)|\cdot\left(\frac{3}{\epsilon^2} + \frac{\log C}{\epsilon}\right).\]
    Otherwise, we have $\wts^\nw(s) > c_s$. Then,
    \[\Phi_\up^\nw(s) - \Phi_\up(s) = (\widehat{\wts}^\nw(s) + \wts^\nw(e) - c_s) \cdot \alpha_{k^\prime} - (\widehat{\wts}(s) + \wts(e) - c_s) \cdot \alpha_k.\]
    By the algorithm, all elements from $A_k(s)$ are raised to level $k+1$, and so their weights decrease by the factor of $1 + \epsilon$. Hence, we get $\widehat{\wts}^\nw(s) = \widehat{\wts}(s) - |A_k(s)| \cdot \epsilon(1 + \epsilon)^{-k - 1}$. Notice that $\alpha_{k + 1} = (1 + \epsilon)\alpha_k$. Therefore,
    \begin{align*}
        \widehat{\wts}^\nw(s) \cdot \alpha_{k + 1} - \widehat{\wts}(s) \cdot \alpha_k &= \widehat{\wts}(s) \cdot (1 + \epsilon)\alpha_k - |A_k(s)| \cdot \epsilon(1 + \epsilon)^{-k} \cdot \alpha_k - \widehat{\wts}(s) \cdot \alpha_k \\
        &= (\widehat{\wts}(s) - |A_k(s)| \cdot (1 + \epsilon)^{-k}) \cdot \epsilon\alpha_k \\
        &= (\wts(s) - \wts(e) - |A_k(s)| \cdot (1 + \epsilon)^{-k}) \cdot \epsilon\alpha_k.
    \end{align*}
    Using \Cref{inv1}, we can bound it by $(c_s - |A_k(s)| \cdot (1 + \epsilon)^{-k - 1}) \cdot \epsilon\alpha_k$. Then,
    \begin{align*}
        \Phi_\up^\nw(s) - \Phi_\up(s) & \le -|A_k(s)| \cdot (1 + \epsilon)^{-k - 1} \cdot \epsilon\alpha_k + \wts^\nw(e) \cdot \alpha_{k + 1} - \wts(e) \cdot \alpha_k \\
        &= -2f \cdot |A_k(s)|\cdot \left(\frac{3}{\epsilon^2} + \frac{\log C}{\epsilon}\right)  + \wts^\nw(e) \cdot \alpha_{k + 1} - \wts(e) \cdot \alpha_k.
    \end{align*}
\end{proof}

The total down potential $\Phi_{\down}$ and the total clean potential $\Phi_\clean$ may change after an iteration of the while loop, and specifically 
increase, due to changes of dead weights and the increase of $\lev(s)$. 
\begin{claim}\label{fix-raise-one-down-clean}
    If $s \in F$, then the overall increase of $\Phi_\down + \Phi_\clean$ after a single iteration of the while loop is at most $(f - 1) \cdot |A_k(s)| \cdot \left(\frac{3}{\epsilon^2} + \frac{\log C}{\epsilon}\right)$.
\end{claim}
\begin{proof}
    Observe that $\phi(s)$ is zeroed out at the beginning of the iteration, and later it does not increase, so $\Phi_\down(s) = 0$. By the algorithm, each element $e^\prime\in A_k(s)$ incurs an increase of each $\phi(s^\prime)$ by at most $\epsilon(1+\epsilon)^{-k-1}$. So the overall increase of dead weights for $s' \ni e^\prime, s^\prime \ne s$ is at most $(f - 1)\cdot |A_{k}(s)|\cdot \epsilon(1+\epsilon)^{-k-1}$. As at the beginning of the iteration, we have $\lev(s^\prime) \leq \lev(e^\prime) = k$ for $s' \ne s$, and $\lev(s^\prime)$ is unchanged during it, the total increase of down potential due to these sets is bounded by \[(f - 1)\cdot |A_{k}(s)|\cdot \epsilon(1+\epsilon)^{-k-1}\cdot \beta_k = (f - 1) \cdot |A_k(s)| \cdot \frac{2}{\epsilon}.\]
    For each $s^\prime$, the increase of $\Phi_\clean(s^\prime)$ is bounded by $\frac{1}{\epsilon^2} + \frac{\log C}{\epsilon}$; therefore, the total increase of clean potential due to these sets is at most
    \[(f - 1)\cdot |A_{k}(s)|\cdot \epsilon(1+\epsilon)^{-k-1}\cdot \beta_k = (f - 1) \cdot |A_k(s)| \cdot \brac{\frac{1}{\epsilon^2} + \frac{\log C}{\epsilon}}.\] 
\end{proof}

\begin{claim}\label{fix-total-potential}
    After the call to $\fix(e, l)$, the sum $\Phi_\up + \Phi_\down + \Phi_\clean$ increases by at most \[3|F| \cdot \left(\frac{3f}{\epsilon^3} + \frac{f \log C}{\epsilon^2}\right) \cdot (1 + \epsilon)^{-d  + 1} + |F| \cdot \left(\frac{1}{\epsilon^2} + \frac{\log C}{\epsilon}\right).\]
\end{claim}
\begin{proof}
    Consider a set $s \ni e$. If $s \notin F$, then $\Phi_\up(s)$ does not increase after refreshing $\wts(s)$ according to the up-to-date $\wts(e)$, since it can only decrease $\wts(s)$ by \Cref{fix-ilev-upper}. By \Cref{fix-enter-F}, we do not raise $s$, so $\Phi_\up + \Phi_\down + \Phi_\clean$ does not increase after this iteration. 

    Now consider the case when $s \in F$. 
    When $\wts(s)$ is refreshed at the beginning of an iteration of the for loop, $\Phi_\up(s)$ increases (compared to the beginning of the call) by at most $\wts(e)\cdot \alpha_{\lev(s)} - \wts^\old(e)\cdot \alpha_{\lev(s)}$, and $\Phi_\down(s)$ and $\Phi_\clean(s)$ remain unchanged. 
   
    Consider an iteration of the while loop for $s$. According to \Cref{fix-raise-one-up}, the first term $-2f \cdot |A_k(s)|\cdot \left(\frac{3}{\epsilon^2} + \frac{\log C}{\epsilon}\right)$ in \Cref{eq:fix-iter-up-pot} is enough to cover the increase of $\Phi_\down + \Phi_\clean$ which is at most $(f - 1) \cdot |A_k(s)| \cdot \left(\frac{3}{\epsilon^2} + \frac{\log C}{\epsilon}\right)$ by \Cref{fix-raise-one-down-clean}.

    Summing the terms $\wts^\nw(e) \cdot \alpha_{k+1} - \wts(e) \cdot \alpha_k$ from \Cref{eq:fix-iter-up-pot} over all iterations of the while loop, together with the increase due to raising $s$ to level $\min\{\bs(s),\zlev(e)\}$ from \Cref{fix-raise-base}, and the initial increase in $\Phi_\up(s)$ due to refreshing $\wts(s)$ according to the up-to-date $\wts(s)$, the result is bounded by the value of $\wts(e) \cdot \alpha_{\lev(s)}$ at the moment after the last iteration of the while loop. Since $\ilev(e) \ge \lev(s) + d$, this value is at most $2f\left(\frac{3}{\epsilon^3} + \frac{\log C}{\epsilon^2}\right) \cdot (1 + \epsilon)^{-d  + 1}$.
    Therefore, the total increase of $\Phi_\up + \Phi_\down + \Phi_\clean$ is at most 
    \begin{equation*}
        2f\left(\frac{3}{\epsilon^3} + \frac{\log C}{\epsilon^2}\right) \cdot (1 + \epsilon)^{-d  + 1}.
    \end{equation*}

    During the finalization step, $\wts(s)$ can only decrease, so $\Phi_\up(s)$ does not increase. If $s \in F$, we also update the value of $\phi(s)$, so $\Phi_\down(s)$ and $\Phi_\clean(s)$ can increase. Consider the moment when we assigned $l_s \leftarrow \ilev(e)$. At this moment we had $\ilev(e) \ge \lev(s) + d$. Observe that since that, $\lev(s)$ remains unchanged until the end. Therefore, the increase of $\Phi_\down(s)$ is at most \[\begin{aligned}
        \beta_{\lev(s)}\cdot (1+\epsilon)^{-l_s}&\leq \frac{2}{\epsilon^2} \cdot (1+\epsilon)^{\lev(s)+1}\cdot (1+\epsilon)^{-\lev(s)-d}\\ &< f\brac{\frac{3}{\epsilon^3} + \frac{\log C}{\epsilon^2}}\cdot (1+\epsilon)^{-d+1}.
    \end{aligned}\]
    The increase of $\Phi_\clean(s)$ is at most $\frac{1}{\epsilon^2} + \frac{\log C}{\epsilon}$. 

   To conclude, the overall increase in potential is at most \[3|F| \cdot \left(\frac{3f}{\epsilon^3} + \frac{f \log C}{\epsilon^2}\right) \cdot (1 + \epsilon)^{-d  + 1} + |F| \cdot \left(\frac{1}{\epsilon^2} + \frac{\log C}{\epsilon}\right).\]
\end{proof}
Now we are ready to bound the amortized runtime cost of the $\fix$ subroutine. The steps outside the while loop and activating passive elements take $O(f)$ time in total. We already have shown that the runtime cost of activating passive elements could be charged to the clearance of their passive potentials. The runtime of a single iteration of the while loop is $O(f \cdot |A_k(s)| + 1)$. If $A_k(s) \ne \emptyset$, then we can cover this runtime cost by the decrease in potential, according to \Cref{fix-raise-one-up,fix-raise-one-down-clean}. 

Next, consider the case $A_k(s) = \emptyset$. If $k \ge \bs(s)$, then we can cover the $O(1)$ runtime cost by the decrease of $\Phi_\lift(s)$. 
Notice that if $k < \bs(s)$ at the first iteration of the while loop, then we have $\zlev(e) = \lev(s)$, and during each iteration, both $\zlev(e)$ and $\lev(s)$ will increase by one. Indeed, if we had $\lev(s) < \zlev(e)$ before we entered the while loop, we must have entered the if statement at \cref{ln:fix-below}, after which $\lev(s)$ becomes equal to $\zlev(e)$. Therefore, after each iteration, where $k < \bs(s)$, we have $\zlev(e) = \lev(s)$ and both of them increase by one. Notice that $\bs(s) \le \ceil{\log_{1 + \epsilon} C}$ for each set $s$. Hence, there are $O(\frac{\log C}{\epsilon}) = O(f)$ such iterations, and the total runtime we spend on them is $O(f)$.

Observe that for any element $e^\prime \ne e$, the passive potential $\Phi(e^\prime)$ could only decrease, and $\Phi_\lift$ does not increase. $\Phi(e)$ increase by at most $f$, if $e$ is a freshly inserted element. The increase of $\Phi_\up + \Phi_\down + \Phi_\clean$ is bounded by \Cref{fix-total-potential}.

Therefore, we conclude our analysis by the following theorem.

\begin{theorem}\label{fix-runtime}
    The amortized runtime cost of $\fix(e, l)$ is bounded by \[3|F| \cdot \left(\frac{3f}{\epsilon^3} + \frac{f \log C}{\epsilon^2}\right) \cdot (1 + \epsilon)^{-d  + 1} + |F| \cdot \left(\frac{1}{\epsilon^2} + \frac{\log C}{\epsilon}\right) + O(f).\]
\end{theorem}

\subsection{Rebuilding}
As before, we use the superscript ``$\old$'' to denote the values of the variables right before $\reset(k)$ started. Similarly to \cite{bhattacharya2021dynamic}, we will argue that due to any call to $\reset(k)$, we have released a large amount of 
potential to compensate for the update time. There is a significant difference, however: the release of down potential allows us to compensate the costs associated with active elements only. Thus, the majority of our argument is devoted to dealing with passive elements.

\begin{lemma}[\cite{bhattacharya2021dynamic}]\label{down}
    We have the following lower bound on the down potential $\Phi_\down(T_{\le k})$
    \[\sum_{s\in T_{\leq k}}2\max\{\lev(s) - \bs(s)+1, 0 \} + \frac{2f}{\epsilon}|A_{\leq k}|\]
\end{lemma}
\begin{proof}
    By definition, $k$ is the minimum index such that: 
    \begin{equation}\label{first}
        \phi_{\leq k} ~>~ \epsilon\cdot \left(c(T_{\leq k}) + f\cdot \wts(E_{\leq k})\right),
    \end{equation}
    \begin{equation}\label{second}
        \phi_{\leq i} ~\leq~ \epsilon\cdot \left(c(T_{\leq i}) + f\cdot \wts(E_{\leq i})\right), \forall 0\leq i<k.
    \end{equation}
    Therefore, we have:
    \[\begin{aligned}
        &\Phi_{\down}(T_{\leq k}) = \sum_{i = 0}^k\phi_{i}\cdot \beta_i = \beta_k\cdot \phi_{\leq k} - \sum_{i =0}^{k-1}(\beta_{i+1} - \beta_i)\cdot \phi_{\leq i}\\
        &>\beta_k\cdot \epsilon\cdot \left(c(T_{\leq k}) + f\cdot \wts(E_{\leq k})\right) - \sum_{i=0}^{k-1}(\beta_{i+1} - \beta_i)\cdot \epsilon\cdot \left(c(T_{\leq i}) + f\cdot \wts(E_{\leq i})\right)\\
        &=\epsilon\sum_{i=0}^k\beta_i\cdot c(T_i) + \epsilon f\cdot\sum_{i =0}^k \beta_i \cdot \wts(E_i) 
        \mbox{~~(by Abel transformation)}
        \\
        &\geq \epsilon\sum_{i = 0}^k\sum_{s\in T_i}c_s\cdot \frac{2(1+\epsilon)^{i+1}}{\epsilon^2} + \epsilon f\cdot\sum_{i =0}^k \sum_{e\in A_i} \wts(e)\cdot (1 + \epsilon)^{i + 1} \cdot \frac{2}{\epsilon^2}\\
        &\geq \sum_{i =0}^k \sum_{s\in T_i}\frac{2(1+\epsilon)^{i - \bs(s)}}{\epsilon} + f\cdot\sum_{i = 0}^k\sum_{e\in A_i} (1+\epsilon)^{-i} \cdot (1 + \epsilon)^{i + 1}\cdot\frac{2}{\epsilon} \\
        &\geq \sum_{s\in T_{\leq k}}2\max\{\lev(s) - \bs(s)+1, 0 \} + \frac{2f}{\epsilon}|A_{\leq k}|,
    \end{aligned}\]
    where the first inequality follows from \Cref{first} and \Cref{second}; the last two inequalities hold, since $c_s \ge (1 + \epsilon)^{-\bs(s) - 1}$ and $(1 + \epsilon)^x \ge 1 + \epsilon x \ge \epsilon (1 + x)$. 
\end{proof}

First, let us analyze the runtime cost of the steps inside $\reset(k)$, together with finding $k$, excluding the steps inside the calls to $\procdet(e)$ and $\procrand(e)$. By \Cref{linkedlist}, the runtime cost of finding $k$, implicitly zeroing the sets up to level $k$, and initializing set $E$ is 
\[O\brac{\frac{\log C}{\epsilon} + |T_{\leq k}\setminus T_{\leq \ceil{\log_{1+\epsilon}C}+1}| + |E_{\le k}|}.\]

Processing clean elements takes $O\left(f |E|\right)$ time; then we spend $O\left(|D|\right) = O\left(|E_{\le k}|\right)$ time on processing dirty elements. 
The time the algorithm spends on the post-processing steps before invoking $\water$ is $O\left(f(|E| + |E'|) + |S| + f|\widehat{E}| + |\widehat{E}| + |\widehat{S}|\right) = O\left(f(|E| + |E'|)\right)$, since $\widehat{E} \subseteq E \cup E'$, $\widehat{S} \subseteq S$ and $|S| \le f(|E| + |E'|)$. 
Finally, invoking $\water(\widehat{k}, \widehat{S}, \widehat{E})$ takes time $O\left(f|\widehat{E}| + \widehat{k}\right) = O\left(f|\widehat{E}| + \frac{|\widehat{E}|}{\epsilon} + \frac{1}{\epsilon^2} + \frac{\log C}{\epsilon}\right) = O\left(f|\widehat{E}| + \frac{1}{\epsilon^2} + \frac{\log C}{\epsilon}\right)$, according to \Cref{waterfill} and since $f > \frac{\log C}{\epsilon}$ and $\widehat{k} \le \ceil{\log_{1 + \epsilon} \frac{2C \cdot |\widehat{E}|}{\epsilon}}$.

Therefore, the total runtime cost is \[O\brac{\frac{1}{\epsilon^2} + \frac{\log C}{\epsilon} + \left|T_{\leq k}\setminus T_{\leq \ceil{\log_{1+\epsilon}C}+1}\right| + |E_{\le k}| + f(|E| + |E'|)}.\]

According to \Cref{down}, the runtime cost $O\brac{|T_{\leq k}\setminus T_{\leq \ceil{\log_{1+\epsilon}C}+1}|}$ can be charged to the elimination of down potential $\Phi_{\down}(T_{\leq k})$, since the first term is at least:
\[\begin{aligned}
	\sum_{s\in T_{\leq k}}\max\{\lev(s) - \bs(s)+1, 0 \}&\geq \sum_{s\in T_{\leq k}\setminus T_{\ceil{\log_{1+\epsilon}C}+1}}\max\{\lev(s) - \bs(s)+1, 0 \}\\
	&\geq \sum_{s\in T_{\leq k}\setminus T_{\ceil{\log_{1+\epsilon}C}+1}}1 = \left|T_{\leq k}\setminus T_{\leq \ceil{\log_{1+\epsilon}C}+1}\right|
\end{aligned}\]

Observe that every element in $E \cup E'$ either was active before the call to $\reset(k)$, or was passive but became active during it. 
The runtime cost $O(f)$ induced by each active element can be charged to the elimination of $\Phi_{\down}(T_{\leq k})$, while the cost induced by passive elements can be charged to the clearance of their passive potentials. Using the same argument, we can cover the runtime cost $O(1)$ spent on processing each dirty element, for which we do not call $\procdet(e)$ or $\procrand(e)$ (because it became active during a previous iteration). For other dirty elements, the $O(1)$ runtime cost can be transferred to the respective call to $\procdet(e)$  or $\procrand(e)$.

The runtime cost $O(\frac{1}{\epsilon^2} + \frac{\log C}{\epsilon})$ can be charged to the elimination of $\Phi_\clean(T_{\le k})$. Since $\phi_{\le k} > 0$, then $\phi(s) > 0$ for at least one set $s$ with $\lev(s) \le k$. Therefore $\Phi_\clean(T_{\le k})$ has decreased by at least $\frac{1}{\epsilon^2} + \frac{\log C}{\epsilon}$.

In the meantime, since $\water$ might decrease the levels of the sets in $\widehat{S}$, we might have increased the lift potentials of $s\in \widehat{S}$ by at most $\max\{\lev^\old(s)-\bs(s)+1, 0\}$. Fortunately, such potential increases can also be paid for by the first term of $\Phi_{\down}(T_{\leq k})$ from \Cref{down}.
As for $\Phi_\up$, it does not increase during that steps, since for every set $s$ for which $\lev(s)$ or $\wts(s)$ were changed, either $s \in \widehat{S}$, or we have $\lev(s) = 0$ as a result. In the former case, we have $\wts(s) < c_s$ after the call to $\water(\widehat{k}, \widehat{S}, \widehat{E})$ by \Cref{waterfill}. In the latter case, it was either because $s \in S$ and slack, or because it was affected by the implicit zeroing, in which case we have $\wts(s) < c_s$, since \Cref{hp-below} holds after processing dirty elements.

It remains to analyze the runtime cost of the calls to $\procdet(e)$ and $\procrand(e)$. We will charge that costs to the insertion of $e$.

\paragraph{Deterministic rebuilding.} 
The runtime cost of $\procdet(e)$ is $O(f)$. As for the potential increase, notice that only sets below level $k + 1$ can increase their level. Any such set $s$ has $\phi(s) = 0$ by \Cref{hp-below}. Therefore, $\Phi_\down$ and $\Phi_\clean$ do not increase. 
$\Phi_\up$ does not increase as well. Observe that if there is a tight set $s \ni e$, then no weight of a set has changed, and if $\lev(s)$ is changed, then $\wts(s) < c_s$ by \Cref{hp-below}. Otherwise, if all sets are slack, then the call to $\decrease(e)$ could increase the weights of the sets; however, it makes sure that $\wts(s) < c_s$ for all $s \ni e$ as a result. $\Phi_\lift$ does not increase, since sets can only increase their levels.

Next, we will show that the gap $\ilev(e) - \zlev(e)$ decreases by at least one each time we call $\procrand(e)$.

\begin{claim}\label{det-gap-decr}
    The call to $\procdet(e)$ decreases the gap $\ilev(e) - \zlev(e)$ by at least one, and the size of the gap  becomes at most $\ilev^\old(e) - k - 1$.
\end{claim}
\begin{proof}
    By \Cref{iter-zlev<=k}, we have $\zlev(e) = \zlev^\old(e) \le k$, and $\ilev(e) = \ilev^\old(e)$ before the call. Observe that during the call, $\zlev(e)$ becomes at least $k + 1$, and $\ilev(e)$ does not increase. Therefore, the gap $\zlev(e) - \ilev(e)$ decreases by at least one and becomes at most  $\ilev^\old(e) - k - 1$.
\end{proof}

Notice that the gap $\ilev(e) - \zlev(e)$ is affected only by calls to $\procdet(e)$. This follows from \Cref{obs:clean-passive,obs:post-processing} and \Cref{handle-local}. Observe that calls to $\del$ do not affect the gap either. Calls to $\ins$ do not affect the gap by \Cref{obs:fix-main}(2). By \Cref{det-gap-decr}, the gap decreases by at least one after each call to $\procdet(e)$. Since after $e$ is inserted, we have $\ilev(e) \le \zlev(e) + \left\lceil\log_{1+\epsilon}\max\{f, \frac{2C}{\epsilon} \}\right\rceil$, there are at most $\left\lceil\log_{1+\epsilon}\max\{f, \frac{2C}{\epsilon} \}\right\rceil$ instances of $\procdet(e)$. As we have shown, the amortized runtime cost of one instance of $\procdet(e)$ is $O(f)$; therefore, the total runtime spent on the calls to $\procdet(e)$ is \[O\brac{\frac{f \log f}{\epsilon} + \frac{f \log C}{\epsilon} + \frac{f}{\epsilon^2}}.\]

\paragraph{Randomized rebuilding.} 
Consider the call to $\procrand(e)$. If the random sampling step succeeds, then the runtime cost is $O(f / \clog_{1 + \epsilon}^{(\eta)} f)$; otherwise, it is $O(f)$ (excluding the call to $\fix$). Notice that if the random sampling step fails, and we enter the cases for $|\widehat{F}| = 0$ or $|\widehat{F}| > \left(\clog^{(\eta)}_{1+\epsilon}f\right)^2$, the potential does not increase, by the same argument as for the deterministic version. The same holds for the case, when the random sampling succeeds. Next, consider the case when the random sampling fails, and we enter the case where we call $\fix$. By the same argument, the steps before the call do not increase the potential. We bound the amortized runtime cost of calling $\fix(e, l)$ by the following lemma.
\begin{lemma}\label{lm:fix-pot-increase}
    The amortized runtime cost of the call to $\fix(e, l)$ during $\procrand(e)$ is $O\left(\frac{f}{\epsilon^2} + \frac{f \log C}{\epsilon}\right)$ if $\eta < 3$, and $O\left(\frac{f}{\epsilon^2} + \frac{\log C}{\epsilon}\log^2 f\right)$ otherwise.
\end{lemma}
\begin{proof}
    Consider $F$ defined in the call to $\fix(e, l)$. By  definition, $F$ is a collection of all sets $s\ni e$ such that $\wts(s) -\wts(e) + (1+\epsilon)^{-\zlev(e) -d}> c_{s}$. Since $\zlev(e)$ becomes at least $k + 1$ right before the call to $\fix(e, l)$, we have $(1 + \epsilon)^{-\zlev(e) - \widehat{d}} \le \delta$, and hence $F\subseteq \widehat{F}$. Therefore, according to \Cref{fix-runtime}, the amortized runtime cost of the call to $\fix$ can be bounded by
    \begin{equation}\label{eq:fix-pot-bound}
        3|\widehat{F}|\cdot \left(\frac{3f}{\epsilon^3} + \frac{f\log C}{\epsilon^2}\right)\cdot (1+\epsilon)^{-\widehat{d} + 1} + |\widehat{F}| \cdot \left(\frac{1}{\epsilon^2} + \frac{\log C}{\epsilon}\right) + O(f)
    \end{equation}
    Since $|\widehat{F}| \leq \left(\clog^{(\eta)}_{1+\epsilon}f\right)^2$ and $\widehat{d} = \ceil{\log_{1+\epsilon}\max\left\{\brac{\clog_{1+\epsilon}^{(\eta)}f}^{4}, \brac{\frac{2C}{\epsilon}}^{2} \right\}}$, we can bound the first term of \Cref{eq:fix-pot-bound}:
    \[3|\widehat{F}|\cdot \left(\frac{3f}{\epsilon^3} + \frac{f\log C}{\epsilon^2}\right)\cdot (1+\epsilon)^{-\widehat{d} + 1} \leq 6\left(\frac{3f}{\epsilon^3} + \frac{f\log C}{\epsilon^2}\right) / \frac{2C}{\epsilon} < 12f/\epsilon^2\]
    By \Cref{lm:iterated-log}, $|\widehat{F}| \leq \left(\epsilon/5 \cdot5\log_{1+\epsilon}f\right)^2 = O(\log^2 f)$ when $\eta \ge 3$. Therefore, we can bound the second term of \Cref{eq:fix-pot-bound}:
    \[|\widehat{F}| \cdot \left(\frac{1}{\epsilon^2} + \frac{\log C}{\epsilon}\right) = O(\log^2 f) \cdot \left(\frac{1}{\epsilon^2} + \frac{\log C}{\epsilon}\right) = O\left(\frac{f}{\epsilon^2}\right) + O\left(\frac{\log C}{\epsilon}\log^2 f\right)\]
    When $\eta < 3$, we can use a trivial bound $|\widehat{F}| \le f$, and bound the second term:
    \[|\widehat{F}| \cdot \left(\frac{1}{\epsilon^2} + \frac{\log C}{\epsilon}\right) \le \frac{f}{\epsilon^2} + \frac{f \log C}{\epsilon}\]
\end{proof}

For an element $e$, consider all the instances of $\procrand(e)$. We will again argue that the gap $\ilev(e) - \zlev(e)$ decreases after each such call. By the same argument as for the deterministic version, the gap is affected by these calls only.

First, we bound the total runtime of the instances of $\procrand(e)$ where we fall back to $\procdet(e)$. Consider the first such an instance of $\procrand(e)$. It was called from within $\reset(k)$ for some $k$. Observe that we call $\procdet(e)$ whenever $\ilev(e) - k - 1 \le \frac{200}{\epsilon^2}$ or $\ilev(e) - \zlev(e) \le 1 + 2\log_{1 + \epsilon} \frac{2C}{\epsilon}$. If the first inequality holds, then, by \Cref{det-gap-decr}, the gap $\ilev(e) - \zlev(e)$ becomes at most $\ilev(e) - k - 1 \le \frac{200}{\epsilon^2}$ as a result. The gap $\ilev(e) - \zlev(e)$ strictly decreases after each call to $\procdet(e)$, and steps outside it do not increase the gap. Therefore, there are at most $O(\max \{1 + 2\log_{1 + \epsilon} \frac{2C}{\epsilon}, 1 + \frac{200}{\epsilon^2}\}) = O(\frac{\log C}{\epsilon} + \frac{1}{\epsilon^2})$ such instances of $\procrand(e)$, and each such an instance takes time $O(f)$. Therefore, the total time spent on such instances is $O(\frac{f}{\epsilon^2} + \frac{f\log C}{\epsilon})$.

Next, consider all the instances of $\procrand(e)$ where we do not fall back to $\procdet(e)$ and group them by the value of $\eta$ during them.

\begin{claim}\label{procrand-decr-weak}
    Each call to $\procrand(e)$ decreases the gap $\ilev(e) - \zlev(e)$ by at least one, and the size of the gap becomes at most $\ilev^\old(e) - k - 1$.
\end{claim}
\begin{proof}
    If we fall back to $\procdet(e)$, then this holds by \Cref{det-gap-decr}. Otherwise, by \Cref{iter-zlev<=k}, we have $\zlev(e) = \zlev^\old(e) \le k$, and $\ilev(e) = \ilev^\old(e)$ before the call. If the random sampling step succeeds, then $\zlev(e)$ becomes at least $k + 1$, and $\ilev(e)$ remains unchanged. If it fails, then we branch into three cases. If the first and the last case, $\zlev(e)$ becomes at least $k + 1$ and $\ilev(e)$ does not increase. In the second case, $\zlev(e)$ becomes at least $k + 1$ as well, and then we call to $\fix(e, l)$, where $l \le \ilev(e)$, which preserves the gap by \Cref{obs:fix-main}(1).
\end{proof}

\begin{claim}\label{reset-runtime}
    For any element $e$, the total amortized runtime spent on the instances of $\procrand(e)$ where we do not fall back to $\procdet(e)$ with the same value of $\eta$ is bounded in expectation by $O\left(\frac{f}{\epsilon^2} + \frac{f \log C}{\epsilon}\right)$ if $\eta < 3$, and $O\left(\frac{f}{\epsilon^2} + \frac{\log C}{\epsilon}\log^2 f\right)$ otherwise.
\end{claim}
\begin{proof}
    First, notice that for a fixed value of $\eta$, there are at most $\clog_{1 + \epsilon}^{(\eta)} f$ instances of $\procrand(e)$. Indeed, consider the first among them. By the definition of $\eta$, at the beginning of it, we have $\ilev(e) - k - 1 \le \clog_{1 + \epsilon}^{(\eta)} f$. By \Cref{procrand-decr-weak}, the gap $\ilev(e) - \zlev(e)$ becomes at most $\ilev(e) - k - 1$ after it.

    The runtime cost of the random sampling step is $O(f / \clog_{1 + \epsilon}^{(\eta)} f)$, so the total runtime spent on it for all such instances is $O(f)$. For the rest, assume that the random sampling step fails. In that case, we spend time $O(f)$ on computing $\widehat{F}$. Next, the analysis splits into cases, depending on the size of $\widehat{F}$.

    If $0 < |\widehat{F}| \le \brac{\clog_{1+\epsilon}^{(\eta)}f}^2$, then after the call to $\fix(e, l)$, the gap $\ilev(e) - \zlev(e)$ is at most 
    \begin{equation}\label{eq:gap-decr}
        \ceil{\log_{1+\epsilon}\max\left\{\brac{\clog_{1+\epsilon}^{(\eta)}f}^4, \brac{\frac{2C}{\epsilon}}^2 \right\}}\leq \max\left\{\clog_{1+\epsilon}^{(\eta+1)}f, 1+2\log_{1+\epsilon}\frac{2C}{\epsilon}\right\}
    \end{equation}
    In that case, we use \Cref{lm:fix-pot-increase} to bound the runtime cost of $\fix(e, l)$.

    If $\widehat{F} = \emptyset$, then we call $\decrease(e)$. Due to that call, either $e$ becomes active, or $\ilev(e)$ becomes at most $\ceil{-\log_{1 + \epsilon} \delta}$. 
    Recall that 
    \[\delta = \min\left\{\left(\frac{1}{\clog^{(\eta)}_{1+\epsilon}f}\right)^4, \brac{\frac{\epsilon}{2C}}^2\right\}\cdot (1+\epsilon)^{-k-1}.\]
    Observe that after the call, we have $\zlev(e) \ge k + 1$. Thus, we get the same bound on the gap as in \Cref{eq:gap-decr}.

    Therefore, if $|\widehat{F}| \le \brac{\clog_{1+\epsilon}^{(\eta)}f}^2$, then either 
    the gap $\ilev(e) - \zlev(e)$ becomes at most $\clog_{1 + \epsilon}^{(\eta + 1)} f$, and hence in the next instance the value of $\eta$ will be strictly larger, or the gap becomes at most $1+2\log_{1+\epsilon}\frac{2C}{\epsilon}$, in which case it is the last instance of $\procrand(e)$ where we do not fall back to $\procdet(e)$.

    Otherwise, if $|\widehat{F}| > \brac{\clog_{1+\epsilon}^{(\eta)}f}^2$, the algorithm spends 
    time $O(f)$. However, the probability that the random sampling step fails would be at most
    \[\left(1 - |\widehat{F}|/f\right)^{50\cdot\ceil{f / \clog_{1+\epsilon}^{(\eta)}f}} ~\leq~ \brac{\frac{1}{e}}^{10\clog_{1+\epsilon}^{(\eta)}f} ~\leq~  \frac{1}{\clog_{1+\epsilon}^{(\eta-1)}f} ~\le~ \frac{1}{\clog_{1+\epsilon}^{(\eta)}f}.\]
    Therefore, the expected time cost of such a call is $O(f / \clog_{1+\epsilon}^{(\eta)}f)$. 

    As we have shown in the beginning of the proof, there are at most $\clog_{1 + \epsilon}^{(\eta)} f$ instances with such a value of $\eta$. Hence the total expected runtime is
    \[O\left(\frac{f}{\clog^{(\eta)}_{1+\epsilon}f}\cdot \clog^{(\eta)}_{1+\epsilon}f\right) ~=~ O(f).\]
\end{proof}

For any passive element $e$, we have $\ilev(e) \le \zlev(e) + \ceil{\log_{1 + \epsilon} \max \{f, \frac{2C}{\epsilon} \}}$. Therefore, the gap $\ilev(e) - \zlev(e)$ is bounded by $\ceil{\log_{1 + \epsilon} f}$, since we assume $f > \frac{2C}{\epsilon}$ for the randomized version. Therefore, the maximum possible value of $\eta$ is bounded by $\clog_{1 + \epsilon}^* (f)$.
By \Cref{lm:iterated-log}, $\clog_{1 + \epsilon}^{(3)} (y) \le \epsilon / 5\cdot 5 \log_{1 + \epsilon} y \le 2 \ln y$ for large enough $y$; otherwise, if $y$ is small, then $\clog_{1 + \epsilon}^* (y) \le 3$. Thus $\clog_{1 + \epsilon}^* (f) = O(\log^* f)$, and hence the total runtime the algorithm spends on instances of $\procrand(e)$ is \[O\left(\frac{f}{\epsilon^2} \log^* f + \frac{\log C}{\epsilon} \log^2 f \log^* f + \frac{f}{\epsilon^2} + \frac{f \log C}{\epsilon}\right) = O\left(\frac{f}{\epsilon^2} \log^* f + \frac{f \log C}{\epsilon}\right).\]

We conclude our analysis with the following theorem.
\begin{theorem}
    In the case of the randomized algorithm, the amortized expected total time spent for each passive element $e$ on the calls to $\procrand(e)$ is $O\left(\frac{f}{\epsilon^2}\log^*f + \frac{f \log C}{\epsilon}\right)$.
\end{theorem}

\subsection{Insertion}

The algorithm spends $O(f)$ runtime before it enters the branching. Next, let us consider the branches separately.

\begin{itemize}[leftmargin=*]
    \item 
    In the case $F = \emptyset$, the algorithm spends $O(f)$ time on finding $h$ and $O(f)$ time on the remaining steps. For the potential increase, notice that we have $\wts(s) < c_s$ for all $s \ni e$ after the insertion of $e$, and hence $\Phi_{\up}(s) = 0$. 
    $\Phi_{\lift}(s)$ remains unchanged, since the levels of the sets remain the same. $\Phi_{\down}(s)$ and $\Phi_\clean(s)$ remain unchanged as well, since $\phi(s)$ does not change.
    \item In the case $F \ne \emptyset$, we can apply \Cref{fix-runtime} with the trivial bound $|F| \le f$ and $d = \ilev(e) - \zlev(e) \ge \log_{1 + \epsilon} f$. 
\end{itemize}

In both cases, since we have added a new element $e$ to the system, $\Phi(e)$ increases by at most $f$. Since we charge the runtime costs associated to the calls to $\procdet(e)$ and $\procrand(e)$ to $\ins(e)$, we obtain the following theorem.

\begin{theorem}
    The amortized runtime cost of $\ins(e)$ is $O\left(\frac{f \log f}{\epsilon} + \frac{f}{\epsilon^3} + \frac{f \log C}{\epsilon^2}\right)$ for the deterministic algorithm, and $O\left(\frac{f}{\epsilon^2}\log^*f + \frac{f}{\epsilon^3} + \frac{f \log C}{\epsilon^2}\right)$ for the randomized algorithm.
\end{theorem}

\subsection{Total runtime}
To conclude the proof of \Cref{main-result}, let $\lambda\in \{\randtime, \dettime\}$ be the upper bound on the amortized update time, where $\randtime = \Theta\left(\frac{f}{\epsilon^2}\log^*f  + \frac{f}{\epsilon^3} + \frac{f}{\epsilon^2}\log C\right)$, and $\dettime =  \Theta\left(\frac{1}{\epsilon}f\log f + \frac{f}{\epsilon^3} + \frac{f\log C}{\epsilon^2}\right)$ for the deterministic algorithm. Let $\Gamma$ be the total number of element updates. Without loss of generality, we can assume that at the end of the update sequence, we have $\univ=\emptyset$; otherwise, we can add $\Gamma$ artificial deletions which does not change the asymptotic runtime bound.
\begin{itemize}[leftmargin=*]
    \item \textbf{Preprocessing.} Initially $\univ = \emptyset$, and all sets are slack and on level $0$.
    
    If the algorithm is deterministic, then for each $0\leq i\leq L$, initialize pointers to (currently empty) sets $S_i, T_i, E_i, A_i(s), P_i(s)$ and store them in a random-accessible array. Add each set to $S_0$. This takes $O(\frac{1}{\epsilon}m\log(Cn))$ time and space.

    If the algorithm is randomized, then initialize randomized dynamic hash tables to store pointers to (currently empty) sets $S_i, T_i, E_i, A_i(s), P_i(s)$ \cite{dietzfelbinger1994dynamic}, and add each set to $S_0$. This takes time $O(m)$.

    \item \textbf{Updates.} As we have proved in previous subsections, for each update, the amortized update time is bounded as $\Delta\Phi + \text{runtime}\leq \lambda$.
\end{itemize}

Let $\Phi^{\text{init}}$ be the total potential at the beginning, and let $\Phi^{\text{end}}$ be the total potential at the end. Taking the summation of the preprocessing procedure and all updates, the total update time is bounded asymptotically by $\Gamma\cdot \lambda + \Phi^{\text{init}} - \Phi^{\text{end}}$. Since $\univ$ is empty both at the beginning and at the end, we have $\Phi^{\text{init}} =\Phi^{\text{end}}$, which finalizes the proof.

\clearpage
\bibliographystyle{alpha}
\bibliography{ref}

\newcommand{\etalchar}[1]{$^{#1}$}
\begin{thebibliography}{DKM{\etalchar{+}}94}

\bibitem[AAG{\etalchar{+}}19]{abboud2019dynamic}
Amir Abboud, Raghavendra Addanki, Fabrizio Grandoni, Debmalya Panigrahi, and
  Barna Saha.
\newblock Dynamic set cover: improved algorithms and lower bounds.
\newblock In {\em Proceedings of the 51st Annual ACM SIGACT Symposium on Theory
  of Computing}, pages 114--125, 2019.

\bibitem[AS21]{assadi2021fully}
Sepehr Assadi and Shay Solomon.
\newblock Fully dynamic set cover via hypergraph maximal matching: An optimal
  approximation through a local approach.
\newblock 204:8:1--8:18, 2021.

\bibitem[BCH17]{bhattacharya2017deterministic}
Sayan Bhattacharya, Deeparnab Chakrabarty, and Monika Henzinger.
\newblock {Deterministic fully dynamic approximate vertex cover and fractional
  matching in $O(1)$ amortized update time}.
\newblock In {\em Integer Programming and Combinatorial Optimization: 19th
  International Conference, IPCO 2017, Waterloo, ON, Canada, June 26-28, 2017,
  Proceedings}, pages 86--98. Springer, 2017.

\bibitem[BCPS23]{BCPS23}
Sayan Bhattacharya, Mart{\'{\i}}n Costa, Nadav Panski, and Shay Solomon.
\newblock Nibbling at long cycles: Dynamic (and static) edge coloring in
  optimal time.
\newblock {\em CoRR (to appear at SODA'24)}, abs/2311.03267, 2023.

\bibitem[BGK{\etalchar{+}}22]{BhattacharyaGKL22}
Sayan Bhattacharya, Fabrizio Grandoni, Janardhan Kulkarni, Quanquan~C. Liu, and
  Shay Solomon.
\newblock Fully dynamic ({\(\Delta\)} +1)-coloring in \emph{O}(1) update time.
\newblock {\em {ACM} Trans. Algorithms}, 18(2):10:1--10:25, 2022.

\bibitem[BGM17]{BhattacharyaGM17}
Sayan Bhattacharya, Manoj Gupta, and Divyarthi Mohan.
\newblock Improved algorithm for dynamic b-matching.
\newblock In {\em 25th Annual European Symposium on Algorithms (ESA)},
  volume~87 of {\em LIPIcs}, pages 15:1--15:13, 2017.

\bibitem[BHI15]{bhattacharya2015design}
Sayan Bhattacharya, Monika Henzinger, and Giuseppe~F Italiano.
\newblock Design of dynamic algorithms via primal-dual method.
\newblock In {\em Automata, Languages, and Programming: 42nd International
  Colloquium, ICALP 2015, Kyoto, Japan, July 6-10, 2015, Proceedings, Part I},
  pages 206--218. Springer, 2015.

\bibitem[BHN19]{bhattacharya2019new}
Sayan Bhattacharya, Monika Henzinger, and Danupon Nanongkai.
\newblock A new deterministic algorithm for dynamic set cover.
\newblock In {\em 2019 IEEE 60th Annual Symposium on Foundations of Computer
  Science (FOCS)}, pages 406--423. IEEE, 2019.

\bibitem[BHNW21]{bhattacharya2021dynamic}
Sayan Bhattacharya, Monika Henzinger, Danupon Nanongkai, and Xiaowei Wu.
\newblock Dynamic set cover: Improved amortized and worst-case update time.
\newblock In {\em Proceedings of the 2021 ACM-SIAM Symposium on Discrete
  Algorithms (SODA)}, pages 2537--2549. SIAM, 2021.

\bibitem[BK19]{bhattacharya2019deterministically}
Sayan Bhattacharya and Janardhan Kulkarni.
\newblock {Deterministically Maintaining a $(2+\epsilon)$-Approximate Minimum
  Vertex Cover in $O(1/\epsilon^2)$ Amortized Update Time}.
\newblock In {\em Proceedings of the Thirtieth Annual ACM-SIAM Symposium on
  Discrete Algorithms}, pages 1872--1885. SIAM, 2019.

\bibitem[DKM{\etalchar{+}}94]{dietzfelbinger1994dynamic}
Martin Dietzfelbinger, Anna Karlin, Kurt Mehlhorn, Friedhelm Meyer Auf
  Der~Heide, Hans Rohnert, and Robert~E Tarjan.
\newblock Dynamic perfect hashing: Upper and lower bounds.
\newblock {\em SIAM Journal on Computing}, 23(4):738--761, 1994.

\bibitem[DS14]{dinur2014analytical}
Irit Dinur and David Steurer.
\newblock Analytical approach to parallel repetition.
\newblock In {\em Proceedings of the forty-sixth annual ACM symposium on Theory
  of computing}, pages 624--633, 2014.

\bibitem[GKKP17]{gupta2017online}
Anupam Gupta, Ravishankar Krishnaswamy, Amit Kumar, and Debmalya Panigrahi.
\newblock Online and dynamic algorithms for set cover.
\newblock In {\em Proceedings of the 49th Annual ACM SIGACT Symposium on Theory
  of Computing}, pages 537--550, 2017.

\bibitem[HP20]{Henzinger020}
Monika Henzinger and Pan Peng.
\newblock Constant-time dynamic ({\(\Delta\)}+1)-coloring.
\newblock In {\em 37th International Symposium on Theoretical Aspects of
  Computer Science (STACS)}, volume 154 of {\em LIPIcs}, pages 53:1--53:18,
  2020.

\bibitem[KR08]{khot2008vertex}
Subhash Khot and Oded Regev.
\newblock Vertex cover might be hard to approximate to within $2-\varepsilon$.
\newblock {\em Journal of Computer and System Sciences}, 74(3):335--349, 2008.

\bibitem[PD06]{patrascu2006logarithmic}
Mihai Patrascu and Erik~D Demaine.
\newblock Logarithmic lower bounds in the cell-probe model.
\newblock {\em SIAM Journal on Computing}, 35(4):932--963, 2006.

\bibitem[PS16]{PelegS16}
David Peleg and Shay Solomon.
\newblock Dynamic $(1 + \epsilon)$-approximate matchings: A density-sensitive
  approach.
\newblock In {\em Proceedings of the Twenty-Seventh Annual {ACM-SIAM} Symposium
  on Discrete Algorithms (SODA)}, pages 712--729, 2016.

\bibitem[PT11]{puatracscu2011don}
Mihai P{\u{a}}tra{\c{s}}cu and Mikkel Thorup.
\newblock Don't rush into a union: take time to find your roots.
\newblock In {\em Proceedings of the forty-third annual ACM symposium on Theory
  of computing}, pages 559--568, 2011.

\bibitem[Sol16]{Solomon16}
Shay Solomon.
\newblock Fully dynamic maximal matching in constant update time.
\newblock In {\em {IEEE} 57th Annual Symposium on Foundations of Computer
  Science (FOCS)}, pages 325--334, 2016.

\bibitem[SU23]{solomon2023dynamic}
Shay Solomon and Amitai Uzrad.
\newblock {Dynamic $((1+\epsilon)\ln n)$-Approximation Algorithms for Minimum
  Set Cover and Dominating Set}.
\newblock In {\em Proceedings of the 55th Annual ACM Symposium on Theory of
  Computing}, pages 1187--1200, 2023.

\bibitem[SW18]{SolomonW18}
Shay Solomon and Nicole Wein.
\newblock Improved dynamic graph coloring.
\newblock In Yossi Azar, Hannah Bast, and Grzegorz Herman, editors, {\em 26th
  Annual European Symposium on Algorithms (ESA)}, volume 112 of {\em LIPIcs},
  pages 72:1--72:16, 2018.

\bibitem[WS11]{williamson2011design}
David~P Williamson and David~B Shmoys.
\newblock {\em The design of approximation algorithms}.
\newblock Cambridge university press, 2011.

\end{thebibliography}

\appendix

\begin{table}
  \centering
   \begin{tabular}{| l | p{13.5cm} |}
     \hline \\[-1em]
     \textbf{Notation} & \textbf{Definition} \\ \hline \\[-1em]
     $\wts(e)$ &  The weight of element $e \in \univ$. $\wts(e) = (1 + \epsilon)^{-\ilev(e)}$.\\ \hline \\[-1em]
     $\wts(s)$ &  The total weight of set $s \in \set$. $\wts(s) = \sum_{e \in s} \wts(e)$.\\ \hline \\[-1em]
     $\phi(s)$ &  The dead weight of set $s \in \set$.\\ \hline \\[-1em]
     $\phi$ &  The total dead weight. $\phi = \sum_{s\in \set}\phi(s)$.\\ \hline \\[-1em]
     $\phi_{i}$ &  The total dead weight of sets on level $i$. $\phi_i = \sum_{s, \lev(s) = i} \phi(s)$.\\ \hline \\[-1em]
     $\phi_{\leq i}$ &  The total dead weight of sets on level $i$ and below. $\phi_{\leq i} = \sum_{s, \lev(s) \le i} \phi(s)$.\\ \hline \\[-1em]
     $\wts^{*}(s)$ &  The composite weight of $s \in \set$. $\wts^{*}(s) = \wts(s) + \phi(s)$.\\ \hline \\[-1em]
     $\wts(s, i)$ & The weight of set $s \in \set$ at level $i$. It is the weight of $s$ if it were raised to level $i$. $\wts(s, i) = \sum_{e\in s}\min\left\{\wts(e), (1+\epsilon)^{-\max\{i, \max_{t\mid e\in t\neq s} \lev(t) \}}\right\}$.\\ \hline \\[-1em]
     $L$ &  The maximum level of a set, i.e. each set is assigned a level $\lev(s) \in [L]$. $L = \ceil{\log_{1+\epsilon}(Cn)}+1$.\\ \hline \\[-1em]
     Tight set &  A set $s \in \set$ is tight if $\wts^*(s) \geq \frac{c_s}{1+\epsilon}$.\\ \hline \\[-1em]
     Slack set &  A set $s \in \set$ which is not tight, i.e. $\wts^*(s) < \frac{c_s}{1+\epsilon}$.\\ \hline \\[-1em]
     $T$ &  The collection of all tight sets.\\ \hline \\[-1em]
     $T_{i}$ &  The collection of all tight sets at level $i$, i.e. a collection of $s \in T$ such that $\lev(s) = i$\\ \hline \\[-1em]
     $S_{i}$ &  The collection of all sets at level $i$, i.e. a collection of $s \in \set$ such that $\lev(s) = i$\\ \hline \\[-1em]
     $\lev(s)$ &  The level of set $s \in \set$. $\lev(s) \in [L]$.\\ \hline \\[-1em]
     $\lev(e)$ &  The level of element $e \in \univ$. $\lev(e) = \max_{s\ni e}\{\lev(s)\}$.\\ \hline \\[-1em]
     $\zlev(e)$ &  The lazy level of element $e \in \univ$. $\zlev(e) \leq \lev(e)$.\\ \hline \\[-1em]
     $\ilev(e)$ &  The intrinsic level of element $e \in \univ$. $\wts(e) = (1+\epsilon)^{-\ilev(e)}$ and $\lev(e)< \ilev(e)\leq \zlev(e) + \left\lceil\log_{1+\epsilon}\max\{f, \frac{2C}{\epsilon} \}\right\rceil$.\\ \hline \\[-1em]
     $\bs(s)$ &  The base level of set $s \in \set$. $\bs(s) = \floor{\log_{1+\epsilon}1/c_s}$.\\ \hline \\[-1em]
     Active element & If an element $e$ is active, then the value $\lev(e)$ is correctly maintained (i.e. $\ilev(e) = \zlev(e) = \lev(e)$), and $\wts(e) = (1 + \epsilon)^{-\lev(e)}$.\\ \hline \\[-1em]
     Passive element &  If an element $e$ is passive, then we have $\zlev(e) \leq \lev(e)$ and $\ilev(e) > \lev(e)$.\\ \hline \\[-1em]
     $A_{i}$ &  The set of all active element of level $i$ (i.e. $\lev(e) = i$)\\ \hline \\[-1em]
     $P_{i}$ &  The set of passive elements of intrinsic level $i$, (i.e. $\ilev(e) = i$).\\ \hline \\[-1em]
     $A_{i}(s)$, $P_{i}(s)$ &  $A_{i} \cap s$ and $P_{i} \cap s$ respectively.\\ \hline \\[-1em]
     $E_{i}$ & The set of elements $e$ such that $\zlev(e) = i$.\\ \hline \\[-1em]
     $A_{\leq i}, P_{\leq i}, E_{\leq i}, S_{\leq i}, T_{\leq i}$ & $A_{\leq i} = \bigcup_{k = 0}^{i} A_{k}$. The rest are defined similarly.\\ \hline \\[-1em]

     Dirty element & During a rebuild on a level $k$, an element $e \in E_{\leq k}$ is called dirty, if $e$ is passive and $\ilev(e) > k + 1$.\\ \hline \\[-1em]
     Clean element & During a rebuild on a level $k$, an element $e \in E_{\leq k}$ is clean if it is not dirty.\\
     \hline
  \end{tabular}
  \caption{Some of the notations used in this paper.}
  \label{table:notation-summary}
\end{table}

\end{document}